\documentclass[12pt]{article}
\usepackage{amsmath}
\usepackage[round]{natbib}
\usepackage{amssymb,color}
\usepackage{appendix}
\usepackage{amsthm}
\usepackage{hyperref}
\usepackage{enumerate}
\usepackage{dsfont}
\usepackage{mathtools} 
\mathtoolsset{showonlyrefs=true}  

\usepackage{caption}

\usepackage[margin=0.7in]{geometry}

\numberwithin{equation}{section}
\usepackage{comment}

\usepackage{titling}
\newtheorem{theorem}{Theorem}[section]
\newtheorem{corollary}[theorem]{Corollary}
\newtheorem{definition}{Definition}[section]
\newtheorem{assumption}{Assumption}[section]
\newtheorem{prop}[theorem]{Proposition}

\begin{document}

\author{Jaehyun Kim\thanks{jaehyun107@snu.ac.kr} and Hyungbin Park\thanks{hyungbin@snu.ac.kr, hyungbin2015@gmail.com}     \\ \\ \normalsize{Department of Mathematical Sciences} \\ 
		\normalsize{Seoul National University}\\
		\normalsize{1, Gwanak-ro, Gwanak-gu, Seoul, Republic of Korea} 
	}

	\title{Designing funding rates for perpetual futures in cryptocurrency markets}

\maketitle

\abstract{
In cryptocurrency markets, a key challenge for perpetual future issuers is maintaining alignment between the perpetual future price and target value. This study addresses this challenge by exploring the relationship between funding rates and perpetual future prices. Our results demonstrate that by appropriately designing funding rates, the perpetual future price can remain aligned with the target value. We develop replicating portfolios for perpetual futures, offering issuers an effective method to hedge their positions. Additionally, we provide path-dependent funding rates as a practical alternative and investigate the difference between the original and path-dependent funding rates. To achieve these results, our study employs path-dependent infinite-horizon BSDEs in conjunction with arbitrage pricing theory. Our main results are obtained by establishing the existence and uniqueness of solutions to these BSDEs and analyzing the large-time behavior of these solutions.}

\section{Introduction}\label{sec:Intro}

\subsection{Overview}

Perpetual futures are among the most widely traded derivative securities in cryptocurrency markets. Since their introduction by BitMEX in 2016, they have gained immense popularity and achieved significant trading volume.
Perpetual futures have accumulated over \$90 trillion in trading volume, surpassing the trading volumes of the underlying cryptocurrencies and now accounting for 93\% of the cryptocurrency future markets (\cite{ruan2024perpetual}).
A key aspect of perpetual futures is the funding mechanism, which operates effectively  within the continuous 24/7 trading environment.
This funding mechanism is more prevalent in cryptocurrency markets but is less frequently used in traditional markets due to higher transaction costs and stricter regulations. Despite the central role of funding mechanisms in perpetual futures,
there has been relatively little research on this
topic. Therefore, exploring its fundamentals is timely and essential to gain a better understanding of its operation.

A funding mechanism is designed to minimize price deviations between the future price and target value.
Cash flows are exchanged periodically between long and short positions  
to ensure that the price of the perpetual future remains aligned with the target value.  
This cash flow, known as the funding fee, is determined by the target value and current price of the perpetual future.
This study aims to 
construct theoretical fundamentals of 
perpetual futures by illuminating the relationship between funding rates and perpetual future prices. 
The structure of perpetual futures is similar to standard options, as the holder receives a payout in the form of a funding fee. This analogy prompts the adoption of the standard derivative pricing theory in this study.

However, compared to conventional derivatives, particularly standard European options, perpetual futures exhibit significant differences.
First, they have no expiration dates. This feature allows them to be bought or sold at any time without restrictions. Long positions in perpetual futures can be maintained for as long as desired without the concern of expiration. In contrast, standard options have a fixed expiry date, marking the last day the option contract remains valid.
Second, the payment stream of perpetual futures occurs periodically throughout the contract period, whereas the payments for standard options are made only at the end of the contract. This periodic payment structure is feasible because the cryptocurrency market operates 24/7.
Third, the payment stream of perpetual futures is influenced by both the underlying asset and perpetual future price, whereas standard options derive their payment solely from the underlying asset. These complexities make it challenging to apply the standard risk-neutral valuation to perpetual futures. This creates a significant distinction in pricing approaches between conventional options and perpetual futures.

This study has four contributions. First, we present a method for designing funding rates that ensures the perpetual future price remains anchored to the target values, including both tradable and non-tradable cases. This has posed a considerable challenge for perpetual future issuers due to the intricate nature of perpetual future structures. 
We address this problem using an arbitrage approach combined with the BSDE method.
Although the problem is conceptually straightforward, providing formal and rigorous proof is complex and challenging. Notably, our study is the first to derive unique prices for perpetual futures anchored to non-tradable target values.

Second, we investigate a path-dependent funding rate for practical implementation. In most exchanges, the funding fee is calculated as the average of values over the past 8 hours rather than relying on the current spot value, making the funding rate inherently path-dependent in practice.
While the instantaneous spot
funding rate ideally guarantees that the perpetual future price perfectly aligns with the target values,
the path-dependent version provides a more realistic and implementable approximation.
Our analysis shows  the perpetual future price is uniquely determined under  path-dependent funding rates.
Moreover, the  price derived from the path-dependent funding rate closely aligns with that obtained from the instantaneous spot funding rate, 
indicating that the practical version is an effective substitute. Notably, 
this specific form of path-dependent funding rate has not been studied before, making our study the first to explore this topic.

Third, this paper discusses the construction of replicating portfolios for perpetual futures. 
This is significant for two main reasons: practically, it provides issuers with a method to hedge perpetual futures, and theoretically, it demonstrates that the derived future price is arbitrage-free. 
While replicating portfolios are well-established for traditional derivatives with finite horizons, the absence of a terminal date in perpetual futures complicates their construction.
Our study shows that it is still possible to construct replicating portfolios even for derivatives without terminal dates.

Finally, this study employs a novel approach, an infinite-horizon BSDE method that provides new insights into perpetual futures. This methodology clarifies the
funding mechanism by revealing how funding rates influence perpetual future prices. 
The relationship between the driver and the solution of the infinite-horizon BSDE mirrors the connection between funding rates and perpetual futures prices. 
Mathematically, this is achieved by proving the existence and uniqueness of solutions to the infinite-horizon BSDEs. Because such types of BSDEs have not been previously studied, this study develops new theoretical results to establish conditions for the existence and uniqueness of solutions in this context.

This study adopts a path-dependent approach for two main reasons. First, it enables working with path-dependent market models that extend beyond traditional Markovian frameworks to accommodate non-Markovian market dynamics. While most of the existing literature focuses on Markovian models, our results hold under a more general framework that includes non-Markovian settings. Second, this approach allows for the analysis of path-dependent funding rates, modeled as integrals over the past 8 hours. These observations lead to the formulation of path-dependent BSDEs, for which standard Markovian BSDE techniques are not applicable. Our analysis primarily builds upon the path-dependent framework developed by \cite{ekren2014viscosity}, \cite{bally2016stochastic}, \cite{dupire2019functional} and \cite{viens2019martingale}.

The theoretical literature on the pricing of perpetual futures is scarce.
\cite{he2022fundamentals} considered proportional funding rates for perpetual futures anchored to a tradable asset price.
They derived arbitrage-free
prices for perpetual futures in frictionless markets and bounds in markets
with trading costs. 
\cite{angeris2023primer} analyzed perpetual future contracts in a continuous-time, arbitrage-free, and frictionless market, deriving model-free formulas for their funding rates, along with replication strategies, especially when asset prices are continuous and positive. They also extended these results to jump models, providing semi-robust expressions that depend on jump intensities and offering explicit replication strategies when the
volatility process is independent of
the underlying risky asset. 
\cite{ackerer2024perpetual}
derived an arbitrage-free price of various
perpetual contracts, including linear, inverse, and quantos futures.
The price is determined by
the risk-neutral expectation of the spot, sampled at a random time that reflects the
intensity of the price anchoring. 
\cite{dai2025arbitrage} investigated the tendency of perpetual futures prices to deviate from their underlying asset prices. They identified the clamping function in the funding mechanism as a key factor and derived model-free no-arbitrage bounds that hold even without transaction fees. 
Although the theoretical literature on perpetual futures is limited, several empirical studies have been conducted in this field. Refer to \cite{alexander2020bitmex}, \cite{christin2022crypto} and \cite{wang2025spot}.

\subsection{Outline}

A strategy for addressing this problem consists of several steps.
First, we introduce the concept of funding portfolios and their associated wealth processes, establishing a fundamental theoretical framework for constructing replicating portfolios in subsequent discussions. 
Consider a complete market comprising multiple assets, denoted as \( X = (X(s))_{s \ge 0} \), along with a money market account. Let \( F = (F(s))_{s \ge 0} \) be a given funding rate process. At this stage, we assume that \( F \) is a general stochastic process without imposing any specific structure.
An \( F \)-funding portfolio consists of holdings in multiple assets and a money market account, which incurs a funding fee \( F(s)\,ds \) over time. The corresponding wealth process, denoted as \( Y = (Y(s))_{s \ge 0} \), represents the evolution of the value of the
\( F \)-funding portfolio.

Second, we examine a specific type of funding rate that depends on both the wealth process  and  the underlying asset value. More precisely, consider a funding rate of the form \( F^\Phi(s) := \Phi(s, X_s, Y(s)) \), where \(\Phi\) is a specified functional.
For the meanings of  $X_s,Y_s$ and $X(s),Y(s)$, see Section \ref{sec:pre}. 
A fundamental question arises: does there exist a \(F^\Phi\)-funding portfolio corresponding to the specified functional \(\Phi\)? This problem is non-trivial, and to address this, we transform the problem into an infinite-horizon BSDE framework. We establish the existence and uniqueness of such \(F^\Phi\)-funding portfolios by analyzing the associated infinite-horizon BSDEs. Furthermore, we characterize the wealth process of the $F^\Phi$-funding portfolio through our BSDE analysis.

Third, we design funding rates
ensuring that
the perpetual future price remains aligned with the target value.
Imagine an issuer seeking a perpetual future price to be aligned with $\varphi(s,X_s)$ for a given function $\varphi.$
We verify that a specific form of funding rate $\Phi(s,X_s,Y(s))$
presented in 
\eqref{eqn:Phi} induces a funding portfolio whose wealth process coincides with the target value $\varphi(s,X_s)$. Thus, this funding portfolio is the replicating portfolio for the perpetual future and  its wealth process generates the desired prices. Consequently,
this provides issuers with a solution on how to design funding rates for perpetual futures.

Fourth, we investigate a path-dependent funding rate, which offers a practical alternative.
The funding rate $\Phi(s,X_s,Y(s))$ mentioned above provides the desired perpetual future prices; however, this is not practical.
Instead, we consider the funding rate of the form
\begin{align}
	\Phi^\delta(s,X_s,Y_s):=\frac{1}{\delta}\int_{s-\delta}^s\Phi(u,X_u,Y(u))\,du  
\end{align}
for $\delta>0$. In practice, \(\delta = \frac{1}{1095}\), corresponding to 8 hours, is commonly used.
We verify that there exists a $\Phi^\delta$-funding portfolio and derive 
the corresponding wealth process $Y^\delta$. Additionally, we estimate the difference between $Y^\delta$ derived from this path-dependent funding rate and $Y$ derived from the original funding rate $\Phi.$

To illustrate the core idea of this study, we examine the following simple example.
Consider an uncorrelated $m$-dimensional Black-Scholes stock model with zero short rate. Under the risk-neutral measure, the stock price process $X = (X_1(s), \ldots, X_m(s))_{s \ge 0}$ evolves according to
\begin{equation}
	\begin{aligned}
		dX_i(s) = \sigma_i X_i(s)\, dB_i(s), \quad i = 1, 2, \ldots, m, 
	\end{aligned}
\end{equation}
where $\sigma_i > 0$ for all $i = 1, 2, \ldots, m$, and $(B_1(s), \ldots, B_m(s))_{s \ge 0}$ is an $m$-dimensional Brownian motion with uncorrelated components.
Suppose an issuer seeks to design funding rates for a perpetual future whose price process $Y$ remains aligned with the square of the first stock's price, i.e.,
$Y(s) = \varphi(X(s))$  for $s \ge 0,$
where $\varphi(x) = x_1^2$ for $x = (x_1, \ldots, x_m) \in \mathbb{R}^m$.
For the moment, assume that the funding rate depends only on the current stock prices. Specifically, we set
$F(s) = \Phi(X(s))
$ for some function $\Phi \in C^2(\mathbb{R}^m)$.
The wealth process from holding one perpetual future is
\begin{equation}\label{eqn:in}
	\begin{aligned}	
		&\quad Y(s)+\int_0^sF(u)\,du
		=\varphi(X(s))+\int_0^s\Phi(X(u))\,du=X_1^2(s)+\int_0^s\Phi(X(u))\,du\,,\;s\ge0\,.
\end{aligned}\end{equation}
By the standard  no-arbitrage principle, any wealth process must be a local martingale  under the risk-neutral measure. Applying this principle, we obtain $\Phi(x)=-\frac{1}{2}\sigma_1^2x_1^2\partial_{x_1x_1}\varphi=-\sigma_1^2x_1^2$.
Therefore, the appropriate funding rate for the perpetual future $X_1^2$ is  $\Phi(X(s))=-\sigma_1^2X_1^2(s)$ for $s\ge0$.

However, an important observation is that this funding rate does not uniquely determine the futures price.  For instance, processes such as 
$Y=2X_1+X_1^2$ and $Y=X_2+2X_m+X_1^2$ 
also serve as perpetual futures prices consistent with the same funding rate.
This implies that the given rate allows for multiple arbitrage-free prices.
To resolve this non-uniqueness, we introduce an additional term of the form $\ell(X_1^2(s)-Y(s))$ for sufficiently large $\ell>0$.
The modified funding rate 
\begin{equation}
	\begin{aligned}
		\Phi(X(s),Y(s))=\ell(X_1^2(s)-Y(s))-\sigma_1^2X_1^2(s)\,,\;s\ge0
	\end{aligned}
\end{equation}
enforces a unique perpetual futures price within a suitable class of admissible processes. Moreover, this unique price aligns with the target value $Y(s)=X_1^2(s)$ for $s\ge0.$ This modification is the core idea of the approach developed in this paper.   We later provide a rigorous justification for why the additional term guarantees the uniqueness of the perpetual futures price.

The remainder of this paper is organized as follows. Section \ref{sec:pre} introduces the basic notations for path-dependent SDEs and PDEs as preliminary concepts. In Section \ref{sec:design}, we analyze funding portfolios and their corresponding wealth processes. Section \ref{sec:design_FR} presents instantaneous spot funding rates as well as path-dependent funding rates. Section \ref{sec:appli} discusses several applications, and Section \ref{sec:con} concludes with a summary of the main findings.
The proofs of the main results are presented
in the appendices.

\section{Preliminary}
\label{sec:pre}

In this section, we introduce basic notations of path-dependent SDEs and PDEs as preliminaries. The reader may refer to \cite{ekren2014viscosity}, \cite{bally2016stochastic}, \cite{dupire2019functional} and \cite{viens2019martingale} for more details.
Throughout this study, let $(\Omega,\mathcal
{F}, \mathbb{P})$ be a complete probability space having a $m$-dimensional Brownian motion $W$.
The augmented $\sigma$-algebra generated by $W$ is denoted as $(\mathcal{F}_s)_{s\ge 0}$. 
Let $\mathbb{L}^0(\mathbb{R}^m)$ be the space of all progressively measurable processes taking values in $\mathbb{R}^m$. For $p>0$ and $T>0,$ we define 
\begin{align}
	L^p(\mathcal{F}_s;\mathbb{R}^m)&=\{\xi:\Omega\to \mathbb{R}^m\,|\, \xi \text{ is a  } \mathcal{F}_s\text{-measurable random variable and } \mathbb{E}[|\xi|^p]<\infty\}\,, \\  
	\mathbb{S}^p(0,T;\mathbb{R}^m)&=\Big\{X \in \mathbb{L}^0(\mathbb{R}^m)\,|\, X  \textnormal{ is  continuous   in time and } \mathbb{E}[|\!|X|\!|_T^p] < \infty \Big\}\,, \\
	\mathbb{H}^p(0,T;\mathbb{R}^m)&=\Big\{Z \in \mathbb{L}^0(\mathbb{R}^m)\,\Big|\, \mathbb{E}\Big[\Big(\int_0^T |Z(u)|^2\,du\Big)^{\frac{p}{2}}\Big]<\infty \Big\}\,,\\
	\mathbb{S}^p(0,\infty;\mathbb{R}^m)&=\cap_{T>0}\mathbb{S}^p(0,T;\mathbb{R}^m)\,, \\
	\mathbb{H}^p(0,\infty;\mathbb{R}^m)&=\cap_{T>0}\mathbb{H}^p(0,T;\mathbb{R}^m)\,.
\end{align}
The spaces $\mathbb{S}^p(s,T;\mathbb{R}^m), $ $\mathbb{H}^p(s,T;\mathbb{R}^m),$ $\mathbb{S}^p(s,\infty;\mathbb{R}^m)$ and $ \mathbb{H}^p(s,\infty;\mathbb{R}^m)$ are similarly defined for $s\in [0,T]$.

Let $\hat{\Lambda}:=\mathbb{D}([0,\infty),\mathbb{R}^m)$ be the space of all c\'adl\'ag functions from $[0,\infty)$ to $\mathbb{R}^m$. For ${\gamma}\in \hat{\Lambda}$, denote by ${\gamma}(s)$ the value of ${\gamma}$ at time $s$ and by ${\gamma}_s={\gamma}(s\wedge\cdot)$ the path of ${\gamma}$ stopped at time $s.$  
Define 
a seminorm $|\!|\cdot|\!|_T$ and norm $|\!|\!|\cdot|\!|\!|$ on 
$\hat{\Lambda}$ and a pseudometric $d$ on $[0,\infty)\times \hat{\Lambda}$ as  { 
	\begin{align}\label{norm}
		& |\!|\gamma|\!|_T=\sup\{|\gamma(s)|\,:\,s\in [0,T]\}\,,\;T\ge0\,,\\ 
		&|\!|\!|{\gamma}|\!|\!|=\sum_{n=1}^{\infty}\frac{1}{2^n}(|\!|\gamma|\!|_n\wedge 1)\,,\\
		&d((s,{\gamma}),(s',{\gamma}'))=|s-s'|+\sup_{r\in[0,s\vee s']}|\gamma(r\wedge s)-\gamma'(r\wedge s')|\,.
\end{align}}
We write as $|\!|\gamma|\!|\le |\!|\gamma'|\!|$
for $\gamma,\gamma'\in\hat{\Lambda}$ if
$|\!|\gamma|\!|_T\le |\!|\gamma'|\!|_T$
for all $T\ge0.$
For c\'adl\'ag processes $X$ and $X'$, the meanings of $X(s), X_s, |\!|X|\!|_T,|\!|X|\!|\le |\!|X'|\!|$ are straightforward.

A map $\varphi:[0,\infty)\times \hat{\Lambda}\to \mathbb{R}$ is called a non-anticipative functional if 
$\varphi(s,\gamma)=\varphi(s,\gamma_s)$ for all $(s,\gamma)\in [0,\infty)\times \hat{\Lambda}$.
A non-anticipative functional \( \varphi \) is said to have polynomial growth of order \( p\) at most if there exists a constant \( L > 0 \) such that
\(
|\varphi(s, \gamma)| \leq L ( 1 + |\!|\gamma|\!|_s^p ) \)
for all \((s, \gamma) \in [0, \infty) \times \hat{\Lambda}\). The constant \( p \) is referred to as the polynomial growth order of \( \varphi \).
We say a non-anticipative functional $\varphi$ is horizontally differentiable (vertically differentiable, respectively) at $(s,{\gamma})$ if the limit
\begin{align}
	\partial_s \varphi(s,{\gamma}):=\lim_{h\to 0^+}\frac{\varphi(s+h,{\gamma}_s)-\varphi(s,{\gamma}_s)}{h} 
\end{align} 
exists (if for $i=1,\cdots,m$, the limit
\begin{align}
	\partial_{i} \varphi(s,{\gamma}):=\lim_{h\to 0}\frac{\varphi(s,{\gamma}_s+he_i\boldsymbol{1}_{[s,\infty)})-\varphi(s,{\gamma}_s)}{h}
\end{align} exists  
where $(e_i)_{1\le i\le m}$ represents the standard basis of $\mathbb{R}^m$, respectively). We denote as $\partial_x \varphi(s,{\gamma})=(\partial_{i}\varphi(s,{\gamma}))_{1\le i\le m}$
and $\partial^2_{xx}\varphi(s,{\gamma})=(\partial_{i}(\partial_{j}\varphi)(s,{\gamma}))_{1\le i,j\le m}.$    
For $p\ge1$, we define the spaces 
 \begin{equation}\label{def:non-anti}
		\begin{aligned}
			&\mathcal{A}([0,\infty)\times \hat{\Lambda})=\{\varphi:[0,\infty)\times \hat{\Lambda}\to \mathbb{R}\,|\,\varphi\textnormal{ is  non-anticipative}\}\,,\\
			&C([0,\infty)\times \hat{\Lambda})
			=\{\varphi\in \mathcal{A}([0,\infty)\times \hat{\Lambda})\,|\, \varphi\textnormal{ is   continuous with respect to } d \}\,,   \\
			&C_{p}([0,\infty)\times \hat{\Lambda})
			=\{\varphi\in C([0,\infty)\times \hat{\Lambda})\,|\,  \varphi\textnormal{ has  polynomial growth of order } p\}\,,\\
			& C_{p}^{1,2}([0,\infty)\times \hat{\Lambda})
			=\{\varphi\in C([0,\infty)\times \hat{\Lambda})\,|\,\partial_s \varphi, \partial_x \varphi,   \partial^2_{xx}\varphi \textnormal{ exist and are in } C_p([0,\infty)\times \hat{\Lambda})\}\,. 
		\end{aligned}
\end{equation}

Let $\Lambda:= C([0,\infty);\mathbb{R}^m)$ be the space of all continuous functions from $[0,\infty)$
to $\mathbb{R}^m$.
Because $\Lambda\subset \hat{\Lambda}$ and $[0,\infty)\times \Lambda\subset [0,\infty)\times \hat{\Lambda}$, the seminorm $|\!|\cdot |\!|_T$, norm $|\!|\cdot |\!|$, and pseudometric $d$ defined in \eqref{norm} are inherited. 
It can be easily checked that $\Lambda$ and $[0,\infty)\times \Lambda$ are  closed subspaces of $\hat{\Lambda}$ and $[0,\infty)\times \hat{\Lambda}$, respectively.
Three spaces $\mathcal{A}([0,\infty)\times {\Lambda}),$
$C([0,\infty)\times {\Lambda}),$
$C_p([0,\infty)\times \Lambda)$ are defined 
similarly to \eqref{def:non-anti}
and
$ C^{1,2}_{p}([0,\infty)\times \Lambda)$ is defined 
as the space of all processes $\varphi\in 
\mathcal{A}([0,\infty)\times {\Lambda})$ such that there exists $\hat\varphi\in C_{p}^{1,2}([0,\infty)\times \hat{\Lambda})$ satisfying  $\varphi(s,\gamma)=\hat{\varphi}(s,\gamma)$ for all $(s,\gamma)\in [0,\infty)\times\Lambda.$ We denote $\partial_s \varphi:=\partial_s\hat{\varphi}$, $\partial_x \varphi:=\partial_x\hat{\varphi}$, $\partial^2_{xx} \varphi:=\partial^2_{xx}\hat{\varphi}$. By \cite{cont2013functional}, the derivatives $\partial_s \varphi$, $\partial_x \varphi$, $\partial^2_{xx} \varphi$ do not depend on the choice of $\hat{\varphi}$.

\section{Pricing perpetual futures}
\label{sec:design}

\subsection{Funding portfolios}\label{sec:F}

Consider a market with $m+1$ assets consisting of $m$ wealth processes, such as stocks, golds, cryptocurrencies and their linear combinations given as a solution to the path-dependent SDE
\begin{align}\label{eqn: path dep SDE}
	X(s)=x&+\int_0^s \mu(u,X_u)\,du+\int_0^s\sigma(u,X_u)\,dW(u)\, \text { for } s\ge 0
\end{align} 
for $x\in \mathbb{R}^m$ and non-anticipative functionals  $\mu:[0,\infty)\times \Lambda\to \mathbb{R}^m$, $\sigma:[0,\infty)\times \Lambda\to \mathbb{R}^{m\times m}$,   
and a
money market account
$    G=(e^{\int_0^sr(u,X_u)\,du})_{s\ge0}$ for a non-anticipative 
functional 
$r:[0,\infty)\times \Lambda \to\mathbb{R}.$
Our wealth process model encompasses negative value scenarios, including short positions in certain assets, linear combinations of multiple assets, and the event of negative oil prices (\cite{corbet2021volatility}).  
In numerous asset market models, the short rate is typically assumed to be constant rather than a functional of asset values. However, this study models the short rate as a functional dependent on asset values to generalize the market framework.

\begin{assumption}\label{hypo: SDE}
	The non-anticipative functionals $\mu$ and $\sigma$ satisfy the following conditions.
	\begin{enumerate}[(i)]
		\item There exists a constant $C_1>0$ such that $|\mu(s,0)|+|\sigma(s,0)|\le C_1$ for $s\in [0,\infty)$. 
		\item There exists constants $C_2,C_3>0$ such that
		\begin{align}
			&|\mu(s,\gamma)-\mu(s',\gamma')|\le C_2\, d((s,\gamma),(s',\gamma'))\,,\\
			&|\sigma(s,\gamma)-\sigma(s',\gamma')|\le C_3\, d((s,\gamma),(s',\gamma'))
		\end{align}
		for $(s,\gamma),(s'\gamma')\in [0,\infty)\times \Lambda$.
	\end{enumerate}
\end{assumption} 

\begin{assumption}\label{hypo: SDE_r}
	The short rate functional $r:[0,\infty)\times \Lambda \to\mathbb{R}$ is bounded. More precisely, there exists a constant $C_r>0$ such that 
	$|r(s,\gamma)|\le C_r$ for all $(s,\gamma)\in [0,\infty)\times  \Lambda$.
\end{assumption} 
The following theorem states the existence and uniqueness of solutions to the SDE \eqref{eqn: path dep SDE}.
For the proof, refer to \cite[Theorem 7, Chapter 5]{protter}.
\begin{theorem}
	Let Assumption \ref{hypo: SDE} hold and $p\ge1.$	
	Then, for any $x\in \mathbb{R}^m$, there exists a unique solution $X$ to the SDE \eqref{eqn: path dep SDE} in $\mathbb{S}^p(0,\infty;\mathbb{R}^m)$.
\end{theorem}

A portfolio is a $m+1$ dimensional progressively measurable process $h=(\phi^0(s),\phi(s))_{s\ge0}=(\phi^0(s),\cdots,\phi^m(s))_{s\ge0}$ with its value process
$
V^h(s):=\phi^0(s)G(s)+\phi(s)X(s)$ for $s\ge0. 
$

\begin{definition}\label{defi:funding_port}
	Let $F=(F(s))_{s\ge0}$ be a progressively measurable process.
	\begin{enumerate}[(i)]
		\item  We say a portfolio $h=(\phi^0(s),\phi(s))_{s\ge0}$ is a $F$-funding portfolio if $\phi^0$ and $F$ are locally integrable and 	$\phi$ is locally square-integrable with respect to the Lebesgue measure on $[0,\infty)$ almost surely and if
		\begin{align}
			dV^h(s)=\phi^0(s)\,dG(s)+\phi(s)\,dX(s)-F(s)\,ds\,.
		\end{align}
		\item For  $\rho\ge1$,  we say a  portfolio $h$ is $\rho$-admissible if there exists a constant $L>0$ such that $|\!|V^h|\!|_T\leq L(1+|\!|X|\!|_T^\rho) $ for all $T\ge0.$
	\end{enumerate}	
\end{definition}
\noindent This definition is similar to the consumption portfolio in the standard investment-consumption portfolio theory. In particular, when $F=0,$ this portfolio $h$ is self-financing.
The local integrability condition is a technical requirement that guarantees the well-definedness of Lebesgue integrals and  It\^o integrals.
The
$\rho$-admissibility indicates that the value process is bounded above by a positive constant multiple of the $\rho$-th power of the underlying wealth process.

We now describe fundamental concepts of perpetual future prices and funding rate functionals.

\begin{definition} 	Let $F=(F(s))_{s\ge0}$ be a progressively measurable process.
	\begin{enumerate}[(i)]
		\item A perpetual future
		with funding rate $F$ 
		is a financial instrument where the short position pays the funding fee 
		$F(s)\,ds$ to the long position continuously until the contract is terminated. Denote the price of this perpetual future as $Y=(Y(s))_{s\ge0}$.	 
		\item A perpetual future $Y$
		with funding rate $F$ is said to be replicable if there exists a $F$-funding portfolio $h$ such that $Y=V^h$.   
	\end{enumerate}
\end{definition}

\noindent In cases where the funding rate is negative, the short position pays $F(s)\,ds$ to the long position, indicating that the short position receives $-F(s)\,ds$ from the long position.

We can express the value process of  $F$-funding portfolio as an infinite-horizon BSDE for a given funding rate $F$.
Observe that a $F$-funding portfolio $h=(\phi^0(s),\phi(s))_{s\ge0}$ satisfies
\begin{equation}
	\begin{aligned}
		dV^h(s)
		&=\phi^0(s)\,dG(s)+\phi(s)\,dX(s)-F(s)\,ds\\
		&=(r(s,X_s)(V^h(s)-\phi(s)X(s))+\phi(s)\mu(s,X_s)-F(s))\,ds+\phi(s)\sigma(s,X_s)\,dW(s)\\
		&=(r(s,X_s)V^h(s) +Z(s)\theta(X_s)
		-F(s))\,ds+Z(s)\,dW(s)
	\end{aligned}
\end{equation}
where $\theta(s,X_s):=\sigma^{-1}(s,X_s)(\mu(s,X_s)-X(s)r(s,X_s))$ and $Z(s):=\phi(s)\sigma(s,X_s)$ for $s\ge0.$
Defining
$Y(s)=V^h(s)$,
we have
\begin{align}
	\label{eqn:F}
	Y(s)=Y_T-\int_s^T&(r(u,X_u)Y(u) +Z(u)\theta(u,X_u)
-F(u))\,du-\int_s^TZ(u)\,dW(u)
\end{align}
for $0\le s\le T<\infty$
in the infinite-horizon BSDE form.

We mainly work with funding rates that depend  on both the $m$ wealth processes $X$
and the perpetual future price $Y.$
Typical funding rates traded in the cryptocurrency market follow this form.

\begin{definition}
	\label{defi:fundingrate} If a funding rate is expressed as \( F(s) = \Phi(s, X_s, Y(s)),s\ge0 \) for a non-anticipative functional \(\Phi: [0, \infty) \times \Lambda \times \mathbb{R} \to \mathbb{R} \), then we refer to \(\Phi\) as the funding rate functional and denote the corresponding funding rate by \( F^\Phi \).
\end{definition}

The first objective is to verify that for any appropriate funding rate functional $\Phi$, there exists a unique
$F^\Phi$-funding portfolio.
If the funding rate functional $\Phi=\Phi(s,X_s)$ is independent of $Y$, proving the existence and uniqueness of $F^\Phi$-funding portfolios becomes relatively straightforward. However, if $\Phi$ depends on $Y$, this is not immediately clear; the formal and rigorous proof is complex and challenging.
We address this problem by formulating it as an infinite-horizon BSDE.
With a funding rate functional $\Phi:[0,\infty)\times \Lambda \times \mathbb{R} \to\mathbb{R}$, it follows that
\begin{equation} \label{BSDE_P}
	\begin{aligned}
		Y(s)=Y_T-\int_s^T&(r(u,X_u)Y(u) +Z(u)\theta(u,X_u)
	-\Phi(u,X_u,Y(u)))\,du-\int_s^TZ(u)\,dW(u)
	\end{aligned}
\end{equation} 
for $0\le s\le T<\infty$
in the infinite-horizon BSDE form.
Consequently, the existence and uniqueness of $F^\Phi$-funding portfolios is equivalent to the existence and uniqueness of solutions to the infinite-horizon BSDE above.

\subsection{Risk-neutral pricing BSDEs}

We now introduce risk-neutral measures and demonstrate how funding portfolios are expressed using the risk-neutral measure.

\begin{assumption}\label{sharpe}
	The matrix $\sigma$ is invertible and the
	local martingale
	\begin{align}
		(e^{-\int_0^s\theta(u,X_u)\,dW(u)-\frac{1}{2}\int_0^s|\theta(u,X_u)|^2\,du})_{s\ge0} 
	\end{align}
	is a martingale.
\end{assumption}

We define a risk-neutral measure by the Girsanov theorem under this assumption. For $s\ge0,$ let $\mathbb{Q}_s$ be a probability measure on $\mathcal{F}_s$ defined as
\begin{equation}
	\label{eqn:Girsa}\frac{d\mathbb{Q}_s}{d\mathbb{P}}=e^{-\int_0^s\theta(X_u)\,dW(u)-\frac{1}{2}\int_0^s\theta^2(X_u)\,du}\,.
\end{equation} 
We can extend this to the sigma-algebra $\mathcal{F}=\sigma(\cup F_s,s\ge0)$
by \cite[Section 5.3]{deuschel1989large}.
It is evident that $\mathbb{P}$ and $\mathbb{Q}_s$ are equivalent for all $s\ge0$ 
and $X^i/G$ is a martingale for all $i=1,2,\cdots,m.$
From the Girsanov theorem, the process
\begin{equation}
	B(s)=W(s)+\int_0^s \theta(u,X_u)\,du,\,s\ge0
\end{equation}
is a Brownian motion under the measure $\mathbb{Q}.$
Using this Brownian motion $B$, 
the $\mathbb{Q}$-dynamics of the wealth processes $X$ is expressed as
\begin{align}\label{SDE_X}
	X(s)=x&+\int_0^s r(u,X_u)X(u)\,du+\int_0^s\sigma(u,X_u)\,dB(u)\,,\;s\ge0\,.
\end{align}
and the infinite-horizon BSDE \eqref{eqn:F} becomes
\begin{align} \label{eqn:Q_F}
	Y(s)=Y_T&-\int_s^T(r(u,X_u)Y(u) 
	-F(u))\,du-\int_s^TZ(u)\,dB(u)
\end{align}
for $0\le s\le T<\infty$.
For convenience, we use the notation $\mathbb{E}$ for $\mathbb{E}^\mathbb{Q}$ throughout this study without ambiguity.

Analyzing the infinite-horizon BSDE \eqref{BSDE_P} becomes more straightforward when conducted under the risk-neutral measure. It follows that
\begin{equation} \label{eqn:BSDE}
	\begin{aligned}
		Y(s)
		=Y(T)-\int_s^T&(r(u,X_u)Y(u)-\Phi(u,X_u,Y(u)))\,du-\int_s^TZ(u)\,dB(u) 
	\end{aligned}
\end{equation} 
for $0\le s\le T<\infty$.
We refer to this as the risk-neutral pricing BSDE for the perpetual future.
As mentioned above, the existence and uniqueness of $F^\Phi$-funding portfolios is equivalent to the existence and uniqueness of solutions to this risk-neutral pricing BSDE.
Theorem \ref{thm: BSDE no delay} states that under Assumptions \ref{hypo: SDE} - \ref{hypo: BSDE no delay}, this BSDE has a unique solution $(Y,Z)=(Y(s),Z(s))_{s\ge0}$, thereby implying the existence and uniqueness of 
$F^\Phi$-funding portfolios.

We state several conditions on the driver to guarantee the existence and uniqueness of solutions 
to the risk-neutral pricing BSDE.
To precisely determine the constant $M_{\rho\vee 2}$ stated in Assumption \ref{hypo: BSDE no delay} \eqref{hypo mu}, we recall the BDG inequality.
For any $q\ge 1$, there exist positive constants $m_q$ and $M_q$ such that
\begin{align}\label{eqn:BDG}     
 m_q\mathbb{E}\Big[\Big(\int_0^T|\eta(u)|^2\,du\Big)^{\frac{q}{2}}\Big] \le \mathbb{E}\Big[\sup_{0\le s\le T}\Big|\int_0^s \eta(u)\,dB(u) \Big|^q\Big]\le M_q\mathbb{E}\Big[\Big(\int_0^T|\eta(u)|^2\,du\Big)^{\frac{q}{2}}\Big]
\end{align}
for all $T\ge 0$ and $\eta\in \mathbb{H}^q(0,T;\mathbb{R}^m)$.

\begin{assumption}\label{hypo: BSDE no delay}
	Let $\Phi:[0,\infty)\times \Lambda \times \mathbb{R} \to\mathbb{R}$ be a non-anticipative functional and define $f(s,\gamma,y):=-r(s,\gamma)y+\Phi(s,\gamma,y)$ for $s\in[0,\infty), \gamma\in \Lambda,y\in \mathbb{R}$.
	Assume the function $f$ satisfies the following conditions.
	\begin{enumerate}[(i)]
		\item The function $f$ is continuous and non-anticipative.
		\item There are constants $C_4>0$ and $\rho\ge 1$ such that
		\begin{align}
			|f(s,\gamma,0)|\le C_4(1+|\!|\gamma|\!|^\rho_s)\,
		\end{align}
		for all $(s,\gamma)\in [0,\infty)\times \Lambda$.
		\item\label{eqn:Lip_f} 
		The function $f$ is Lipschitz in $y$, uniformly in $(s,\gamma).$ 
		\item There exists a constant $\ell>0$ such that
		\begin{align}
			(y-y')(f(s,\gamma,y)-f(s,\gamma,y'))\le -\ell|y-y'|^2 
		\end{align}
		for all $\gamma\in \Lambda$, $y,y'\in \mathbb{R}$.
		\item\label{hypo mu} $\ell>\inf_{K>0}(K+\frac{1}{2}(\frac{C_r}{\sqrt{2K}}+M_{\rho\vee 2}C_3)^2)\rho$, where $M_{\rho\vee 2}$ represents the constant in the BDG inequality.
	\end{enumerate}
\end{assumption}

The following theorem indicates that if the function 
$f$ meets certain conditions, then the perpetual future with funding rate $(\Phi(s,X_s,Y(s))_{s\ge0}$ can be replicated, that is, there exists a unique $F^\Phi$-funding portfolio.
The proof of the following theorem is stated in Appendix \ref{app:3}.

\begin{theorem}\label{thm: BSDE no delay}
	Let Assumptions \ref{hypo: SDE} - \ref{hypo: BSDE no delay} hold and $\rho$ be the constant in Assumption \ref{hypo: BSDE no delay}. Then, the BSDE \eqref{eqn:BSDE} has a unique solution $(Y,Z)$ in $\mathbb{S}^2(0,\infty;\mathbb{R})\times \mathbb{H}^2(0,\infty;\mathbb{R}^m)$ such that $|\!|Y|\!|\le L(1+|\!|X|\!|^\rho)\label{eqn:path BSDE mean}$ for some constant $L>0$.
\end{theorem}

Note that our methodology differs from the standard derivative pricing framework. In the traditional approach, derivatives with a fixed terminal date are priced using a risk-neutral measure, with replicating portfolios derived via the martingale representation theorem. However, this standard method is not applicable to perpetual futures, which lack a fixed terminal date.
Instead, we transform the problem into an infinite-horizon BSDE and establish the existence and uniqueness of its solutions. Notably, without introducing  risk-neutral measures, the BSDE \eqref{BSDE_P} can be analyzed directly under the physical measure; however, this involves more complex analysis and is valid under more restrictive conditions. The risk-neutral measure introduced here is employed  to simplify the BSDE and facilitate the straightforward proofs of the existence and uniqueness of solutions.

We now examine the Feynman-Kac formula in the context of path-dependent infinite-horizon BSDEs and the corresponding path-dependent PDEs (PPDEs). Appendix \ref{app:FK} provides the definition of viscosity solutions for PPDEs and the proof of the following theorem.

\begin{theorem}\label{thm: fey} Let Assumptions \ref{hypo: SDE}-\ref{hypo: BSDE no delay} hold and let $\rho$ be the constant in Assumption \ref{hypo: BSDE no delay}.
	Assume further that the mapping  $(s,\gamma)\mapsto r(s,\gamma)\gamma(s)$ is Lipschitz continuous.	
	For $(s,\gamma)\in [0,\infty)\times \Lambda$, let   $X^{s,\gamma}$ be a solution to 
	\begin{align}\label{eqn: fey sde}
		X^{s,\gamma}(v)&=\gamma(s)+\int_s^v r(u,X_u^{t,\gamma})X^{t,\gamma}(u)\,du+\int_s^v\sigma(u,X_u^{t,\gamma})\,dB(u)\;,\quad v\ge s\,,\\
		X^{s,\gamma}(v)&=\gamma(v)\;,\quad 0\le v\le s\,.
	\end{align}
	
	\begin{enumerate}[(i)]
		\item\label{eqn: fey i} The infinite-horizon BSDE 
		\begin{align}
			&Y^{s,\gamma}(v)=Y^{s,\gamma}(T)+\int_v^T f(u,X^{s,\gamma}_u,Y^{s,\gamma}(u))\,du \\&\quad\quad\quad-\int_v^TZ^{s,\gamma}(u)\,dB(u)\,,\; s\le v\le T<\infty\label{eqn: fey bsde}\,
		\end{align} has a unique solution  $(Y^{s,\gamma},Z^{s,\gamma})$  in $\mathbb{S}^2(s,\infty;\mathbb{R})\times \mathbb{H}^2(s,\infty;\mathbb{R}^m)$ such that $|\!|Y^{s,\gamma}|\!|\le L(1+|\!|X^{s,\gamma}|\!|^{\rho})$  for some constant $L>0$.
		\item\label{eqn: fey ii}  Define a function $\varphi:[0,\infty)\times \Lambda\to\mathbb{R} $ as $\varphi(s,\gamma)=Y^{s,\gamma}(s),$
		then $Y^{0,x}(s)=\varphi(s,X^{0,x}_s)$ for all $x\in\mathbb{R}^m.$
		Then $\varphi$ is continuous and is a viscosity solution to the PPDE
		\begin{align}
			&-\partial_s \varphi(s,\gamma)-\frac{1}{2}\textnormal{tr}(\sigma\sigma^{\top}\partial_{xx}\varphi)(s,\gamma)- r(s,\gamma)\partial_{x}\varphi(s,\gamma)\gamma(s)-f(s,\gamma,\varphi(s,\gamma))=0\,.\qquad\label{eqn: fey pde}
		\end{align}
		\item\label{eqn: fey iii}
		Suppose $\varphi\in C_\rho^{1,2}([0,\infty)\times\Lambda)$ is a solution to the PPDE
		\begin{align}
			&\label{eqn:ppde}-\partial_s \varphi(s,\gamma)-\frac{1}{2}\textnormal{tr}(\sigma\sigma^{\top}\partial_{xx}\varphi)(s,\gamma)- r(s,\gamma)\partial_{x}\varphi(s,\gamma)\gamma(s)-f(s,\gamma,\varphi(s,\gamma))=0\,.
		\end{align}
		Then we have $(Y^{0,x},Z^{0,x})=(\varphi(s,X_s^{0,x}),\sigma^{\top}\partial_x \varphi(s,X_s^{0,x}))_{s\ge0}$ for all $x\in \mathbb{R}^m.$
	\end{enumerate}  
\end{theorem}

\section{Designing funding rates}
\label{sec:design_FR}

In this section, we address key questions related to the funding mechanism. 
Consider an issuer seeking to maintain the perpetual future price in alignment with a target value \( \varphi(s, X_s) \). A common example is \( \varphi(s,\gamma)= \varphi(s,\gamma_1,\cdots,\gamma_m)= c_0 + \sum_{i=1}^mc_i \gamma_i(s) \) for $c_0,c_1,\cdots,c_m\in\mathbb{R}$, which represents an index of composite assets. Additional examples are provided in Section \ref{sec:appli}. 
We emphasize that the target value $\varphi(s,X_s)$ is not necessarily tradable.
The central question is how to design the funding rate to ensure that the perpetual future price remains consistent with the desired process \( \varphi(s, X_s) \). In this section, we provide an answer to that question for any given $p\ge1$ and $\varphi\in C_p^{1,2}([0,\infty)\times {\Lambda})$.

\subsection{Instantaneous spot funding rates}\label{sec:inst}

We analyze instantaneous spot funding rates, an ideal mechanism for keeping the perpetual future price aligned with the specified target values.

\begin{assumption}
	\label{Hypo:H}
	Suppose that a function $H:\mathbb{R}^2\to\mathbb{R}$ satisfies the following conditions. \begin{enumerate}[(i)]
		\item\label{label: 2-2} $H(y_1,y_2)=0$ if $y_1=y_2.$
		\item\label{label: 2-3} The function $H$ is Lipschitz in $y_2$, uniformly in $y_1.$
		\item\label{label: 2-4} There exists a constant $\ell>0$ such that 
		\begin{align}
			(y_2-y_2')(H(y_1,y_2)-H(y_1,y_2'))\le -\ell |y_2-y_2'|^2
		\end{align}
		for all $y_1,y_2,y_2'\in \mathbb{R}$.
	\end{enumerate}
\end{assumption}

For any $p\ge1$ and $\varphi\in C_p^{1,2}([0,\infty)\times {\Lambda})$,
we present a method for designing funding rates that ensures the perpetual future price remains anchored to  \( \varphi(s, X_s) \) for $s\ge0$.
The following theorem states that 
the risk-neutral pricing BSDE \eqref{eqn:BSDE} with the funding rate functional $\Phi$ stated in \eqref{eqn:Phi} has a unique solution; moreover, the unique solution $(Y,Z)$ coincides with $(\varphi(s,X_s),(\partial_x\sigma\,\varphi)(s,X_s))_{s\ge0}$. Consequently, this funding rate functional $\Phi$ 
induces a perpetual future price aligned with \( \varphi(s, X_s) \) for $s\ge0.$
The constant $M_{\rho\vee 2}$ in the following theorem is the BDG inequality constant in \eqref{eqn:BDG}.
The proof of the following theorem is described in Appendix \ref{app:fundingrate}.

\begin{theorem}\label{thn:main}
	Let Assumptions \ref{hypo: SDE}-\ref{hypo: SDE_r} hold and $H:\mathbb{R}^2\to\mathbb{R}$ be a function satisfying Assumption \ref{Hypo:H}. For any $p\ge1$ and $\varphi\in C_p^{1,2}([0,\infty)\times {\Lambda})$, consider the funding rate functional
	\begin{equation}\label{eqn:Phi}
		\begin{aligned}
			\Phi(s,\gamma,y):&=H(\varphi(s,\gamma),y)
			-\partial_s \varphi(s,\gamma)-\frac{1}{2}\textnormal{tr}(\sigma\sigma^{\top}\partial_{xx}\varphi)(s,\gamma)- r(s,\gamma)\partial_{x}\varphi(s,\gamma)\gamma(s)+r(s,\gamma)y
		\end{aligned}
	\end{equation}
	for $(s,\gamma,y)\in [0,\infty)\times \Lambda\times \mathbb{R}$. Then, we have the following.
	\begin{enumerate}[(i)]
		\item 
		The map $(s,\gamma) \mapsto \Phi(s,\gamma,0)$ has polynomial growth. Let \(\rho \ge p\) and $C_\Phi>0$ be constants such that
		\(
		|\Phi(s,\gamma,0)| \leq C_{\Phi} ( 1 + |\!|\gamma|\!|_s^\rho ) \)
		for all \((s, \gamma) \in [0, \infty) \times {\Lambda}\).  
		\item Let $\ell$ be the constant in Assumption \ref{Hypo:H}.
		If
		\begin{equation}
			\label{eqn:ell} 
			\ell>\inf_{K>0}(K+\frac{1}{2}(\frac{C_r}{\sqrt{2K}}+M_{\rho\vee2 }C_3)^2)\rho\,,
		\end{equation} 
		the risk-neutral pricing BSDE \eqref{eqn:BSDE}	
		has a unique solution $(Y,Z)$ in $\mathbb{S}^2(0,\infty;\mathbb{R})\times \mathbb{H}^2(0,\infty;\mathbb{R}^m)$ such that $|\!|Y|\!|\le L(1+|\!|X|\!|^{\rho})$ for some constant $L>0$. Moreover, $Y(s)=\varphi(s,X_s)$ and $Z(s)=(\partial_x \varphi\, \sigma)(s,X_s)$ for $s\ge0.$
	\end{enumerate}	  
\end{theorem}

For given $p\ge1$ and $\varphi\in C_p^{1,2}([0,\infty)\times {\Lambda})$, a funding portfolio is said to be admissible if it is $p$-admissible. In the above theorem, because the unique solution is given by \( Y(s) = \varphi(s, X_s) \) and \( \varphi \) has polynomial growth of order \( p \) with $p\le \rho,$ it follows that there exists a unique \( F^\Phi \)-funding portfolio among all admissible funding portfolios. An important observation is that this uniqueness holds across a broad class of funding portfolios.
Recall the concept of \(\rho\)-admissibility from Definition \ref{defi:funding_port}. The theorem guarantees the existence of a unique \( F^\Phi \)-funding portfolio within the class of all \(\rho\)-admissible funding portfolios. As stated in \eqref{eqn:ell}, as \(\ell\) increases, \(\rho\) can be chosen larger. Therefore, the uniqueness result extends to a wider class of portfolios as \(\ell\) grows.
As $\ell$ is a constant selected by the issuer, the issuer can ensure the existence of a unique price and a unique replicating portfolio within a desired class of funding portfolios.

The following corollary presents a simpler method of designing funding rates that ensures the perpetual future price aligns with the target value $\varphi(s,X_s)$ for $s\ge0.$ 
It is obtained by setting \(\rho = p + 2\) in the above theorem.
The detailed proof is provided in Appendix \ref{app:fundingrate}.

\begin{corollary}\label{cor:main}
	Let Assumptions \ref{hypo: SDE}-\ref{hypo: SDE_r} hold. For given $p\ge1$ and $\varphi\in C_p^{1,2}([0,\infty)\times {\Lambda})$,
	choose 
	any function $H$ satisfying Assumption \ref{Hypo:H} for the constant
	$\ell$ such that
	\begin{align}
		\ell>\inf_{K>0}(K+\frac{1}{2}(\frac{C_r}{\sqrt{2K}}+M_{p+2 }C_3)^2)(p+2)
	\end{align}
	and define the funding rate functional $\Phi$ as stated in \eqref{eqn:Phi}. 
	Then, the risk-neutral pricing BSDE \eqref{eqn:BSDE}	
	has a unique solution $(Y,Z)$ in $\mathbb{S}^2(0,\infty;\mathbb{R})\times \mathbb{H}^2(0,\infty;\mathbb{R}^m)$ such that $|\!|Y|\!|\le L(1+|\!|X|\!|^{p+2})$ for some constant $L>0$. Moreover, $Y(s)=\varphi(s,X_s)$ and $Z(s)=(\partial_x \varphi\, \sigma)(s,X_s)$ for $s\ge0.$  
\end{corollary}

The funding rate $\Phi$ presented in \eqref{eqn:Phi} consists of three components:
$H(\varphi(s,\gamma), y)$, $-\mathcal{L}\varphi(s,\gamma)$ and $r(s,\gamma)y$, where
\begin{equation}
	\begin{aligned}
		\mathcal{L}\varphi(s,\gamma) := \partial_s \varphi(s,\gamma) &+ \frac{1}{2} \textnormal{tr}(\sigma\sigma^{\top} \partial_{xx}\varphi)(s,\gamma) + r(s,\gamma) \partial_x \varphi(s,\gamma) \gamma(s)\,.
	\end{aligned}
\end{equation}
Each term plays a distinct role in the funding mechanism.
The first term $H(\varphi(s,\gamma), y)$ plays a central role and will be discussed in more detail below.
The second term $-\mathcal{L}\varphi(s,\gamma)$ is associated with the no-arbitrage condition. Specifically, observe that the process
\begin{equation}
	\begin{aligned}
		&\quad Y(s) - \int_0^s \big(r(u, X_u)Y(u) - \Phi(u, X_u, Y(u))\big) \, du= \varphi(s, X_s) - \int_0^s \mathcal{L}\varphi(u, X_u)\, du\,,\;s\ge0
	\end{aligned}
\end{equation} 
is a local martingale under the risk-neutral measure because $\mathcal{L}$ is the infinitesimal generator of $\varphi$.
This implies that the perpetual futures price $Y$, when adjusted by the funding fee $\Phi$, is arbitrage-free.
More precisely, 
it satisfies the condition of no free lunch with vanishing risk (\cite{delbaen1994general}).
The third term $r(s,\gamma)y$  reflects the interest rate cost or benefit from holding the perpetual future.

The first term $H(\varphi(s,\gamma), y)$ plays a crucial role in the funding mechanism.
A remarkable phenomenon is that the transaction corresponding to $H(\varphi(s,\gamma), y)$ does not occur in practice, since Assumption \ref{Hypo:H} \eqref{label: 2-2} implies $H(\varphi(s, X_s), Y(s)) = H(Y(s), Y(s)) = 0$; nevertheless, this term is essential for ensuring the uniqueness of perpetual futures prices.	
From a mathematical perspective, it ensures the uniqueness of solutions to the infinite-horizon BSDE \eqref{eqn:BSDE}.
From an economic standpoint, it enforces uniqueness through the law of supply and demand. To illustrate this, assume for simplicity that the short rate is zero and consider the case where \( \varphi(s, \gamma) = \varphi(s, \gamma_1, \ldots, \gamma_m) = \gamma_1(s) \) and \( H(y_1, y_2) = \ell(y_1 - y_2) \). Then the funding rate functional defined in \eqref{eqn:Phi} simplifies to
$
\Phi(s, X_s, Y(s)) = \ell ( X_1(s) - Y(s) ).$
If $Y(s)<X_1(s)$ at time $s$, indicating that the perpetual future price is below the first asset price, the long positions receive a positive amount $-\ell (Y(s) - X_1(s))\, ds$ from the short positions. Consequently, investors are incentivized to purchase more futures, increasing its price.	
Conversely, if $Y(s)>X_1(s)$ at time $s$, the long positions must pay $\ell (Y(s) -X_1(s)) \, ds$ to the short positions. Consequently, they are likely to sell their futures, which drives the price down. This dynamic of supply and demand keeps the perpetual future price aligned with the first asset price.

In this context, the parameter \(\ell\) must be chosen sufficiently large to ensure the proper functioning of the law of supply and demand. Specifically, the condition stated in \eqref{eqn:ell} needs to be satisfied. If this condition is not met, the uniqueness of funding portfolios cannot be guaranteed. 
To illustrate this, consider the two-dimensional uncorrelated Black-Scholes model \(X = (X_1(s), X_2(s))_{s \ge 0}\) with a constant short rate \(r=1\). The risk-neutral dynamics described in \eqref{SDE_X} simplify to
\begin{align}
	X_i(s) = X_i(0) &+ \int_0^s X_i(u)\, du + \int_0^s \sigma_i X_i(u)\, dB_i(u)\,,\; i=1,2
\end{align}
where $\sigma_1$ and $\sigma_2$ are volatility constants.
Suppose we set 
\( \varphi(s, \gamma) = \varphi(s, \gamma_1, \ldots, \gamma_m) = \gamma_1(s) \) and \( H(y_1, y_2) = \ell(y_1 - y_2) \) with $\ell=1.$
The parameter \(\ell=1\) does not satisfy the condition \eqref{eqn:ell} because
$\inf_{K>0}(K+\frac{1}{2}(\frac{C_r}{\sqrt{2K}}+M_{\rho\vee2 }C_3)^2)\rho > \inf_{K>0}(K+\frac{1}{4K})p=1$, 
where we have used \(C_r = 1\), \(M_{\rho \vee 2} > 0\), \(C_3 > 0\), and \(\rho \ge p = 1\).
The funding rate functional \(\Phi\) defined in \eqref{eqn:Phi} is zero in this case. It is evident that a zero funding rate does not result in a unique future price. Indeed, the risk-neutral pricing BSDE 
\begin{align}
	Y(s)=Y(T)-\int_s^T Y(u)\,du&-\int_s^TZ_1(u)\,dB_1(u)-\int_s^TZ_2(u)\,dB_2(u) 
\end{align}
admits multiple solutions. For examples, $(Y,Z_1,Z_2)=(2X_1,2\sigma_1X_1,0)$, $(X_1+2X_2,\sigma_1X_1,2\sigma_2X_2),$ $(X_1^{-{2}/{\sigma_1^2}},\frac{2}{\sigma_1}X_1^{-{2}/{\sigma_1^2}},0)$ are solutions.
This highlights that choosing 
$\ell$ sufficiently large is essential to ensure the uniqueness of funding portfolios.

The funding rate stated in \eqref{eqn:Phi} can be expressed in the model-free form
\begin{equation} 
	\begin{aligned}
		\Phi(s,X_s,Y(s))
		&=H(\varphi(s,X_s),Y(s))
		-\partial_s \varphi(s,X_s)-\frac{1}{2}\sum_{1\le i,j\le m} \partial_{i}(\partial_{j}{\varphi})(s,X_s)\frac{d}{ds}\langle X_i,X_j \rangle_s\\
		&\quad-\partial_{x}\varphi(s,X_s)X_s\frac{d}{ds}\ln G_s+ Y(s)\frac{d}{ds}\ln G_s\,.
	\end{aligned}
\end{equation}
This is directly obtained by
$r(s,X_s)=\frac{d}{ds}\ln G_s$ and $(\sigma\sigma^{\top})_{i,j}(s,X_s)=\frac{d}{ds}\langle X_i,X_j \rangle_s$.   
Thus, the precise knowledge of the non-anticipative functionals $r$, $\mu$, and $\sigma$ is not needed to determine the funding rate $\Phi.$
This observation is from \cite[Remark 4]{angeris2023primer}.

Assumption \ref{Hypo:H} encompasses a broad range of funding rate functions. Typical examples include linear functions, such as $\ell(y_1 - y_2)$, as well as piecewise linear functions, such as
$
\ell_1(y_1 - y_2)\,\mathbb{I}_{\{|y_1 - y_2| \leq 1\}} + \ell_2(y_1 - y_2)\,\mathbb{I}_{\{|y_1 - y_2| > 1\}}$
and $
\ell_1(y_1 - y_2)\,\mathbb{I}_{\{y_1 > y_2\}} + \ell_2(y_1 - y_2)\,\mathbb{I}_{\{y_1 \leq y_2\}},$
where \(\ell, \ell_1, \ell_2 > 0\). Furthermore, instead of constant coefficients, the assumption allows these coefficients to be stochastic processes \((\ell(s))_{s \ge 0}\), \((\ell_1(s))_{s \ge 0}\), and \((\ell_2(s))_{s \ge 0}\), which are bounded below by some positive constants. The most common form among these is the linear function
$
H(y_1, y_2) = \ell(y_1 - y_2),$
in which case the funding rate \(\Phi\) defined in \eqref{eqn:Phi} is referred to as a constant proportion funding rate. 
While most existing literature concentrates exclusively on constant proportion funding rates, our findings are applicable to a broad class of funding schemes. This is the first study to derive unique pricing and replicating portfolios for perpetual futures across a diverse range of funding rate structures.

\subsection{Path-dependent funding rates\label{section: delay}}

We examine the constant proportion funding rate and its path-dependent variant for practical applications. The short rate is assumed to be a constant \( r \) throughout this section. We begin by defining path-dependent funding rate functionals, which generalize the concept introduced in Definition \ref{defi:fundingrate}. A key distinction is that they are dependent on the historical trajectory of \( Y \).

\begin{definition}\label{defi:p_fundingrate} If a funding rate is expressed as \( F(s) = \Phi(s, X_s, Y_s) \) for a non-anticipative functional $\Phi:[0,\infty)\times \Lambda \times C([0,\infty);\mathbb{R}) \to\mathbb{R}$, then we refer to \(\Phi\) as the path-dependent funding rate functional and denote the corresponding funding rate by \( F^\Phi \). 
\end{definition}

For any $p\ge1,$ $\ell>0,$  $\delta>0$ and $\varphi\in C_p^{1,2}([0,\infty)\times {\Lambda})$, we define two non-anticipative functionals
$\Phi: [0,\infty)\times \Lambda\times \mathbb{R} \to\mathbb{R}$ and $\Phi^\delta: [0,\infty)\times \Lambda\times C([0,\infty);\mathbb{R}) \to\mathbb{R}$ as 
\begin{equation} 
	\begin{aligned}
		\label{eqn:const_prop}
		\Phi(s,\gamma,y)=\ell(\varphi(s,\gamma)-y)
		-\partial_s \varphi(s,\gamma)-\frac{1}{2}\textnormal{tr}(\sigma\sigma^{\top}\partial_{xx}\varphi)(s,\gamma)- r\partial_{x}\varphi(s,\gamma)\gamma(s)+ry
	\end{aligned} 
\end{equation}
and
\begin{align}\label{eqn:variant}
	\Phi^\delta(s,\gamma,\eta)&=\frac{1}{\delta}\int_{s-\delta}^s \Phi(u,\gamma,\eta(u))\,du\,.
\end{align}
The non-anticipative functional \(\Phi\) represents a constant proportion funding rate based on current spot prices. However, on most exchanges, the funding fee is calculated as an average of values over the past 8 hours rather than relying on the current spot price. This can be mathematically modeled using the path-dependent functional \(\Phi^\delta\), where \(\delta = \frac{1}{1095}\), corresponding to an 8 hour averaging window.
Because \(\Phi^\delta\) is obtained by averaging over recent short periods, we consider \(\Phi^\delta\) to be a practical approximation of \(\Phi\).

This path-dependent funding rate induces an infinite-horizon delayed BSDE.
By setting $F(s)=\Phi(s, X_s, Y_s)$ in \eqref{eqn:Q_F},
we obtain  
\begin{equation} \label{eqn: infinie delay bsde}
	\begin{aligned}
		Y^{\delta}(s)
		=Y^{\delta}(T)&-\int_s^T(rY^{\delta}(u)-\Phi^\delta(u,X_u,Y_u^{\delta}))\,du-\int_s^TZ^{\delta}(u)\,dB(u)
	\end{aligned}
\end{equation} 
for $0\le s\le T<\infty$,
which is 
the risk-neutral pricing BSDE for the path-dependent funding rate $\Phi^\delta$. 
A direct calculation yields 
\begin{equation} 
	\begin{aligned}
		Y^{\delta}(s) 
		&=Y^{\delta}(T)-\int_s^T(rY^{\delta}(u)-\frac{1}{\delta}\int_{u-\delta}^u\Phi(v,X_v,Y^{\delta}(v))\,dv\,du-\int_s^TZ^{\delta}(u)\,dB(u)\\
		&=Y^{\delta}(T)+\int_s^T g(u,X_u)-r Y^{\delta}(u)-\frac{\ell-r}{\delta}\int_{u-\delta}^u Y^{\delta}(v)\,dv\,du-\!\int_s^TZ^{\delta}(u)\,dB(u) 
	\end{aligned}
\end{equation}
where
$g:[0,\infty)\times \Lambda\to \mathbb{R}$ is a non-anticipative functional defined 
as
\begin{align}
	g(s,\gamma)&:=\frac{1}{\delta}\int_{s-\delta}^s\ell \varphi(v,\gamma)
	-\partial_s \varphi(v,\gamma)-\frac{1}{2}\textnormal{tr}(\sigma\sigma^{\top}\partial_{xx}\varphi)(v,\gamma)- r\partial_{x}\varphi(v,\gamma)\gamma(v)\,dv
\end{align}
for $(s,\gamma) \in [0,\infty)\times \Lambda.$
This BSDE is classified as a delayed BSDE due to the presence of an additional term \(\int_{u-\delta}^u Y(v)\,dv\). This extra component introduces additional complexity compared to the classical BSDE framework.
The infinite-horizon delayed BSDE has not been studied before.
To handle this delay term, we extend the domain of all processes \(Y \in \mathbb{S}^2(0,\infty;\mathbb{R})\) to include the interval \([- \delta, 0]\) by defining
\(
Y(s) = Y(0)\) for
\(s \in [-\delta, 0].\)
Despite this extension, we continue to denote the process space as \(\mathbb{S}^2(0,\infty;\mathbb{R})\) rather than \(\mathbb{S}^2(-\delta,\infty;\mathbb{R})\).

We recall from Theorem \ref{thn:main} that
for any $p\ge1$ and $\varphi\in C_p^{1,2}([0,\infty)\times {\Lambda})$, 
the non-anticipative functional 
$\Phi(s,\gamma,0)$ has polynomial growth in $(s,\gamma)$.
There are constants
\(\rho \ge p\) and $C_\Phi>0$ such that
\(
|\Phi(s,\gamma,0)| \leq C_{\Phi} ( 1 + |\!|\gamma|\!|_s^\rho ) \)
for all \((s, \gamma) \in [0, \infty) \times {\Lambda}\).
It can be easily shown that $|g(s,\gamma)|\le C_{\Phi}(1+|\!|\gamma|\!|^{\rho}_s)$
for all \((s, \gamma) \in [0, \infty) \times {\Lambda}\).

\begin{assumption}\label{hypo: BSDE delay}
	Consider the funding rate functionals $\Phi$ and $\Phi^\delta$ defined in \eqref{eqn:const_prop} and \eqref{eqn:variant}, respectively.
	Let
	\(\rho \ge p\) and $C_\Phi>0$ be constants such that
	\(
	|\Phi(s,\gamma,0)| \leq C_{\Phi} ( 1 + |\!|\gamma|\!|_s^\rho ) \)
	for all \((s, \gamma) \in [0, \infty) \times {\Lambda}\).	
	Suppose $\ell\in \mathbb{R}$ and $0<\delta<1$ satisfy the following conditions.
	\begin{enumerate}[(i)]
		\item\label{hypo: delay infinite decreasing} $\ell>1+\inf_{K>0}(K+\frac{1}{2}(\frac{r}{\sqrt{2K}}+M_{\rho\vee 2}C_3)^2)\rho$.  
		\item The constant $\delta$ satisfies
		\begin{align}
			&\frac{1}{3}\frac{e^{|6(\ell-r)^2-2\ell+2|\delta}-1}{|6(\ell-r)^2-2\ell+2|\delta}<1\,,\\
			&e^\rho((\ell-r)^2+\frac{1}{2}|\ell-r|\ell+2|\ell-r|)\delta<1\,.
		\end{align}
	\end{enumerate}
\end{assumption}

We primarily focus on two key topics related to path-dependent funding rate functionals. 
First, we examine the existence and uniqueness of \(\Phi^\delta\)-funding portfolios. Theorem \ref{thm: BSDE no delay} is not applicable as \(\Phi^\delta\) depends on the historical trajectory of \(Y\).
Second, we compare the values of the \(\Phi^\delta\)-funding portfolio and the \(\Phi\)-funding portfolio. Specifically, we demonstrate that when \(\delta\) is small, the funding rate functionals \(\Phi\) and \(\Phi^\delta\) produce similar perpetual future prices \(Y\) and \(Y^\delta\) as well as similar replicating portfolios \(Z\) and \(Z^\delta\). 
The main results on these topics are presented in Theorem \ref{thm:approx}, with detailed proofs provided in Appendix \ref{app:delay}. Although these topics are conceptually straightforward, establishing rigorous proofs is  
complex and challenging. The upper bounds \(L_1, L_2, L_3, L_4, L_5\) specified in the theorem can be explicitly calculated. More refined upper bounds are provided in Appendix \ref{L speci}. The first inequality in \eqref{eqn:esti} holds without taking expectations.
Note that the following theorem holds in particular for $\rho=p+2$, similar to Corollary \ref{cor:main}.

\begin{theorem}\label{thm:approx}
	Let Assumptions \ref{hypo: SDE}-\ref{sharpe} hold. Suppose that the 
	funding rate functionals $\Phi$ and $\Phi^\delta$ defined in \eqref{eqn:const_prop} and \eqref{eqn:variant}, respectively, satisfy Assumption \ref{hypo: BSDE delay}. 
	Then, we have the followings.
	\begin{enumerate}[(i)]
		\item \label{thm:delay_1} The BSDE \eqref{eqn:BSDE} has a unique solution $(Y,Z)$ in $\mathbb{S}^2(0,\infty;\mathbb{R})\times \mathbb{H}^2(0,\infty;\mathbb{R}^m)$ such that  $|\!|Y|\!|\le L(1+|\!|X|\!|^{\rho})$ for some constant $L>0$. In addition, $Y(s)=\varphi(s,X_s)$ and $Z(s)=(\partial_x \varphi\,\sigma)(s,X_s)$ for $s\ge0$. 
		\item\label{thm:delay_2} The BSDE \eqref{eqn: infinie delay bsde} has a unique solution $(Y^\delta,Z^\delta)$ in $\mathbb{S}^2(0,\infty;\mathbb{R})\times \mathbb{H}^2(0,\infty;\mathbb{R}^m)$ such that $|\!|Y^\delta|\!|\le L(1+|\!|X|\!|^{\rho})$ for some constant $L>0$.
	\end{enumerate} Moreover, we have
	\begin{align}\label{eqn:estimate}
		&\lim_{\delta\to 0}\mathbb{E}\big[|\!|Y^\delta-Y|\!|_T\big]=0 \,, \\
		&\lim_{\delta\to 0}\mathbb{E}\Big[\int_0^T|Z^\delta(u)-Z(u)|^2\,du\Big]=0
	\end{align} for all $T\ge 0$. In particular, if 
	there is a constant $C_5>0$ such that 
	$|\Phi(s,\gamma,0)-\Phi(s',\gamma',0)| \le C_5(1+|\!|\gamma_s|\!|_{s\vee s'}^{\rho-1}+|\!|\gamma'_{s'}|\!|_{s\vee s'}^{\rho-1})(\sqrt{|s-s'|}+|\!|\gamma_s-\gamma'_{s'}|\!|_{s\vee s'})$ for all $(s,\gamma),(s',\gamma')\in [0,\infty)\times \Lambda$,
	then for some positive constants $L_1,L_2,L_3,L_4,L_5,$ we have    
	\begin{equation} 
		\begin{aligned}
			\label{eqn:esti} 
			&|\!|Y^{\delta}-Y|\!|_T\le (L_1+L_2|\!|X|\!|_T^\rho)\sqrt{\delta}\\
			&\mathbb{E}\Big[\int_0^T|Z^\delta(u)-Z(u)|^2\,du\Big]\le (1+T)(L_3\mathbb{E}[|\!|X|\!|_T^{2\rho}]+L_4\mathbb{E}[|\!|X|\!|_T^{\rho}]+L_5)\sqrt{\delta}
		\end{aligned} 
	\end{equation}   
	for all $T\ge0$.
\end{theorem}

\section{Applications}
\label{sec:appli}

\subsection{Perpetual power index futures}\label{sec:PIS}

We study how to design funding rate functionals for perpetual power index futures.
Consider a market with $m+1$ assets consisting of
a
money market account
with constant short rate $r\ge0$
and
$m$ tradable assets given as a solution to the path-dependent SDE \eqref{eqn: path dep SDE}
for $x\in \mathbb{R}^m$ and non-anticipative functionals $\mu:[0,\infty)\times \Lambda\to \mathbb{R}^m$, $\sigma:[0,\infty)\times \Lambda\to \mathbb{R}^{m\times m}$. We take Assumptions \ref{hypo: SDE} and \ref{sharpe} to be valid.

We examine an issuer seeking to keep the perpetual future price aligned with the market's power index.
More precisely, the perpetual future price at time $s$ is anchored to  $\varphi(s,X_s)$ for all $s\ge0$ where
$\varphi(s,\gamma)=\varphi(s,\gamma_1,\cdots,\gamma_m):=c_0+\sum_{i=1}^mc_i\gamma_i^{p_i}(s)$ for given $c_0,c_1,\cdots,c_m\in\mathbb{R}$ and $p_1,\cdots,p_m\in\mathbb{N}.$
By 
Theorem \ref{thn:main},  
a funding rate functional $\Phi: [0,\infty)\times \Lambda\times \mathbb{R} \to\mathbb{R}$ defined as
\begin{equation} 
\begin{aligned}
	\Phi(s,\gamma,y)
	:=H(\varphi(s,\gamma),y)-\frac{1}{2}\textnormal{tr}(\sigma\sigma^{\top}\partial_{xx}\varphi)(s,\gamma)
	 -r\gamma(s)(c_1p_1\gamma_1^{p_1-1}(s),\cdots,c_mp_m\gamma_m^{p_m-1}(s) )+ry
\end{aligned}
\end{equation}
where  the matrix $\partial_{xx}\varphi(s,\gamma)$ is   
\begin{align}
	\begin{pmatrix}
		c_1p_1(p_1-1)\gamma_1^{p_1-2}(s) & \cdots &0 \\ \vdots & \ddots & \vdots \\ 0&\cdots  & c_mp_m(p_m-1)\gamma_1^{p_m-2}(s)
	\end{pmatrix}
\end{align}
gives the desired perpetual future price if $H$  
satisfies Assumption \ref{Hypo:H} with the constant  $\ell>\inf_{K>0}(K+\frac{1}{2}(\frac{r}{\sqrt{2K}}+M_{2\vee \rho}C_3)^2)\rho$ for $\rho:=\max_{1\le i\le m}p_i$. The risk-neutral pricing BSDE \eqref{eqn:BSDE}
has a unique solution $(Y,Z)$
in $\mathbb{S}^2(0,\infty;\mathbb{R})\times \mathbb{H}^2(0,\infty;\mathbb{R}^m)$ such that $|\!|Y|\!|\le L(1+|\!|X|\!|^\rho)$ for some constant $L>0$ and this unique solution is
$$(Y,Z)=(\varphi(s,X_s),(c_1 p_1X_1(s),\cdots, c_m p_mX_m(s))\sigma(s,X_s))_{s\ge0}\,.$$
The perpetual future price $Y$ with this funding rate functional $\Phi$ coincides with the 
market's power index
and the corresponding replicating portfolio is
$\phi(s):=Z(s)\sigma^{-1}(s,X_s)$ for $s\ge0.$

We examine the constant proportion funding rate functional \(\Phi\) and its path-dependent variant \(\Phi^\delta\) defined in \eqref{eqn:const_prop} and \eqref{eqn:variant}, respectively, with the constants \(\ell\) and \(\delta\) satisfying Assumption \ref{hypo: BSDE delay}. 
According to Theorem \ref{thm:approx}, the path-dependent funding rate \(\Phi^\delta\) generates uniquely the perpetual future price process \(Y^\delta\) and the replicating portfolio \(\phi^\delta(s) :=  Z^\delta(s)\sigma^{-1}(s, X_s)\). Furthermore, these  are close to the original future price \(Y\) and replicating portfolio \(\phi\), as described in \eqref{eqn:esti}.

As a special case, consider the traditional Black-Scholes model. Under the risk-neutral measure, the stock price follows
$dX(s)=rX(s)\,ds+\sigma X(s)\,dB(s)$
where $r\ge0$ denotes the short rate.
Suppose $m=1,$ $r=0.02$, $\sigma=0.3$, $c_0=0$, $c_1=1$, $p_1=1$,
then \eqref{eqn:ell} is satisfied for $\ell>0.26227$. The original funding rate $\Phi$ with this constant $\ell$ yields the desired perpetual future price and the replicating portfolio. For the path-dependent funding rate $\Phi^\delta$, we set $\delta=\frac{1}{1,095}$, which corresponds to an 8 hour period, a commonly used time interval in cryptocurrency markets. 
We choose $\ell$ such that $1.26227<\ell<15.75125$. Then these constants $\ell$ and $\delta$ satisfy Assumption \ref{hypo: BSDE delay}.
Let us compare the perpetual futures derived from the funding rate functionals \(\Phi\) and \(\Phi^\delta\).
According to \eqref{eqn:esti} and \eqref{eqn:estiZ}, we have that for all $T\ge0,$
\begin{align} 	
	|\!|Y^{\delta}-Y|\!|_T&\le (3.68432|\!|X|\!|_T+0.84216)\sqrt{\delta}\le 0.11134|\!|X|\!|_T+0.02545 	
\end{align}
and
\begin{align}
&\quad\mathbb{E}\Big[\int_0^T|Z^\delta(u)-Z(u)|^2\,du\Big]\\
	&\le \big( (0.41021+82.041T)\mathbb{E}[|\!|X|\!|_T^2]+(0.06227+31.97983T)\mathbb{E}[|\!|X|\!|_T]+0.00236+2.78067T\big) \sqrt{\delta} \\
	&\le  (0.01239+2.47927T)\mathbb{E}[|\!|X|\!|_T^2]+ (0.00188+0.96642T)\mathbb{E}[|\!|X|\!|_T]+0.00007+0.084031T\,.
\end{align}

As an additional application, we consider an issuer seeking to keep the perpetual future price aligned with \( 
X_1^{p_1} X_2^{p_2} \cdots X_m^{p_m}
\)
where \( p_1, p_2, \ldots, p_m \in \mathbb{N} \). Most of the results are analogous to those obtained for the previous power index futures.
Define
\(
\varphi(s,\gamma) := \prod_{i=1}^m \gamma_i^{p_i}(s)
\) and 
\(
\rho := \sum_{i=1}^m p_i.
\)
While our analysis holds for any \( m \in \mathbb{N} \), for simplicity, we focus on the case \( m=2 \) in the following discussion.
We introduce a funding rate functional \(\Phi: [0,\infty) \times \Lambda \times \mathbb{R} \to \mathbb{R}\) defined as
$\Phi(s,\gamma,y)=H(\varphi(s,\gamma),y)
-\frac{1}{2}\textnormal{tr}(\sigma\sigma^{\top} \partial_{xx}\varphi)(s,\gamma)-r\gamma(s)(p_1\gamma_1^{p_1-1}(s)\gamma_2^{p_2}(s),p_2\gamma_1^{p_1}(s)\gamma_2^{p_2-1}(s) )+ry$
where the matrix  $\partial_{xx}\varphi(s,\gamma)$ is  \begin{align}
		\begin{pmatrix}
			p_1(p_1-1)\gamma_1^{p_1-2}(s)\gamma_2^{p_2}(s) & p_1p_2\gamma_1^{p_1-1}(s)\gamma_2^{p_2-1}(s) \\ p_1p_2\gamma_1^{p_1-1}(s)\gamma_2^{p_2-1}(s) & p_2(p_2-1)\gamma_1^{p_1}(s)\gamma_2^{p_2-2}(s) 
		\end{pmatrix}
\end{align}
and  \( H \) satisfies Assumption \ref{Hypo:H} with the constant  
$\ell>\inf_{K>0}(K+\frac{1}{2}(\frac{r}{\sqrt{2K}}+M_{2\vee \rho }C_3)^2)\rho$ for $\rho:=p_1+p_2$.
The perpetual future price \( Y \), associated with this funding rate \(\Phi\), coincides with $X_1^{p_1} X_2^{p_2}$. The corresponding replicating portfolio is
$\phi(s)=Z(s)\sigma^{-1}(s,X_s)=(p_1X_1^{p_1-1}(s)X_2^{p_2}(s),p_2X_1^{p_1}(s)X_2^{p_2-1}(s))\sigma^{-1}(s,X_s)$
for $s\ge0$.
As in the case of power index futures, the funding rate functionals $\Phi$ and $\Phi^\delta$, with the parameters $\ell$ and $\delta$ satisfying Assumption \ref{hypo: BSDE delay},  yield similar perpetual futures prices $Y$ and $Y^\delta$, as well as similar replicating portfolios $Z$ and $Z^\delta$.

\subsection{Perpetual foreign exchange futures}
One of the commonly traded futures in 
cryptocurrency markets is a perpetual foreign exchange future.
We consider 
an issuer seeking to keep
the perpetual future price aligned with the foreign exchange rate.
Suppose $m=1$ and let $r_d$ and $r_f$ be the domestic and foreign short rates, respectively, and $U=(U(s))_{s\ge0}$ be the exchange rate. We assume that $U$ is a stochastic process satisfying
$	U(s)=U(0)+\int_0^s b(u,U_u)\,du+\int_0^s v(u,U_u)\,dW(u)
$
for $U(0)\in \mathbb{R}$ and non-anticipative functionals $b:[0,\infty)\times \Lambda\to \mathbb{R}$, $v:[0,\infty)\times \Lambda\to \mathbb{R}$. The process
$X=(U(s)e^{r_f s})_{s\ge0}$
is a wealth process and satisfies
$
X(s)=X(0)+\int_0^s \mu(u,X_u) \,du+\int_0^s  \sigma(u,X_u)\,dW(u) $
where
$\mu(s,\gamma):=r_f\gamma(s)+e^{r_fs}b(s,e^{-r_f(\cdot\wedge s)}\gamma_s)$ and $\sigma(s,\gamma):=e^{r_fs}v(s,e^{-r_f(\cdot\wedge s)}\gamma_s).$
Assume that these non-anticipative functionals $\mu$ and $\sigma$ satisfy Assumptions \ref{hypo: SDE} and \ref{sharpe}. 
Define $\varphi(s,\gamma):=e^{-r_fs}\gamma(s)$, then $U(s)=\varphi(s,X_s)$ for $s\ge0.$

We examine an issuer seeking to keep the perpetual future price aligned with the foreign exchange rate.
By 
Theorem \ref{thn:main},  
a funding rate functional $\Phi: [0,\infty)\times \Lambda\times \mathbb{R} \to\mathbb{R}$ defined as
$\Phi(s,\gamma,y):=H(\varphi(s,\gamma),y)
-(r_d-r_f)\varphi(s,\gamma)+r_d y$
yields the desired perpetual future if $H$ satisfies Assumption \ref{Hypo:H} for the constant  $\ell>\inf_{K>0}(K+\frac{1}{2}(\frac{r_d}{\sqrt{2K}}+M_{2}C_3)^2)$.
The perpetual future price derived from this funding rate function $\Phi$ coincides with the 
foreign exchange rate $U$
and the replicating portfolio is
$\phi(s):=e^{-r_fs}$ for $s\ge0.$

We compare the constant proportion funding rate functional \(\Phi\) with its path-dependent variant \(\Phi^\delta\)
for the parameters \(\ell\) and \(\delta\) satisfying Assumption \ref{hypo: BSDE delay}. The path-dependent funding rate \(\Phi^\delta\) uniquely generates the perpetual future price process \(Y^\delta\) and the associated replicating portfolio
\(
\phi^\delta.
\)
Furthermore, these approximations are close to the original future price \(Y=U\) and the replicating portfolio \(\phi(s)=e^{-r_fs},s\ge0\) as described in \eqref{eqn:esti}.

\subsection{Geometric mean constant funds}\label{sec:CFMM}

Following \cite{angeris2020improved}, \cite{evans2020liquidity}, and \cite{angeris2023primer},  
we investigate a geometric mean constant function market
maker (CFMM). 
Consider the 
multi-dimensional Black-Scholes model.
Assume that the short rate is a constant $r\ge0$ and $m$ tradable assets are given as
$  dX(s)=D(X(s))\mu\,ds+D(X(s))\sigma\,dW(s)  $
where $\mu=(\mu_1\cdots\mu_m)^\top\in\mathbb{R}^{m\times 1}$ and
\begin{align}
	\sigma
	=\begin{pmatrix}\sigma_1\\
		\vdots
		\\ \sigma_m
	\end{pmatrix}
	=\begin{pmatrix}\sigma_{11} &\cdots &\sigma_{1m} 
		\\ \vdots & \ddots &\vdots 
		\\ \sigma_{m1}&\cdots & \sigma_{mm}
	\end{pmatrix}
	\in\mathbb{R}^{m\times m}
\end{align}
and $\sigma$ is invertible. Here, for any \( x = (x_1, \cdots, x_m) \), we denote by \( D(x) \)  the diagonal matrix whose \( i \)-th diagonal entry is \( x_i \). It is evident that
Assumptions \ref{hypo: SDE}-\ref{sharpe} are met. 
Under the risk-neutral measure, the process $X$ satisfies
$dX(s)=D(X(s))r{\bf{1}}\,ds+D(X(s))\sigma\,dB(s)$
where ${\bf{1}}$ denotes the $m$ dimensional column vector with all entries equal to $1.$

If one deposits a unit amount into a geometric mean CFMM, then the value of this deposit at time $s\ge0$ is
\begin{align}
	Y(s):=\Big(\frac{X_1(s)}{X_1(0)}\Big)^{p_1}\Big(\frac{X_2(s)}{X_2(0)}\Big)^{p_2}\cdots\Big(\frac{X_m(s)}{X_m(0)}\Big)^{p_m}  
\end{align}
where $p_1,\cdots,p_m$ are positive constants with 
$\sum_{i=1}^mp_i=1.$ For simplicity, we assume $X_1(0)=\cdots=X_m(0)=1$.  Imagine that an issuer designs a funding rate to make the perpetual future price aligned with this value. Define 
$\varphi(s,\gamma):=\gamma_1^{p_1}(s)\gamma_2^{p_2}(s)\cdots \gamma_m^{p_m}(s)$
then the funding rate
$  \Phi(s,\gamma,y):=H(\varphi(s,\gamma),y)-(r-\kappa)\varphi(s,\gamma)+ry  
$
with
$\kappa:=\frac{1}{2}\sum_{i=1}^mp_i|\sigma_i|^2-\frac{1}{2}|\sum_{i=1}^mp_i\sigma_i|^2$ yields the desired perpetual future if $H$ satisfies Assumption \ref{Hypo:H} for the constant  $\ell>\inf_{K>0}(K+\frac{1}{2}(\frac{r}{\sqrt{2K}}+M_{2}C_3)^2)$.
The perpetual future price derived from this funding rate function $\Phi$ coincides with $Y$
and the replicating portfolio is
$\phi =(p_1X_1^{p_1-1}X_2^{p_2}\cdots X_m^{p_m},\cdots,p_mX_1^{p_1}X_2^{p_2}\cdots X_m^{p_m-1}).$  
As in the previous sections, the funding rate functionals $\Phi$ and $\Phi^\delta$, with the parameters $\ell$ and $\delta$ satisfying Assumption \ref{hypo: BSDE delay},  
yield similar perpetual futures prices $Y$ and $Y^\delta$, as well as similar replicating portfolios $Z$ and $Z^\delta$.

Most results of our paper cannot be applied directly to prove these
because
the function $\varphi$ is not differentiable at the origin nor the partial derivative  has polynomial growth at the origin, thus $\varphi$ is not in $C_{p}^{1,2}([0,\infty)\times \Lambda)$.
To avoid this, we detour the strategy.
The main idea is to construct a wealth process useful for this analysis. Observe that  
$
Y(s)=\prod_{i=1}^m X_i^{p_i}(s)
=\prod_{i=1}^me^{(rp_i-\frac{1}{2}p_i|\sigma_i|^2)s+p_i\sigma_iB_s}=e^{-\kappa s}e^{rs-\frac{1}{2}|\Sigma|^2s+\Sigma B_s} 
$
for $s\ge0$
where $\Sigma:=p\sigma=\sum_{i=1}^mp_i\sigma_i\in \mathbb{R}^{1\times m}$
and $p:=(p_1,\cdots,p_m)\in \mathbb{R}^{1\times m}.$
Define
$\hat{B}(s):=\frac{\Sigma}{|\Sigma|} B(s)$
and
$\hat{X}(s):=e^{rs-\frac{1}{2}|\Sigma|^2s+|\Sigma| \hat{B}(s)}$, then $Y(s)=e^{-\kappa s}\hat{X}(s)$, $\hat{B}$ is a one-dimensional Brownian motion and   
$d\hat{X}(s)=r\hat{X}(s)\,ds+|\Sigma| \hat{X}(s)\,d\hat{B}(s)$ for $s\ge0.$ 
It can be shown that the process $\hat{X}$ is the wealth process of the self-financing portfolio 
{ \begin{equation}
		\begin{aligned}
			\hat{\pi}(s)&:=
			\hat{X}(s)pD^{-1}(X(s))=\hat{X}(s)D^{-1}(X(s))p^\top\\ 
			&=e^{\kappa s}(p_1X_1^{p_1-1}(s)X_2^{p_2}(s)\cdots X_m^{p_m}(s),\cdots,p_mX_1^{p_1}(s)X_2^{p_2}(s)\cdots X_m^{p_m-1}(s))\\
			&=e^{\kappa s}\Big(p_1\frac{Y(s)}{X_1(s)},\cdots,p_m\frac{Y(s)}{X_m(s)} \Big)\,.
		\end{aligned}
\end{equation}}
As mentioned in Section \ref{sec:F}, the results of our paper are applicable to wealth process models.
To design a funding rate to make the perpetual future price aligned with $Y$, we define
$\hat{\varphi}(s,\hat{\gamma}):=e^{-\kappa s}\hat{\gamma}(s)$,
then it is evident that $\hat{\varphi}\in C_{1}^{1,2}([0,\infty)\times \Lambda)$.
The funding rate functional
$  \hat{\Phi}(s,\hat{\gamma},y):=H(\hat{\varphi}(s,\hat{\gamma}),y)-(r-\kappa)\hat{\varphi}(s,\hat{\gamma})+ry
$
gives the desired perpetual future price and its replicating portfolio is $\hat{\phi}(s)=e^{-\kappa s}$ for $s\ge0.$
Because 
$\varphi(s,X_s)=\hat{\varphi}(s,\hat{X}_s)$ and $\Phi(s,X_s,Y(s))=\hat{\Phi}(s,\hat{X}_s,Y(s))$ for $s\ge0,$
two infinite-horizon BSDEs derived from $\Phi$ and $\hat{\Phi}$ coincides, thus $\Phi$ and $\hat{\Phi}$ induce the same perpetual future prices. In addition,  
because $\hat{X}$ is the wealth process from $\hat{\pi}$, holding $\hat{\phi}$ number of $\hat{X}$ indicates holding the portfolio $\hat{\phi}\hat{\pi}$, which is equal to $\phi.$


\section{Conclusion}
\label{sec:con}

This study focuses on analyzing the funding rate mechanism for perpetual future contracts traded in cryptocurrency markets. 
Our findings indicate that, through careful design of funding rates, the perpetual futures can be kept consistent with their target values. 
Furthermore, we construct replicating portfolios for perpetual futures, providing issuers with a robust strategy to hedge their exposures. In addition, we introduce path-dependent funding rates suitable for practical implementation and examine the discrepancies between the original and path-dependent funding rates.

There are several potential directions for extending this work. The parameter $\ell$ introduced in Assumption \ref{Hypo:H} \eqref{label: 2-4} must be sufficiently large to ensure the uniqueness of perpetual future prices. In particular, the condition given in \eqref{eqn:ell} is necessary. However, this condition does not represent the optimal lower bound, and the authors believe there is significant room for improvement. Since $\ell$ plays a crucial role in the funding mechanism, identifying tighter lower bounds would be a valuable contribution.
Another possible extension is to explore a broader class of path-dependent funding rates. This paper focuses on a path-dependent version of the constant proportion funding rate, but in practice, a variety of funding rate structures may be employed. Analyzing more general forms of path-dependent funding rates would therefore be an important and relevant direction for future research.


$ $

\appendix

\section{Proof of Theorem \ref{thm: BSDE no delay}}
\label{app:3}

The following proposition is a variant of \cite[Lemma 2.1]{confortola2019backward} tailored to our context.
This proposition will be 
used to prove Theorem \ref{thm: BSDE no delay}.

\begin{prop}  \label{prop:X_esti}
	Let Assumptions \ref{hypo: SDE}-\ref{hypo: SDE_r} hold and $X$ be a solution to \eqref{eqn: path dep SDE}. Then  $X\in\mathbb{S}^p(0,\infty;\mathbb{R}^m)$ for any $p\ge 1,$ and there exist positive constants $L_6,$ which depends only on $p$, and $L_7, L_8,L_9$, which depend only on $C_1,C_r,C_3,p,$ such that
	\begin{align}
		&\mathbb{E}_s[|\!|X|\!|^p_T]\le (L_6|\!|X|\!|^p_s+L_7)e^{L_8(T-s)}\,,\label{eqn: path dep SDE mean_app}\\
		&\mathbb{E}_s[|\!|X-X_s|\!|^p_{s+\delta}]\le L_9(1+\mathbb{E}_s[|\!|X|\!|^p_{s+\delta}]\delta^{\frac{p}{2}}
	\end{align}
	for all $T \in(0,\infty)$, $\delta\in(0,1)$ and $s\in [0,T].$
	
	
\end{prop}

\begin{proof}
	
	We first prove that  $X\in \mathbb{S}^p(0,\infty;\mathbb{R}^m)$ for any $p\ge1$. Observe that the SDE
	\begin{align}
		\tilde{X}(s)=x+\int_0^sr(u,X_u)\tilde{X}(u)\,du+\int_0^s\sigma(u,\tilde{X}_u)\,dB(u)\,, \,\,s\ge0\label{NewSDE}
	\end{align}
	has a solution $\tilde{X}$ in $\mathbb{S}^p(0,\infty;\mathbb{R}^m)$ by \cite[Theorem 7, Chapter 5]{protter}. 
	By showing $X=\tilde{X}$, we conclude that  $X\in \mathbb{S}^p(0,\infty;\mathbb{R}^m)$. For each $n\in \mathbb{N}$, define a stopping time 
	\begin{align}
		\tau_n=\inf\{s\ge0 \, |\, |X(s)|\ge n \text{ or }|\tilde{X}(s)|\ge n\,\}\,,
	\end{align}
	and let $\hat{X}:=X-\tilde{X}$. 
	Because both $X$ and $\tilde{X}$ are solutions to \eqref{NewSDE}, 
	we have 
	\begin{align}
		\hat{X}(s\wedge \tau_n)=\int_0^{s\wedge\tau_n}r(u,X_u)\hat{X}(u)\,du+\int_0^{s\wedge \tau_n}\sigma(u,X_u)-\sigma(u,\tilde{X}_u)\,dB(u)\,.
	\end{align}
	From the BDG inequality and Jensen's inequality, it follows that
	\begin{equation} 
		\begin{aligned}
			&\quad\; \mathbb{E}[\sup_{0\le r\le s}|\hat{X}(r\wedge \tau_n)|^2]\\
			&\le 2\mathbb{E}\Big[\sup_{0\le r\le s} \Big|\int_0^{r\wedge\tau_n}r(u,X_u)\hat{X}(u)\,du \Big|^2+\sup_{0\le r\le s}\Big|\int_0^{r\wedge\tau_n}\sigma(u,\tilde{X}_u)-\sigma(u,X_u)\,dB(u)\Big|^2\Big]\\
			&\le 
			L\Big(\mathbb{E}\Big[\Big(\int_0^s \sup_{0\le r\le u}|\hat{X}(r\wedge \tau_n)|\,du\Big)^2\Big] +\mathbb{E}\Big[\int_0^T \sup_{0\le r\le u}|\hat{X}(r\wedge \tau_n)|^2\,du\Big]\Big)\\
			&\le L(1+s)\mathbb{E}\Big[\int_0^s \sup_{0\le r\le u}|\hat{X}(r\wedge \tau_n)|^2\,du\Big]\,
		\end{aligned}
	\end{equation}
	for some positive constant $L$.
	Define a function $\Psi:[0,\infty)\to [0,\infty)$ as $$\Psi(s):=\mathbb{E}[\sup_{0\le r\le s}|\hat{X}(r\wedge \tau_n)|^2]\,,\;s\ge0\,.$$  Applying Gr\"{o}nwall's inequality to $\Psi$, we have $\Psi(s)=0$ for all $s\ge 0$, which implies
	$X(s\wedge\tau_n)=\tilde{X}(s\wedge\tau_n)$ for all $s\ge 0$.
	Letting $n\to \infty$, we obtain $X=\tilde{X}$ and thus $X\in \mathbb{S}^p(0,\infty;\mathbb{R}^m).$

	For the first inequality in \eqref{eqn: path dep SDE mean_app}, we prove it for $p\ge 2.$ The case with   $0<p<2$ is directly obtained by Jensen's inequality.
	Using 
	\begin{align}
		X(r)=X(s)+\int_s^r r(u,X_u)X(u)\,du+\int_s^r\sigma(u,X_u)\,dB(u)\,,\;0\le s\le r\,,
	\end{align}
	one can easily show that   for $A\in \mathcal{F}_s$,
	\begin{align}
		|\!|X|\!|_T \boldsymbol{1}_A
		\le|\!|X|\!|_s \boldsymbol{1}_A
		+\sup_{r\in[s,T]}\Big|\int_s^r r(u,X_u)X(u)\,du\Big|  \boldsymbol{1}_A+\sup_{r\in[s,T]}\Big|\int_s^r\sigma(u,X_u)\,dB(u)\Big| \boldsymbol{1}_A\,.
	\end{align}
	According to the Minkowski inequality and the BDG inequality, it follows that
	\begin{align}
		(\mathbb{E}[|\!|X|\!|^p_T \boldsymbol{1}_A])^{\frac{1}{p}} &\le(\mathbb{E}[|\!|X|\!|^p_s \boldsymbol{1}_A])^{\frac{1}{p}}+\Big(\mathbb{E}\Big[\Big(\int_s^T C_r|\!|X|\!|_u\,du \boldsymbol{1}_A\Big)^p\Big]\Big)^{\frac{1}{p}}\\
		&\quad+M_p\Big(\mathbb{E}\Big[\Big(\int_s^T(C_1+C_3|\!|X|\!|_u)^2\,du\Big)^{\frac{p}{2}} \boldsymbol{1}_A\Big]\Big)^{\frac{1}{p}}\\
		& \le(\mathbb{E}[|\!|X|\!|^p_s \boldsymbol{1}_A])^{\frac{1}{p}} +C_1(\mathbb{E}[\boldsymbol{1}_A])^{\frac{1}{p}}M_p(T-s)^{\frac{1}{2}}\\
		&\quad +\int_s^T(\mathbb{E}[(C_r|\!|X|\!|_u\boldsymbol{1}_A)^p])^{\frac{1}{p}}\,du+M_p\Big(\int_s^T(\mathbb{E}[(C_3|\!|X_u|\!|_u\boldsymbol{1}_A)^p])^{\frac{2}{p}} \,du\Big)^{\frac{1}{2}} 
	\end{align}
	where $M_p$ is the constant from the BDG inequality.
	For $K>0$, by multiplying $e^{-K(T-s)}$, we have
	\begin{align}
		e^{-K(T-s)}(\mathbb{E}[|\!|X|\!|^p_T\boldsymbol{1}_A])^{\frac{1}{p}} &\le(\mathbb{E}[|\!|X|\!|^p_s\boldsymbol{1}_A])^{\frac{1}{p}} +C_1(\mathbb{E}[\boldsymbol{1}_A])^{\frac{1}{p}}M_p(T-s)^{\frac{1}{2}}\\
		&\quad +\int_s^Te^{-K(T-s)}(\mathbb{E}[(C_r|\!|X|\!|_u\boldsymbol{1}_A)^p])^{\frac{1}{p}}\,du\\
		&\quad+M_p\Big(\int_s^Te^{-2K(T-s)}(\mathbb{E}[(C_3|\!|X_u|\!|_u\boldsymbol{1}_A)^p])^{\frac{2}{p}} \,du\Big)^{\frac{1}{2}}\,.
	\end{align}
	Observe that 
	\begin{align}
		\int_s^Te^{-K(T-s)}(\mathbb{E}[(|\!|X|\!|_u\boldsymbol{1}_A)^p])^{\frac{1}{p}}\,du&\le \int_s^Te^{-K(T-u)}e^{-K(u-s)}(\mathbb{E}[(|\!|X|\!|_u\boldsymbol{1}_A)^p])^{\frac{1}{p}}\,du\\
		&\le \frac{1}{\sqrt{2K}}\Big(\int_s^T(e^{-K(u-s)}(\mathbb{E}[(|\!|X|\!|_u\boldsymbol{1}_A)^p])^{\frac{1}{p}})^2\,du\Big)^{\frac{1}{2}}
	\end{align}
	and 
	\begin{align}
		\Big(\int_s^Te^{-2K(T-s)}(\mathbb{E}[(|\!|X_u|\!|_u\boldsymbol{1}_A)^p])^{\frac{2}{p}} \,du\Big)^{\frac{1}{2}}\le \Big(\int_s^T(e^{-K(u-s)}(\mathbb{E}[(|\!|X_u|\!|_u\boldsymbol{1}_A)^p])^{\frac{1}{p}})^2 \,du\Big)^{\frac{1}{2}}\,.
	\end{align}
	Thus,
	\begin{align}
		e^{-K(T-s)}(\mathbb{E}[|\!|X|\!|^p_T\boldsymbol{1}_A])^{\frac{1}{p}} &\le(\mathbb{E}[|\!|X|\!|^p_s\boldsymbol{1}_A])^{\frac{1}{p}} +C_1(\mathbb{E}[\boldsymbol{1}_A])^{\frac{1}{p}}M_p(1+(T-s))
		\\
		&\quad +\Big(\frac{C_r}{\sqrt{2K}}+M_pC_3\Big)\Big(\int_s^T(e^{-K(u-s)}(\mathbb{E}[(|\!|X|\!|_u\boldsymbol{1}_A)^p])^{\frac{1}{p}}\,)^2du\Big)^{\frac{1}{2}}\,.
	\end{align}
	Then, \cite[Corollary 2]{butler1971generalization} yields
	\begin{align}
		&e^{-K(T-s)}(\mathbb{E}[|\!|X|\!|^p_T\boldsymbol{1}_A])^{\frac{1}{p}}\\ \le&\;\big((\mathbb{E}[|\!|X|\!|^p_s\boldsymbol{1}_A])^{\frac{1}{p}} +C_1(\mathbb{E}[\boldsymbol{1}_A])^{\frac{1}{p}}M_p(1+(T-s)) \big)e^{1+\frac{1}{2}(\frac{C_r}{\sqrt{2K}}+M_pC_3)^2(T-s)}\,.
	\end{align}
	Using the inequality $T-s\le\frac{1}{\epsilon}e^{\epsilon(T-s)}$ for any $\epsilon>0$,
	we obtain
	\begin{align}\label{eqn:bb}
		(\mathbb{E}[|\!|X|\!|^p_T\boldsymbol{1}_A])^{\frac{1}{p}} &\le  ((\mathbb{E}[|\!|X|\!|^p_s\boldsymbol{1}_A])^{\frac{1}{p}}e +L(\mathbb{E}[\boldsymbol{1}_A])^{\frac{1}{p}})e^{(\epsilon+K+\frac{1}{2}(\frac{C_r}{\sqrt{2K}}+M_pC_3)^2)(T-s)}\,
	\end{align} 
	for some positive constant $L$. 
	Because this holds for all $A\in \mathcal{F}_s,$ it follows that
	\begin{align}\label{eqn:X esi}
		\mathbb{E}_s[|\!|X|\!|_T^p]\le (L_6|\!|X|\!|^{p}_s+L_7)e^{L_8(T-s)} 
	\end{align} 
	for positive constants $L_6,$ which  depends only on $p$, and $L_7$ and $L_8,$ which depend only on $C_1,C_r,C_3,p.$

	Now we prove the second inequality. From the BDG inequality and Jensen's inequality, we have 
	\begin{align}
		&\;\quad\mathbb{E}_s[|\!|X-X_s|\!|^p_{s+\delta}]=\mathbb{E}_s\Big[\sup_{s\le r\le s+\delta}|X(r)-X(s)|^p\Big]\\
		&\le L\mathbb{E}_s\Big[\Big(\int_s^{s+\delta}|r(u,X_u)X(u)|\,du\Big)^p+\sup_{s\le r\le s+\delta}\Big|\int_s^r\sigma(u,X_u)\,dB(u)\Big|^p\Big]\\
		&\le L\Big(\mathbb{E}_s\Big[\Big(\int_s^{s+\delta}|r(u,X_u)X(u)|\,du\Big)^p\Big]+\mathbb{E}_s\Big[\sup_{s\le r\le s+\delta}\Big(\int_s^r|\sigma(u,X_u)|^2\,du\Big)^{\frac{p}{2}}\Big]\Big) \\
		&\le L\mathbb{E}_s\Big[\delta^{p-1}\int_s^{s+\delta}|X(u)|^p\,du+\delta^{\frac{p}{2}-1}\int_s^{s+\delta}|\sigma(u,X_u)|^p\,du\Big] \\
		&\le L(1+\mathbb{E}_s[|\!|X|\!|^p_{s+\delta}]\delta^{\frac{p}{2}}
	\end{align}
	for some positive constant $L$, which depends only on $C_1,C_r,C_3,p$ and may change line by line. 
	This completes the proof. 
\end{proof}

We now prove   Theorem \ref{thm: BSDE no delay}.

\begin{proof}	
	In this proof, $L$ denotes a generic constant depending only on $C_1,C_r, C_3,C_4,\ell,\rho$ and may differ line by line.
	We first prove the uniqueness of solutions. Suppose there are two solutions $(Y^1,Z^1)$ and $(Y^2,Z^2)$ to \eqref{eqn:BSDE}. Define three processes
	\begin{align}\label{eqn:aa}
		&\hat{Y}(s)=Y^1(s)-Y^2(s)\,,\\
		&\hat{Z}(s)=Z^1(s)-Z^2(s)\,,\\
		&\alpha(s)=\frac{f(s,X_s,Y^1(s))-f(s,X_s,Y^2(s))}{\hat{Y}(s)}\boldsymbol{1}_{\{|\hat{Y}(s)|>0\}}-\ell\boldsymbol{1}_{\{|\hat{Y}(s)|=0\}}\,.
	\end{align}
	It can be easily shown that $\alpha(s)\le -\ell$ and 
	\begin{align}
		\hat{Y}(s)=\hat{Y}(T)+\int_s^T\alpha(u)\hat{Y}(u)\,du-\int_s^T \hat{Z}(u)\,dB(u)\,.
	\end{align}
	By Proposition \ref{prop:X_esti}, 
	there are constants $L_6$ and $L_7$, depending only on $C_1,C_r,C_3,\rho$, such that  \begin{align}
		\mathbb{E}_s[|\!|X|\!|_T^\rho]\le (L_6|\!|X|\!|_s^\rho+L_7)e^{(\epsilon+K+\frac{1}{2}(\frac{C_r}{\sqrt{2K}}+M_{\rho\vee 2}C_3)^2)\rho(T-s)}
	\end{align}
	for any $\epsilon>0$, $K>0$, $T>0$ and $s\in [0,T].$
	As $\ell>\inf_{K>0}\{(K+\frac{1}{2}(\frac{C_r}{\sqrt{2K}}+M_{\rho\vee 2}C_3)^2)\rho\}$,  for some constant $L_8$ with $0<L_8<\ell$, we have 
	\begin{align}\label{eqn: lemma}
		\mathbb{E}_s[|\!|X|\!|^\rho_T]\le (L_6|\!|X|\!|^\rho_s+L_7)e^{L_8(T-s)}
	\end{align}
	for any   $T>0$ and $s\in [0,T].$ 
	It follows that
	\begin{align}
		|\hat{Y}(s)|=\Big|\mathbb{E}_s\Big[e^{\int_s^T\alpha(u)\,du}\hat{Y}(T)\Big]\Big|&\le \mathbb{E}_s\Big[e^{-\ell(T-s)}L(1+|\!|X|\!|^\rho_T)\Big]\\
		&\le Le^{-(\ell-L_8)(T-s)}(|\!|X|\!|^\rho_s+1)\,.
	\end{align}
	Letting $T\to \infty$, we have $Y^1-Y^2=\hat{Y}=0$. Moreover, this directly yields $Z^1=Z^2$.
	
	Now we prove the existence of solutions. For each $n\in \mathbb{N}$, there exists a unique solution $(Y^n(s),Z^n(s))_{0\le s\le n}$ to the BSDE
	\begin{align}\label{eqn: n path dep BSDE}
		Y^n(s)=\int_s^nf(u,X_u,Y^n(u))\,du-\int_s^nZ^n(u)\,dB(u)\,.
	\end{align}
	We extend this solution $(Y^n(s),Z^n(s))_{0\le s\le n}$ to $(Y^n(s),Z^n(s))_{0\le s<\infty}$ by defining 
	$       Y^n(s)=Z^n(s)=0$ for all $s>n$.
	Then 
	\begin{align}
		Y^n(s)=Y^n(T)+\int_s^Tf(u,X_u,Y^n(u))-\boldsymbol{1}_{\{u>n\}}f(u,X_u,0)\,du-\int_s^TZ^n(u)\,dB(u)\,.
	\end{align}
	For $m>n$, define three processes
	\begin{align}
		&\tilde{Y}(s)=Y^m(s)-Y^n(s)\,,\\
		&\tilde{Z}(s)=Z^m(s)-Z^n(s)\,,\\
		&\tilde{\alpha}(s)=\frac{f(s,X_s,Y^m(s))-f(s,X_s,Y^n(s))}{\tilde{Y}(s)}\boldsymbol{1}_{\{|\tilde{Y}(s)|>0\}}-\ell\boldsymbol{1}_{\{|\tilde{Y}(s)|=0\}}\,.
	\end{align}
	It can be easily checked that $\tilde{\alpha}\le -\ell$ and
	\begin{align}
		\tilde{Y}(s)=\int_s^m\tilde{\alpha}(u)\tilde{Y}(u)+\boldsymbol{1}_{\{u>n\}}f(u,X_u,0)\,du-\int_s^m\tilde{Z}(u)\,dB(u)\,.
	\end{align}
	Then It\^o's formula and \eqref{eqn: lemma} yield
	\begin{align}
		|\tilde{Y}(s)|&=\mathbb{E}_s\Big[\int_s^me^{\int_s^u\tilde{\alpha}(v)\,dv} \boldsymbol{1}_{\{u>n\}}f(u,X_u,0)\,du\Big]\\
		&\le\mathbb{E}_s\Big[\int_n^mC_4e^{-\ell(u-s)}(1+|\!|X|\!|^\rho_u) \,du\Big]\\
		&\le L(1+|\!|X|\!|^\rho_s)(e^{-(\ell-L_8)(n-s)}-e^{-(\ell-L_8)(m-s)})\,.
	\end{align}
	Therefore for $0\le T\le n\le m$,
	\begin{align}
		\lim_{n,m\to \infty}\mathbb{E}[\sup_{0\le s\le T}|Y^n(s)-Y^m(s)|^2]=0\,.
	\end{align}
	The sequence $(Y^n)_{n\in\mathbb{N}}$  is a Cauchy sequence in $\mathbb{S}^2(0,T;\mathbb{R})$ for each $T>0$. Because $\mathbb{S}^2(0,T;\mathbb{R})$ is complete, the limit $Y:=\lim_{n\to \infty}Y^n$ exists. 
	Applying It\^o's formula, we have
	\begin{align}
		|\tilde{Y}(0)|^2+\mathbb{E}\Big[\int_0^T|\tilde{Z}(u)|^2\,du\Big]=\mathbb{E}\Big[|\tilde{Y}(T)|^2+\int_0^T2\tilde{\alpha}(u)|\tilde{Y}(u)|^2\,du\Big]\,.
	\end{align}
	From the inequality $\tilde{\alpha}(s)\le -\ell$, 
	\begin{align}
		\mathbb{E}\Big[\int_0^T|Z^m(u)-Z^n(u)|^2\,du\Big]&= \mathbb{E}\Big[\int_0^T|\tilde{Z}(u)|^2\,du\Big]=\mathbb{E}[|\tilde{Y}(T)|^2]\\
		&\le L(1+\mathbb{E}[|\!|X|\!|^{2\rho}_T])(e^{-(\ell-L_8)(n-s)}-e^{-(\ell-L_8)(m-s)})^2\,.
	\end{align}  
	Thus, $(Z^n)_{n\in\mathbb{N}}$ is a Cauchy sequence in 	$\mathbb{H}^2(0,T;\mathbb{R}^m)$. The limit $Z:=\lim_{n\to \infty}Z^n$ exists in 	$\mathbb{H}^2(0,T;\mathbb{R}^m)$.  Because $(Y^n(s),Z^n(s))_{0\le s\le n}$ satisfies 
	\begin{align}
		Y^n(s)=Y^n(T)+\int_s^Tf(u,X_u,Y^n(u))\,du-\int_s^TZ^n(u)\,dB(u)\,,\;0\le s\le T,
	\end{align}
	the Lebesgue dominated convergence theorem implies that 
	the pair $(Y,Z)$ is a solution to \eqref{eqn:BSDE}.
	
	Now we prove that there exists a constant $L>0$ such that
	$|\!|Y|\!|\le L(1+|\!|X|\!|^\rho)$.
	As $Y=\lim_{n\to \infty}Y^n$, it suffices to prove   $|\!|Y^n|\!|\le L(1+|\!|X|\!|^\rho)$. Define
	\begin{align}
		\overline{\alpha}(s):=\frac{f(s,X_s,Y^n(s))-f(s,X_s,0)}{{Y}^n(s)}\boldsymbol{1}_{\{|{Y}^n(s)|>0\}}-\ell\boldsymbol{1}_{\{|{Y}^n(s)|=0\}}\,,
	\end{align}
	then   
	\begin{align}
		Y^n(s)=\int_s^n(\overline{\alpha}(u)Y^n(u)+f(u,X_u,0))\,du-\int_s^n Z^n(u)\,dB(u)\,,\;0\le s\le n\,.
	\end{align}
	Using It\^o's formula, \eqref{eqn: lemma} and the inequality $\overline{\alpha}(s)<-\ell$, we have
	\begin{align}\label{eqn:finite_esti}
		|Y^n(s)|\le \mathbb{E}_s\Big[\int_s^ne^{-\ell(u-s)}|f(u,X_u,0)|\,du \Big]\le L(1+|\!|X|\!|^{\rho}_s)\,.
	\end{align}
	This completes the proof.
\end{proof}

\section{Feynman-Kac formula}
\label{app:FK}

We present 
the notions of path-dependent PDEs (PPDEs) and provide the proof of the Feynman-Kac formula stated in Theorem \ref{thm: fey}.
For
non-anticipative functionals $r:[0,\infty)\times \Lambda\to \mathbb{R}$, $\sigma:[0,\infty)\times \Lambda \to \mathbb{R}^{m\times m}$ and   \(f: [0, \infty) \times \Lambda \times \mathbb{R} \to \mathbb{R} \), 
consider the PPDE
\begin{align}
	&-\partial_s \varphi(s,\gamma)-\frac{1}{2}\textnormal{tr}(\sigma\sigma^{\top}\partial_{xx}\varphi)(s,\gamma)- r(s,\gamma)\partial_{x}\varphi(s,\gamma)\gamma(s)-f(s,\gamma,\varphi(s,\gamma))=0\label{eqn: path dependent pde}
\end{align}
for  $(s,\gamma)\in [0,\infty)\times \Lambda$.
The definition of classical solutions to PPDEs is as follows. Classical solutions are often referred to simply as solutions.

\begin{definition}
	Let $\varphi\in C^{1,2}_p([0,\infty)\times \Lambda)$. 
	\begin{enumerate}[(i)]
		\item We say $\varphi$ is a classical subsolution (supersolution, respectively) to the PPDE \eqref{eqn: path dependent pde} if
		\begin{equation} 
			\begin{aligned}
				-\partial_s \varphi(s,\gamma)-\frac{1}{2}\textnormal{tr}(\sigma\sigma^{\top}\partial_{xx}\varphi)(s,\gamma)- r(s,\gamma)\partial_{x}\varphi(s,\gamma)\gamma(s)&-f(s,\gamma,\varphi(s,\gamma))\le  0\\
				&\quad\;\; (\ge 0,\text{ respectively}) 
			\end{aligned}   
		\end{equation}
		for all $(s,\gamma)\in [0,\infty)\times \Lambda$.
		\item  We say $\varphi$ is a classical solution to the  PPDE \eqref{eqn: path dependent pde} if $\varphi$ is both a classical subsolution and classical supersolution. 
	\end{enumerate}
\end{definition}

For $s\ge0$, define $\Lambda_s := C([0,s];\mathbb{R}^m)$ and $\Lambda^s_{s+1}:=C([s,s+1];\mathbb{R}^m)$ equipped with the supremum norm. Let $(\mathcal{F}_u^s)_{u\ge 0}$ be the filtration generated by $((B(u)-B(s))\boldsymbol{1}_{u\ge s})_{u\ge 0}.$
Denote as $\mathcal{T}^s_{s+1,+}$ the set of $(\mathcal{F}_u^s)_{0\le u\le s+1}$-stopping times such that $\tau>s$.
For  $(s,\gamma)\in [0,\infty)\times \Lambda$ and $\varphi\in C_p([0,\infty)\times\Lambda)$, we define the spaces
\begin{align}
	\underline{\mathcal{A}}\varphi(s,\gamma):=\{\psi\in& C^{1,2}_p([0,\infty)\times \Lambda)\,|\, \text{there exists }\tau\in \mathcal{T}_{s+1,+}^s\\
	&\text{ such that }0=\psi(s,\gamma)-\varphi(s,\gamma)
	=\min_{\tilde{\tau}\in \mathcal{T}_{s+1}^s}\mathbb{E}[(\psi-\varphi)(\tau\wedge\tilde{\tau},X^{s,\gamma})]\}\,,\\
	\overline{\mathcal{A}}\varphi(s,\gamma):=\{\psi\in &C^{1,2}_p([0,\infty)\times \Lambda)\,|\, \text{there exists }\tau\in \mathcal{T}_{s+1,+}^s\\
	&\text{ such that } 0=\psi(s,\gamma)-\varphi(s,\gamma)
	=\max_{\tilde{\tau}\in \mathcal{T}_{s+1}^s}\mathbb{E}[(\psi-\varphi)(\tau\wedge\tilde{\tau},X^{s,\gamma})]\}\,,
\end{align}
where   $X^{s,\gamma}$ is a solution to the SDE \eqref{eqn: fey sde}.

\begin{definition}
	Let $\varphi\in C_p([0,\infty)\times\Lambda)$. 
	\begin{enumerate}[(i)]
		\item We say $\varphi$ is a viscosity subsolution (superslution, respectively) to the PPDE \eqref{eqn: path dependent pde} if 
		\begin{equation} 
			\begin{aligned}
				-\partial_s \psi(s,\gamma)-\frac{1}{2}\textnormal{tr}(\sigma\sigma^{\top}\partial_{xx}\psi)(s,\gamma)- r(s,\gamma)\partial_{x}\psi(s,\gamma)\gamma(s)
				&-f(s,\gamma,\psi(s,\gamma))\le 0 \\
				&\quad\;\;(\ge 0\,, \text{respectively}) 
			\end{aligned}
		\end{equation}
		for all $(s,\gamma)\in [0,\infty)\times \Lambda$ and $\psi\in \underline{\mathcal{A}}\varphi(s,\gamma)$ ($\psi\in \overline{\mathcal{A}}\varphi(s,\gamma)$, respectively).
		\item We say $\varphi$ is a viscosity solution to the PPDE \eqref{eqn: path dependent pde} if $\varphi$ is both a viscosity subsolution and viscosity supersolution.
	\end{enumerate}
\end{definition}

\begin{theorem}\label{thm : clasic vis}
	Suppose that the functionals $r, \sigma,f$ are continuous and  $\varphi\in C^{1,2}_p([0,\infty)\times\Lambda)$. Then $\varphi$ is a classical subsolution (supersolution, respectively) to the PPDE \eqref{eqn: path dependent pde} if and only if  $\varphi$ is a viscosity subsolution (supersolution, respectively) to the PPDE \eqref{eqn: path dependent pde}. 
\end{theorem}

\begin{proof}
	We only prove the subsolution property.
	Assume that $\varphi$ is a viscosity subsolution to \eqref{eqn: path dependent pde}. Since   $\varphi\in \underline{\mathcal{A}}\varphi(s,\gamma)$ for $(s,\gamma)\in [0,\infty)\times\Lambda$, it follows  that $\varphi$ is a classical subsolution.
	
	Now we show that if $\varphi$ is a classical subsolution then $\varphi$ is a viscosity subsolution. Suppose $\varphi$ is not a viscosity subsolution. Then, there exists $(s,\gamma)\in [0,\infty)\times \Lambda$ and $\psi\in \underline{\mathcal{A}}\varphi(s,\gamma)$ such that 
	\begin{align}
		-M:=\partial_s \psi(s,\gamma)+\frac{1}{2}\textnormal{tr}(\sigma\sigma^{\top}\partial_{xx}\psi)(s,\gamma)+ r(s,\gamma)\partial_{x}\psi(s,\gamma)\gamma(s)+f(s,\gamma,\psi(s,\gamma))<0\,.
	\end{align}
	From the  definition of $\underline{\mathcal{A}}\varphi(s,\gamma)$, there exists $\tau\in \mathcal{T}^s_{s+1,+}$ such that
	\begin{align}
		0=\psi(s,\gamma)-\varphi(s,\gamma)=\min_{\tilde{\tau}\in \mathcal{T}_{s+1}^s}\mathbb{E}[(\psi-\varphi)(\tau\wedge\tilde{\tau},X^{s,\gamma})]\,.  
	\end{align}
	Define six processes 
	\begin{align}
		&Y^1(v)=\psi(v,X_v^{s,\gamma})\,,\,\,Z^1(v)=\sigma^{\top}(v,X_v^{s,\gamma})\partial_{x}\psi(v,X_v^{s,\gamma}),\,\\
		&Y^2(v)=\varphi(v,X_v^{s,\gamma})\,,\,\,Z^2(v)=\sigma^{\top}(v,X_v^{s,\gamma})\partial_{x}\varphi(v,X_v^{s,\gamma})\,, \,\\ 
		&\hat{Y}(v)=Y^1(v)-Y^{s,\gamma}(v)\,,\,\,\hat{Z}(v)=Z^1(v)-Z^{s,\gamma}(v)
	\end{align}
	for $s\le v\le s+1$
	and a stopping time 
	\begin{align}
		\hat{\tau}:=(s+1)\wedge\tau\wedge\inf\Big\{v>s\,\Big|\, &\partial_v \psi(v,X_v^{s,\gamma})+\frac{1}{2}\textnormal{tr}(\sigma\sigma^{\top}\partial_{xx}\psi)(v,X_v^{s,\gamma})\\&+ r(v,X_v^{s,\gamma})\partial_{x}\psi(v,X_v^{s,\gamma})X^{s,\gamma}(v)+f(v,X_v^{s,\gamma},Y^2(v))>-\frac{M}{2}\Big\}\,.
	\end{align}
	By definition,  we have $\hat{\tau}\in \mathcal{T}^s_{s+1,+}$ since $\varphi,\psi\in C^{1,2}_p([0,\infty)\times\Lambda).$
	By Theorem \ref{thm: functional ito}, we have
	\begin{align}
		Y^1(s)=\psi(s+1,X^{s,\gamma}_{s+1})-\int_s^{s+1}&\partial_u\psi(u,X_u^{s,\gamma})+\frac{1}{2}\textnormal{tr}(\sigma\sigma^{\top}\partial_{xx}\psi)(u,X_u^{s,\gamma})\\
		&+r(u,X_u^{s,\gamma})\partial_{x}\psi(u,X_u^{s,\gamma})X^{s,\gamma}(u)\,du-\int_s^{s+1}Z^1(u)\,dB(u)\,,\\
		Y^2(s)=\varphi(s+1,X^{s,\gamma}_{s+1})-\int_s^{s+1}&\partial_u\varphi (u,X_u^{s,\gamma})+\frac{1}{2}\textnormal{tr}(\sigma\sigma^{\top}\partial_{xx}\varphi)(u,X_u^{s,\gamma})\\&
		+r(u,X_u^{s,\gamma})\partial_{x}\varphi(u,X_u^{s,\gamma})X^{s,\gamma}(u)\,du-\int_s^{s+1}Z^2(u)\,dB(u)\,.
	\end{align}     
	Then
	\begin{align}
		\hat{Y}(\hat{\tau})&= \int_s^{\hat{\tau}}\partial_u \psi(u,X_u^{s,\gamma})+\frac{1}{2}\textnormal{tr}(\sigma\sigma^{\top}\partial_{xx}\psi)(u,X_u^{s,\gamma})+ r(u,X^{s,\gamma}_u)\partial_{x}\psi(u,X_u^{s,\gamma})X^{s,\gamma}(u)\,du\\
		&-\int_s^{\hat{\tau}}\partial_u \varphi(u,X_u^{s,\gamma})+\frac{1}{2}\textnormal{tr}(\sigma\sigma^{\top}\partial_{xx}\varphi)(u,X_u^{s,\gamma})+ r(u,X_u^{s,\gamma})\partial_{x}\varphi(u,X_u^{s,\gamma})X^{s,\gamma}(u)\,du\\&+\int_s^{\hat{\tau}}\hat{Z}(u)\,dB(u)\,.
	\end{align}
	We have
	\begin{align}
		&\partial_v \psi(v,X_v^{s,\gamma})+\frac{1}{2}\textnormal{tr}(\sigma\sigma^{\top}\partial_{xx}\psi)(v,X_v^{s,\gamma})+ r(v,X^{s,\gamma}_v)\partial_{x}\psi(v,X_v^{s,\gamma})X^{s,\gamma}(v)\\
		&+f(v,X_v^{s,\gamma},Y^2(v))\le-\frac{M}{2} \,\label{eqn:4}
	\end{align}
	for $v\in[s,\hat{\tau}]$ from the definition of $\tilde{\tau}$.
	Since $\varphi$ is a classical solution to \eqref{eqn: path dependent pde},
	\begin{align}
		&\partial_v \varphi(v,X_v^{s,\gamma})+\frac{1}{2}\textnormal{tr}(\sigma\sigma^{\top}\partial_{xx}\varphi)(v,X_v^{s,\gamma})+ r(v,X_v^{s,\gamma})\partial_{x}\varphi(v,X_v^{s,\gamma})X^{s,\gamma}(v)\\
		&+f(v,X_v^{s,\gamma},\varphi(v,X_v^{s,\gamma}))\ge 0\,.\label{eqn:5}
	\end{align}
	Using \eqref{eqn:4} and \eqref{eqn:5}, we obtain 
	\begin{align}
		\hat{Y}(\hat{\tau})&\le \int_s^{\hat{\tau}} -\frac{M}{2}\,du+\int_s^{\hat{\tau}}\hat{Z}(u)\,dB(u)\,,
	\end{align}
	which implies that 
	$\mathbb{E}[\hat{Y}(\hat{\tau})]<0.$	This is a contradiction since $\psi \in \underline{\mathcal{A}}\varphi(s,\gamma)$ yields that
	\begin{align}
		0=\min_{\tilde{\tau}\in \mathcal{T}^s}\mathbb{E}[(\psi-\varphi)(\tau\wedge\tilde{\tau},X^{s,\gamma})]\}\le \mathbb{E}[(\psi-\varphi) (\tau\wedge\hat{\tau},X^{s,\gamma})]=\mathbb{E}[\hat{Y}(\hat{\tau})]<0\,.
	\end{align}
\end{proof}

We now provide the proof of Theorem \ref{thm: fey}.

\begin{proof}	
	First we show that \eqref{eqn: fey bsde} has a unique solution $(Y^{s,\gamma}(v),Z^{s,\gamma}(v))_{v\ge s}$ such that $|\!|Y^{s,\gamma}|\!|\le L(1+|\!|X^{s,\gamma}|\!|^{\rho})$ for some constant $L>0.$
	The proof follows a similar approach to that of Theorem \ref{thm: BSDE no delay}, so we will only outline the main idea.
	For each $n\in \mathbb{N}$, consider the  BSDE
	\begin{align}
		Y^{n,s,\gamma}(v)=\int_v^n f(u,X^{s,\gamma}_u,Y^{n,s,\gamma}(u))\,du -\int_v^nZ^{n,s,\gamma}(u)\,dB(u)\,,\;s\le v\le n\,.
	\end{align}
	There exist positive constants $L_6,L_7,L_8,L$ such that
	\begin{align}
		&L_8<\ell\,,\\
		&\mathbb{E}_v[|\!|X^{s,\gamma}|\!|^\rho_T]\le (L_6|\!|X^{s,\gamma}|\!|^\rho_v+L_7)e^{L_8(T-v)}\,,\\
		&|Y^{n,s,\gamma}(v)-Y^{m,s,\gamma}(v)|\le L(1+|\!|X^{s,\gamma}|\!|^{\rho}_v)(e^{-(\ell-L_8)(n-v)}-e^{-(\ell-L_8)(m-v)})\label{eqn: 1}
	\end{align}
	for all  $T>0$ and $v\in [s,T].$      
	Using these inequalities, one can verify that a sequence $(Y^{n,s,\gamma})_{n\in \mathbb{N}}$ ($(Z^{n,s,\gamma})_{n\in \mathbb{N}}$, respectively) converges in $\mathbb{S}^2(0,T;\mathbb{R})$ (in $\mathbb{H}^2(0,T;\mathbb{R}^d)$, respectively) for each $T>0$. Then,  $(Y^{s,\gamma},Z^{s,\gamma}):=\lim_{n\to \infty}(Y^{n,s,\gamma},Z^{n,s,\gamma})$ is a unique solution to \eqref{eqn: fey bsde}. This proves \eqref{eqn: fey i}.

	For \eqref{eqn: fey ii}, we verify that $Y^{0,x}(v)=\varphi(v,X_v^{0,x})$ and $\varphi=\varphi(s,\gamma)$ is a viscosity solution to \eqref{eqn: fey pde}. For each $n\in \mathbb{N}$, define $\varphi^n(s,\gamma):=Y^{n,s,\gamma}(s)$. By \cite[Theorem 4.5, Theorem 4.7]{cordoni2020stochastic}, $\varphi^n$ is a continuous function and $Y^{n,0,x}(v)=\varphi^n(v,X_v^{0,x})$. Since 
	the sequence $(\varphi^n)_{n\in\mathbb{N}}$ converges locally uniformly on $[0,\infty)\times \Lambda$ 
	by     
	\eqref{eqn: 1},
	the function  $\varphi:[0,\infty)\times\Lambda\to \mathbb{R}$ given as 
	$\varphi(s,\gamma):=\lim_{n\to\infty}\varphi^n(s,\gamma)$  for $(s,\gamma)\in [0,\infty)\times\Lambda$
	is also  continuous.
	In addition,  for $\omega\in \Omega$
	\begin{align}
		\varphi(v,X_v^{0,x}(\omega))=\lim_{n\to \infty}\varphi^n(v,X_v^{0,x}(\omega))=\lim_{n\to \infty}Y^{n,0,x}(v,\omega)=Y^{0,x}(v,\omega)\,.
	\end{align}  
	Now we will demonstrate that 
	$\varphi$ is a viscosity solution to \eqref{eqn: fey pde}. For our purposes, we will focus on proving that 
	$\varphi$ is a viscosity subsolution, as a similar argument will show that 
	$\varphi$ is also a viscosity supersolution. Assume that 
	$\varphi$ is not a viscosity subsolution. We will derive a contradiction from this assumption.
	There exist $(s,\gamma)\in [0,\infty)\times \Lambda$ and $\psi\in \underline{\mathcal{A}}\varphi(s,\gamma)$ such that 
	\begin{align}
		-M:=\partial_s \psi(s,\gamma)+\frac{1}{2}\textnormal{tr}(\sigma\sigma^{\top}\partial_{xx}\psi)(s,\gamma)+ r(s,\gamma)\partial_{x}\psi(s,\gamma)\gamma(s)+f(s,\gamma,\psi(s,\gamma))<0\,.
	\end{align}
	From the definition of $\underline{\mathcal{A}}\varphi(s,\gamma)$, there exists $\tau\in \mathcal{T}^s_{s+1,+}$ such that  
	\begin{equation}
		\label{eqn:vis} 0=\psi(s,\gamma)-\varphi(s,\gamma)=\min_{\tilde{\tau}\in \mathcal{T}_{s+1}^s}\mathbb{E}[(\psi-\varphi)(\tau\wedge\tilde{\tau},X^{s,\gamma})\,.
	\end{equation}
	Let $C_f$ denote the Lipschitz constant of the function $f$ presented in \eqref{eqn:Lip_f}, i.e.,
	$$
	|f(v, \gamma, y) - f(v, \gamma, y')| \leq C_f |y - y'|
	$$
	for all $(s, \gamma) \in [0, \infty) \times \Lambda$ and $y, y' \in \mathbb{R}$.
	Define four processes 
	\begin{align}
		&Y^1(v):=\psi(v,X_v^{s,\gamma})\,,\,\,Z^1(v):=\sigma^{\top}(v,X_v^{s,\gamma})\partial_{x}\psi(v,X_v^{s,\gamma})\,,\\
		&\hat{Y}(v):=Y^1(v)-Y^{s,\gamma}(v)\,,\,\,\hat{Z}(v):=Z^1(v)-Z^{s,\gamma}(v)  
	\end{align}
	for $s\le v\le s+1$,  
	and a stopping time
	\begin{align}
		\hat{\tau}:=(s+1)\wedge\tau\wedge\inf \Big\{v>s\,\Big|\, &\partial_v \psi(v,X_v^{s,\gamma})+\frac{1}{2}\textnormal{tr}(\sigma\sigma^{\top}\partial_{xx}\psi)(v,X_v^{s,\gamma})\\&+ r(v,X^{s,\gamma}_v)\partial_{x}\psi(v,X_v^{s,\gamma})X^{s,\gamma}(v)+f(v,X_v^{s,\gamma},\psi(v,X_v^{s,\gamma}))\\&+C_f|\hat{Y}(v)|>-\frac{M}{2}\Big\}\,.
	\end{align}
	By definition, it is evident that $\hat{\tau}\in \mathcal{T}^s_{s+1,+}$.
	By Theorem \ref{thm: functional ito}, $(Y^1(v),Z^1(v))_{s\le v\le s+1}$ satisfies
	\begin{align}
		Y^1(s)=\psi(s+1,X^{s,\gamma}_{s+1})-\int_s^{s+1}&\partial_u\psi(u,X_u^{s,\gamma})+\frac{1}{2}\textnormal{tr}(\sigma\sigma^{\top}\partial_{xx}\psi)(u,X_u^{s,\gamma})\\&+r(u,X^{s,\gamma}_u)\partial_{x}\psi(u,X_u^{s,\gamma})X^{s,\gamma}(u)\,du-\int_s^{s+1}Z^1(u)\,dB(u)\,.
	\end{align}
	Then,
	\begin{align}
		\hat{Y}(s)=\hat{Y}(\hat{\tau})-\int_s^{\hat{\tau}}&\partial_u \psi(u,X_u^{s,\gamma})+\frac{1}{2}\textnormal{tr}(\sigma\sigma^{\top}\partial_{xx}\psi)(u,X_u^{s,\gamma})+ r(u,X_u^{s,\gamma})\partial_{x}\psi(u,X_u^{s,\gamma})X^{s,\gamma}(u)\\
		&+f(u,X_u^{s,\gamma},Y^{s,\gamma}(u))\,du-\int_s^{\hat{\tau}}\hat{Z}(u)\,dB(u)\,.
	\end{align} 
	Observe that $\hat{Y}(s)=0$ and
	\begin{align}
		&\partial_v \psi(v,X_v^{s,\gamma})+\frac{1}{2}\textnormal{tr}(\sigma\sigma^{\top}\partial_{xx}\psi)(v,X_v^{s,\gamma})+ r(v,X^{s,\gamma}_v)\partial_{x}\psi(v,X_v^{s,\gamma})X^{s,\gamma}(v)\\
		&+f(v,X_v^{s,\gamma},\psi(v,X_v^{s,\gamma}))+C_f|\hat{Y}(v)|\le-\frac{M}{2} \text{ for } v\in[s,\hat{\tau}]\,.
	\end{align}
	Thus,
	\begin{align}
		\hat{Y}(\hat{\tau})&\le \int_s^{\hat{\tau}} -\frac{M}{2}-C_f|\hat{Y}(u)|-f(u,X_u^{s,\gamma},Y^1(u))+f(u,X_u^{s,\gamma},Y^{2}(u))\,du-\int_s^{\hat{\tau}}\hat{Z}(u)\,dB(u)\,.
	\end{align}
	Observe that 
	$\mathbb{E}[\hat{Y}(\hat{\tau})]=\mathbb{E}[ \psi(\hat{\tau},X_{\hat{\tau}})-\varphi(\hat{\tau},X_{\hat{\tau}}))] \ge0$, which is derived from \eqref{eqn:vis}, and
	$$-C_f|\hat{Y}(u)|-f(u,X_u^{s,\gamma},Y^1(u))+f(u,X_u^{s,\gamma},Y^{2}(u))<0\,.$$  This gives a contradiction. 
	
	For \eqref{eqn: fey iii}, we suppose $\varphi\in C_\rho^{1,2}([0,\infty)\times\Lambda)$.  Theorem \ref{thm: functional ito}  and \eqref{eqn: fey pde} imply
	\begin{align}
		\varphi(v,X_v^{0,x})&=\varphi(T,X_T^{0,x})-\int_v^T \partial_u\varphi(u,X_u^{0,x})+\frac{1}{2}\textnormal{tr}(\sigma\sigma^{\top}\partial_{xx}\varphi)(u,X_u^{0,x})\\
		&\quad\quad\quad\quad\quad\quad\quad+r(u,X_u^{0,x})\partial_{x}\varphi(u,X_u^{0,x})X^{0,x}(u)\,du-\int_v^T\sigma^{\top}\partial_{x}\varphi(u,X_u^{0,x})\,dB(u)\\
		&=\varphi(T,X_T^{0,x})+\int_v^T f(u,X_u^{0,x},\varphi(s,X_s^{0,x}))\,du-\int_v^T\sigma^{\top}\partial_{x}\varphi(u,X_u^{0,x})\,dB(u)\,.
	\end{align}
	It is evident that $|\varphi(v,X_v^{0,x})|\le L(1+|\!|X^{0,x}|\!|^\rho_v)$ since $\varphi\in C_\rho^{1,2}([0,\infty)\times\Lambda)$.
	The uniqueness of  Theorem \ref{thm: BSDE no delay} gives the desired result.
\end{proof}

\section{Proofs of the main results in Section  \ref{sec:inst}}
\label{app:fundingrate}

We first state the functional It\^o formula,  which is used to show Theorem \ref{thn:main}. The proof of this formula is in \cite{bally2016stochastic}.
\begin{theorem}\label{thm: functional ito}
	Let $X$ be a continuous semimartingale. Then for any $F\in C^{1,2}_{p}([0,\infty)\times\Lambda)$, we have
	\begin{align}
		F(s,X_s)=F(0,X_0)+\int_0^s \partial_u F(u,X_u)\,du+\int_0^s \partial_x F(u,X_u)\,dX(u)+\frac{1}{2}\int_0^s \textnormal{tr}\partial^2_{xx}F(u,X_u)\,d[X](u)\,.
	\end{align}
\end{theorem}

The proof of Theorem \ref{thn:main} is as follows.
\begin{proof}
	For (i), we first show that the map $y_1\mapsto H(y_1,0)$ has linear growth. More precisely,
	there exists a constant $C>0$ such that $|H(y_1,0)|\le C|y_1|$ for all $y_1\in \mathbb{R}$.
	From \eqref{label: 2-2} in Assumption \ref{Hypo:H},
	the function 
	$H$ can be written as  $H(y_1,y_2)=g(y_1,y_2)(y_1-y_2)$
	for some  function $g:\mathbb{R}^2\to\mathbb{R}$.
	Thus, it suffices to show that  the map $y_1\mapsto g(y_1,0)$ is bounded. By \eqref{label: 2-3} in Assumption \ref{Hypo:H}, there exists a constant $C_H>0$ such that $$|(y_1-y_2)g(y_1,y_2)-(y_1-y_2')g(y_1,y_2')|=|H(y_1,y_2)-H(y_1,y_2')|\le C_H|y_2-y_2'|$$ for all $y_1,y_2,y_2'\in \mathbb{R}$. Setting $y_1=y_2$ and $y_2'=0$  yields $|y_1g(y_1,0)|\le C_H|y_1|$, which implies $|g(y_1,0)|\le C_H $ for $y_1\neq0.$ 
	Hence, the map $y_1 \mapsto g(y_1, 0)$ is bounded, and the desired linear growth condition for $H(y_1, 0)$ follows.
	
	Now observe that
	$$\Phi(s,\gamma,0)=H(\varphi(s,\gamma),0)
	-\partial_s \varphi(s,\gamma)-\frac{1}{2}\textnormal{tr}(\sigma\sigma^{\top}\partial_{xx}\varphi)(s,\gamma)- r(s,\gamma)\partial_{x}\varphi(s,\gamma)\gamma(s)\,.$$
	The terms $H(\varphi(s,\gamma),0)$, $\partial_s \varphi(s,\gamma)$, 
	$\textnormal{tr}(\sigma\sigma^{\top}\partial_{xx}\varphi)(s,\gamma)$, 
	$r(s,\gamma)\partial_{x}\varphi(s,\gamma)\gamma(s)$
	have polynomial growth  of orders $p$, $p$, $p+2$, $p+1$ at most, respectively, in $(s,\gamma)$. 
	This implies that the map   $(s,\gamma) \mapsto \Phi(s,\gamma,0)$ has polynomial growth. There exist constants \(\rho \ge 1\) and $C_\Phi>0$  be constants such that
	\(
	|\Phi(s,\gamma,0)| \leq C_{\Phi} ( 1 + |\!|\gamma|\!|_s^\rho ) \)
	for all \((s, \gamma) \in [0, \infty) \times {\Lambda}\).  
	One may choose $\rho\ge p.$

	We now show (ii). Recall the BSDE \eqref{eqn:BSDE} 
	\begin{align}
		Y(s)
		&=Y(T)-\int_s^T(r(u,X_u)Y(u)-\Phi(u,X_u,Y(u)))\,du-\int_s^TZ(u)\,dB(u)\,.
	\end{align}
	By Theorem \ref{thm: BSDE no delay},
	this BSDE 
	has a unique solution $(Y,Z)$ in $\mathbb{S}^2(0,\infty;\mathbb{R})\times\mathbb{H}^2(0,\infty;\mathbb{R}^m)$ such that  $|Y|\le L(1+|\!|X|\!|^\rho)$ for some positive constant $L$.
	Define 
	$$(Y',Z')=(\varphi(s,X_s),(\partial_x \varphi\,\sigma)(s,X_s))_{s\ge0}\,,$$
	then $(Y',Z')$ is also a solution to the BSDE \eqref{eqn:BSDE}  because 
	\begin{equation}
		\begin{aligned}
			Y'(s)&=Y'(T)+\int_s^T-\partial_u \varphi(u,X_u)-\frac{1}{2}\textnormal{tr}(\sigma\sigma^{\top}\partial_{xx}\varphi)(u,X_u)- r(u,X_u)\partial_{x}\varphi(u,X_u)X(u)\,du\\ 
			&\quad\quad\quad\quad-\int_s^TZ'(u)\,dB(u)\\
			&=Y'(T)+\int_s^TH(\varphi(u,X_u),Y'(u))-\partial_u \varphi(u,X_u)\\
			&\quad\quad\quad\quad-\frac{1}{2}\textnormal{tr}(\sigma\sigma^{\top}\partial_{xx}\varphi)(u,X_u)- r(u,X_u)\partial_{x}\varphi(u,X_u)X(u)\,du-\int_s^TZ'(u)\,dB(u)\\
			&=Y'(T)-\int_s^Tr(u,X_u)Y'(u)-\Phi(u,X_u,Y'(u)) \,du-\int_s^TZ'(u)\,dB(u)\,.
		\end{aligned}
	\end{equation}
	As $\varphi\in C_p^{1,2}([0,\infty)\times {\Lambda})$, we know that
	$(Y',Z')\in\mathbb{S}^2(0,\infty;\mathbb{R})\times\mathbb{H}^2(0,\infty;\mathbb{R}^m)$ 
	and $|Y'|\le L(1+|\!|X|\!|^p)$ for some positive constant $L$. Since $p\le \rho,$ by the uniqueness in Theorem \ref{thm: BSDE no delay}, we have 
	$Y=Y'$ and $Z=Z'.$
\end{proof}

The proof of Corollary \ref{cor:main} is as follows.
\begin{proof}
	This is directly obtained by showing 
	\(
	|\Phi(s,\gamma,0)| \leq C_{\Phi} ( 1 + |\!|\gamma|\!|_s^{p+2} ) \)
	for all \((s, \gamma) \in [0, \infty) \times {\Lambda}\). 
	Because the terms $H(\varphi(s,\gamma),0)$, $\partial_s \varphi(s,\gamma)$, 
	$\textnormal{tr}(\sigma\sigma^{\top}\partial_{xx})\varphi(s,\gamma)$, 
	$r(s,\gamma)\partial_{x}\varphi(s,\gamma)\gamma(s)$
	have polynomial growth  of orders $p$, $p$, $p+2$, $p+1$ at most,  respectively, the map
	$\Phi(s,\gamma,0)$ has polynomial growth of order $p+2$ at most.
\end{proof}

\section{Proofs of the main results in Section  \ref{section: delay}}
\label{app:delay}

\subsection{Finite-horizon delayed BSDEs}\label{C.1}

In this section, we study finite-horizon delayed BSDEs. The following proposition will be used to prove \eqref{thm:delay_2} in Theorem \ref{thm:approx}.

\begin{theorem}\label{prop: delay} Let Assumptions \ref{hypo: SDE},\ref{sharpe} and \ref{hypo: BSDE delay} hold. 
	Then  for any $T>0$ and  $\xi\in L^2(\mathcal{F}_T;\mathbb{R})$,
	there exists a unique solution $(Y,Z)$ in $\mathbb{S}^2(0,T;\mathbb{R})\times \mathbb{H}^2(0,T;\mathbb{R}^m)$ to the BSDE
	\begin{align} 
		Y(s)&=\xi+\int_s^Tg(u,X_u)-r Y(u)-\frac{\ell-r}{\delta}\int_{u-\delta}^uY(v)\,dv\,du\\
		&\quad-\int_s^TZ(u)\,dB(u)\,,\;0\le s\le T\,.\label{eqn: BSDE}
	\end{align}
	In particular, if $\xi\equiv0$ then there is a constant $L>0$ depending only on $C_1,C_3,C_\Phi, r, \ell$ such that $|Y(s)|\le L(1+|\!|X|\!|^{\rho}_s)$  for all $s\in [0,T]$ where $\rho$ is the constant in Assumption \ref{hypo: BSDE delay}.
\end{theorem}

\begin{proof}
	
	We prove the existence and uniqueness of solutions to \eqref{eqn: BSDE} through the Banach fixed-point theorem. Define a map $\Gamma:\mathbb{S}^2(0,T;\mathbb{R})\to \mathbb{S}^2(0,T;\mathbb{R})$ by $\Gamma(U)=Y$ that satisfies 
	\begin{align}
		Y(s)=\mathbb{E}\Big[\xi+\int_s^Tg(u,X_u)-\ell Y(u)-\frac{\ell-r}{\delta}\int_{u-\delta}^uU(v)-U(u)\,dv\,du\Big| \mathcal{F}_s\Big]\,,\;0\le s\le T\,.
	\end{align}
	To verify that this map $\Gamma$ is well-defined, it suffices to check that the BSDE
	\begin{equation}
		\begin{aligned}
			Y(s)=\xi+\int_s^Tg(u,X_u)-\ell Y(u)&-\frac{\ell-r}{\delta}\int_{u-\delta}^uU(v)-U(u)\,dv\,du\\
			&\quad\quad -\int_s^TZ(u)\,dB(u)\,,\;0\le s\le T
		\end{aligned}
	\end{equation}
	has a unique solution $(Y,Z)$ in $\mathbb{S}^2(0,T;\mathbb{R})\times\mathbb{H}^2(0,T;\mathbb{R}^m)$.
	This is directly obtained from  \cite[Theorem 4.2.1, Theorem 4.3.1]{zhang2017backward} because the process $$\big(g(s,X_s)-\frac{\ell-r}{\delta}\int_{s-\delta}^sU(v)-U(u)\,dv\big)_{0\le s\le T}$$ is in $\mathbb{H}^2(0,T;\mathbb{R})$ 
	for $U\in \mathbb{S}^2(0,T;\mathbb{R})$, which is easily observed by 
	\begin{align}
		&\int_0^T\int_{u-\delta}^u|U(v)|^2\,dv\,du\le T\delta|\!|U|\!|_T\,.
	\end{align} 
	
	Now we verify that $\Gamma$ is a contraction map with respect to the norm 
	\begin{align}
		|\!|Y|\!|:=\Big(\mathbb{E}\Big[|Y(0)|^2+\int_0^Te^{M u}|Y(u)|^2\,du\Big]\Big)^{\frac{1}{2}}
	\end{align}
	where $M:=2-2\ell+6(\ell-r)^2$.
	Let $U,U'\in  \mathbb{S}^2(0,T;\mathbb{R})$ and $Y=\Gamma(U)$, $Y'=\Gamma(U')$.
	For simplicity, we define
	$$\hat{U}(u)=U(u)-U'(u)\,,\;
	\hat{Y}(u)=Y(u)-Y'(u)\,,\; \hat{Z}(u)=Z(u)-Z'(u)\,.$$
	Applying It\^{o}'s formula to $e^{Ms}|\hat{Y}(s)|^2$, we have
	\begin{align}
		&\quad|\hat{Y}(0)|^2+\int_0^T(M+2\ell) e^{Mu}|\hat{Y}(u)|^2\,du+\int_0^Te^{M u}|\hat{Z}(u)|^2\,du\\
		&=-\int_0^T\frac{2(\ell-r)}{\delta}e^{M u}\hat{Y}(u)\int_{u-\delta}^u\hat{U}(v)-\hat{U}(u)\,dv\,du-\int_0^T2e^{M u}\hat{Y}(u)\hat{Z}(u)\,dB(u)\,.
	\end{align}
	From the inequality $3ab\leq 3 a^2+\frac{1}{3}b^2$ for all $a,b>0$ and the Jensen inequality, it follows that
	\begin{align}
		&\quad|\hat{Y}(0)|^2+\int_0^T2e^{M u}|\hat{Y}(u)|^2\,du+\int_0^Te^{M u}|\hat{Z}(u)|^2\,du+\int_0^T2e^{M u}\hat{Y}(u)\hat{Z}(u)\,dB(u)\\
		&\le \int_0^T\frac{1}{3}e^{M u}|\hat{U}(u)|^2+\frac{e^{M u}}{3\delta}\int_{u-\delta}^u|\hat{U}(v)|^2\,dv\,du\,.
	\end{align}
	Observe that
	\begin{align}
		\int_0^T\frac{e^{M u}}{\delta}\int_{u-\delta}^u|\hat{U}(v)|^2\,dv\,du&=\int_0^T\frac{e^{M u}}{\delta}\int_{-\delta}^0|\hat{U}(v+u)|^2\,dv\,du\\
		&=\int_{-\delta}^0\int_0^T\frac{e^{M u}}{\delta}|\hat{U}(v+u)|^2\,du\,dv\\
		&=\int_{-\delta}^0\int_v^{T+v}\frac{e^{M(u-v)}}{\delta}|\hat{U}(u)|^2\,du\,dv\\
		&\le \int_{-\delta}^0 e^{-M v} \,dv (\int_{-\delta}^{0}\frac{e^{M u}}{\delta}|\hat{U}(0)|^2\,du+\int_0^T\frac{e^{M u}}{\delta}|\hat{U}(u)|^2\,du)\\
		&\le \frac{e^{|M|\delta}-1}{|M|\delta}(|\hat{U}(0)|^2+\int_0^Te^{M u}|\hat{U}(u)|^2\,du)\,.
	\end{align}
	Then we obtain
	\begin{align}
		&\mathbb{E}\Big[|\hat{Y}(0)|^2+\int_0^T 2e^{M u}|\hat{Y}(u)|^2\,du+\int_0^Te^{M u}|\hat{Z}(u)|^2\,du\Big]\\
		&\le \frac{1}{3}\mathbb{E}\Big[\int_0^Te^{Mu}|\hat{U}(u)|^2\,du\Big]+\frac{e^{M \delta}-1}{3M \delta}\mathbb{E}\Big[|\hat{U}(0)|^2+\int_0^Te^{M u}|\hat{U}(u)|^2\,du\Big] \,\\
		&\le \frac{e^{|M|\delta}-1}{3|M|\delta}\mathbb{E}\Big[|\hat{U}(0)|^2+2\int_0^Te^{M u}|\hat{U}(u)|^2\,du\Big]\,.
	\end{align}
	Thus, the map $\Gamma$ is a contraction map.

	We now construct a solution   
	$(Y,Z)$ to \eqref{eqn: BSDE}  in $\mathbb{S}^2(0,T;\mathbb{R})\times\mathbb{H}^2(0,T;\mathbb{R}^m)$.
	Because $\Gamma$ is a contraction map, by the Banach fixed point theorem, there exists a unique fixed point  $Y$ in the completion of $\mathbb{S}^2(0,T;\mathbb{R})$ with respect to the norm $|\!|\cdot|\!|$. Then 
	\begin{align}
		Y(s)=\mathbb{E}\Big[\xi+\int_s^Tg(u,X_u)-r Y(u)-\frac{\ell-r}{\delta}\int_{u-\delta}^uY(v)\,dv\,du\Big|\mathcal{F}_s\Big]\,,\;0\le s\le T\,.
	\end{align}
	It can be shown that the process inside the conditional expectation belongs to $\mathbb{H}^2(0,T;\mathbb{R}^m)$
	using that
	\begin{align}
		&\quad\int_0^T\frac{1}{\delta}\int_{u-\delta}^u|Y(v)|^2\,dv\,du\le \int_0^T\frac{1}{\delta}\int_{-\delta}^T|Y(v)|^2\,dv\,du\le L(1+T)\Big(|Y(0)|^2+\int_0^T|Y(u)|^2\,du\Big)\label{eqn: 7}
	\end{align}
	for some positive constant $L$.
	The martingale representation theorem yields that there exists a unique $Z\in \mathbb{H}^2(0,T;\mathbb{R}^m)$ such that 
	\begin{align}
		Y(s)=\xi+\int_s^Tg(u,X_u)-r Y(u)-\frac{\ell-r}{\delta}\int_{u-\delta}^uY(v)\,dv\,du-\int_s^TZ(u)\,dB(u)\,,\;0\le s\le T\,.
	\end{align}
	Thus, $(Y,Z)$ satisfies \eqref{eqn: BSDE}.
	By a simple calculation, we have
	
	\begin{align}
		\mathbb{E}[|\!|Y|\!|^2_T]&\le L\mathbb{E}\Big[|\xi|^2+\int_0^T|g(u,X_u)|^2+|Y(u)|^2+\frac{1}{\delta}\int_{u-\delta}^u|Y(v)|^2\,dv\,du+\int_0^T|Z(u)|^2\,du\Big]\\
		&\le L(1+T)\mathbb{E}\Big[|\xi|^2+|Y(0)|^2+\int_0^T|g(u,X_u)|^2+2|Y(u)|^2+|Z(u)|^2\,du\Big]
	\end{align}
	for some positive constant $L$. Therefore the solution $(Y,Z)$ belongs to $\mathbb{S}^2(0,T;\mathbb{R})\times\mathbb{H}^2(0,T;\mathbb{R}^m)$.

	We now prove that 
	if $\xi\equiv0$ then there is a constant $L>0$ depending only on $C_1,C_3,C_\Phi, r, \ell,\rho$ such that $|Y(s)|\le L(1+|\!|X|\!|^{\rho}_s)$  for all $s\ge0.$ 
	We construct a sequence of processes $(Y^k,Z^k)_{k\in\mathbb{N}}$ inductively. Define $(Y^0,Z^0)=(0,0)$
	and for $k\in \mathbb{N}$ let $(Y^k,Z^k)$ be a solution to 
	\begin{align}
		Y^k(s)&=\int_s^Tg(u,X_u)
		-\ell Y^k(u)-\frac{\ell-r}{\delta}\int_{u-\delta}^uY^{k-1}(v)-Y^{k-1}(u)\,dv\,du\\
		&\quad-\int_s^TZ^k(u)\,dB(u)\,.\label{eqn:8}
	\end{align}
	Because $Y^k=\Gamma(Y^{k-1})$ and $Y$ is a fixed point of $\Gamma,$ we know    
	$Y$ is the limit of $Y^k$. Thus it suffices to show that there exists a constant $L>0$, which is independent to $k$, such that $|Y^k(s)|\le L(1+|\!|X|\!|^{\rho}_s)$ for all $s\in [0,T].$

	To prove this, we need the inequality \eqref{eqn:bbb}  as a lemma.    From \eqref{eqn:bb}, it follows that  for any $\epsilon>0$ and $K>0$ there is a constant $L>0$ satisfying  
	\begin{align}
		(\mathbb{E}[|\!|X|\!|^\rho_T\boldsymbol{1}_A])^{\frac{1}{\rho}} \le  (e(\mathbb{E}[|\!|X|\!|^\rho_s\boldsymbol{1}_A])^{\frac{1}{\rho}} +L(\mathbb{E}[\boldsymbol{1}_A])^{\frac{1}{\rho}}) e^{(\epsilon+K+\frac{1}{2}(\frac{r}{\sqrt{2K}}+M_{\rho\vee 2}C_3)^2)(T-s)}\,.\label{eqn:31} 
	\end{align}    
	Expanding the term $
	(e(\mathbb{E}[|\!|X|\!|^\rho_s\boldsymbol{1}_A])^{\frac{1}{\rho}} +L(\mathbb{E}[\boldsymbol{1}_A])^{\frac{1}{\rho}})^\rho$ with  
	Young's inequality 
	\begin{align}
		ca^ib^{\rho-i}= (ua^i)(\frac
		{cb^{\rho-i}}{u})\le \frac{iu^{\frac{\rho}{i}}}{\rho}a^\rho+\frac{(\rho-i)c^{\frac{\rho}{\rho-i}}}{\rho u^{\frac{\rho}{\rho-i}}}b^\rho
	\end{align}
	for $a,b,c,u>0$  and $i=1,\cdots ,\rho-1$, we obtain that for any $\kappa>0$ there is $C_\kappa>0$ such that     
	\begin{align}
		\mathbb{E}[|\!|X|\!|^\rho_T\boldsymbol{1}_A] \le  ((1+\kappa)e^\rho(\mathbb{E}[|\!|X|\!|^\rho_s\boldsymbol{1}_A] +C_\kappa\mathbb{E}[\boldsymbol{1}_A]) e^{(\epsilon+K+\frac{1}{2}(\frac{r}{\sqrt{2K}}+M_{\rho\vee 2}C_3)^2)\rho(T-s)}\,.
	\end{align}
	Since $\ell>1+\inf_{K>0}(K+\frac{1}{2}(\frac{r}{\sqrt{2K}}+M_{\rho\vee 2}C_3)^2)\rho$ and $e^\rho((\ell-r)^2+\frac{1}{2}|\ell-r|\ell+2|\ell-r|)\delta<1$, for some positive constants $L_6,L_7,L_8$ satisfying $L_8<\ell-1$ and $L_6((\ell-r)^2+\frac{1}{2}|\ell-r|\ell+2|\ell-r|)\delta<1$, we have 
	\begin{align}\label{eqn:bbb}
		\mathbb{E}_s[|\!|X|\!|_T^\rho]\le (L_6|\!|X|\!|_s^\rho+L_7)e^{L_8(T-s)}\,.
	\end{align}

	We construct four sequences $(a_k)_{k\ge0},(b_k)_{k\ge0},(\tilde{a}_k)_{k\ge0},(\tilde{b}_k)_{k\ge0}$   that satisfy   
	\begin{equation}\label{eqn:ineq}
		\begin{aligned}
			&|Y^{k}(s)|\le a_{k}|\!|X|\!|^{\rho}_s+b_{k} \,,\\
			&|\mathbb{E}_s[Y^{k}(s_1)-Y^{k}(s_2)]\le (s_2-s_1)(\tilde{a}_{k}\mathbb{E}_s[|\!|X|\!|^{\rho}_{s_2}]+\tilde{b}_{k})
		\end{aligned}
	\end{equation} for all $0\le s\le s_1\le s_2$. 
	Define $a_0=b_0=\tilde{a}_0=\tilde{b}_0=0$ then \eqref{eqn:ineq} is satisfied with $Y^0=0.$    
	From \eqref{eqn:finite_esti},  there are positive constants $a_1$ and $b_1$,  depending only on $C_1,C_3, C_\Phi,r,\ell, \rho$, such that $|Y^1(s)|\le a_1|\!|X|\!|^{\rho}_s+b_1$.
	Given $a_0,b_0,\tilde{a}_0,\tilde{b}_0, a_1,b_1$, 
	we define inductively    
	\begin{align}
		a_{k+1}&=a_1+2|\ell-r| L_6\delta a_k+\frac{1}{2}|\ell-r| L_6\delta \tilde{a}_k\,,\label{eqn:17}\\
		b_{k+1}&=b_1+2|\ell-r| L_7\delta a_k+2|\ell-r| \delta b_k+\frac{1}{2}|\ell-r| L_7\delta \tilde{a}_k+\frac{1}{2}|\ell-r|\delta \tilde{b}_k\,,
	\end{align}
	and
	\begin{align}
		&\tilde{a}_{k+1}=C_\Phi+\ell a_{k+1}+2|\ell-r| a_k\,,\label{eqn: 16}\\ 
		&\tilde{b}_{k+1}=C_\Phi+\ell b_{k+1}+2|\ell-r| b_k\,.
	\end{align}
	Applying It\^o's formula,  we have
	\begin{align}
		Y^{k+1}(s)&=\mathbb{E}_s\Big[\int_s^Te^{-\ell(u-s)}g(u,X_u)-\frac{(\ell-r) e^{-\ell(u-s)}}{\delta}\int_{u-\delta}^uY^{k}(v)-Y^{k}(u)\,dv\,du\Big]\,.
	\end{align}
	Using \eqref{eqn:bbb}, each term inside the conditional expectation of the above equation can be estimated as 
	\begin{align}
		&\quad\mathbb{E}_s\Big[\int_s^Te^{-\ell(u-s)}g(u,X_u)\Big]\\
		&\le\mathbb{E}_s\Big[\int_s^Te^{-\ell(u-s)}C_\Phi(1+|\!|X|\!|^{\rho}_u)\Big] \le a_1|\!|X|\!|^{\rho}_s+b_1\,,\\
		&\quad\mathbb{E}_s\Big[\int_s^{s+\delta}\frac{(\ell-r) e^{-\ell(u-s)}}{\delta}\int_{u-\delta}^uY^{k}(v)-Y^{k}(u)\,dv\,du\Big]\\
		&\le \mathbb{E}_s\Big[\int_s^{s+\delta}2|\ell-r| e^{-\ell(u-s)}(a_k|\!|X|\!|^{\rho}_u+b_k)\,du\Big] \\
		&\le 2|\ell-r| L_6\frac{1-e^{-(\ell-L_8)\delta}}{\ell-L_8}a_k|\!|X|\!|^{\rho}_s+2|\ell-r| L_7\frac{1-e^{-(\ell-L_8)\delta}}{\ell-L_8}a_k+2|\ell-r|\frac{1-e^{-\ell\delta}}{\ell}b_k\,,\\
		&\quad\mathbb{E}_s\Big[\int_{s+\delta}^{T}\frac{(\ell-r) e^{-\ell(u-s)}}{\delta}\int_{u-\delta}^uY^{k-1}(v)-Y^{k-1}(u)\,dv\,du\Big]\\
		&\le
		\int_{s+\delta}^T\frac{1}{2}|\ell-r| e^{-\ell(u-s)}\delta(\tilde{a}_k\mathbb{E}_s[|\!|X|\!|^{\rho}_u]+\tilde{b}_k)\,du\\
		&\le \frac{|\ell-r| L_6e^{-(\ell-L_8)\delta}}{2(\ell-L_8)}\delta \tilde{a}_k|\!|X|\!|^{\rho}_s+\frac{|\ell-r| L_7e^{-(\ell-L_8)\delta}}{2(\ell-L_8)}\delta \tilde{a}_k+|\ell-r|\frac{e^{-\ell\delta}}{2\ell}\delta \tilde{b}_k\,.
	\end{align} 
	Observe that $\ell-L_8>1$ and $\frac{1-e^{a\delta}}{a}\le \delta$ for any $a>0$. 
	Thus,
	we have $$|Y^{k+1}(s)|\le a_{k+1}|\!|X|\!|^{\rho}_s+b_{k+1}\,.$$ By \eqref{eqn:8},  
	\begin{align}
		&\;\quad\big|\mathbb{E}_s[Y^{k+1}(s_1)-Y^{k+1}(s_2)]\big|\\
		&=\Big|\mathbb{E}_s\Big[\int_{s_1}^{s_2}g(u,X_u)
		-\ell Y^{k+1}(u)-\frac{\ell-r}{\delta}\int_{u-\delta}^uY^{k}(v)-Y^{k}(u)\,dv\,du\Big]\Big|\\
		&\le (s_2-s_1)((C_\Phi+\ell a_{k+1}+2|\ell-r| a_k)\mathbb{E}_s[|\!|X|\!|^{\rho}_{s_2}]+C_\Phi+\ell b_{k+1}+2|\ell-r| b_k )\,.
	\end{align}
	Thus, we obtain  $$|\mathbb{E}_s[Y^{k+1}(s_1)-Y^{k+1}(s_2)]\le (s_2-s_1)(\tilde{a}_{k+1}\mathbb{E}_s[|\!|X|\!|^{\rho}_{s_2}]+\tilde{b}_{k+1})\,.$$ 
	By substituting \eqref{eqn: 16} into \eqref{eqn:17}, we have 
	\begin{align}\label{eqn:recur}
		a_{k+1}&=a_1+\frac{1}{2}|\ell-r| L_6C_\Phi\delta+L_6(\frac{1}{2}|\ell-r|\ell+2|\ell-r|)\delta a_k+L_6(\ell-r)^2\delta a_{k-1}\,,\\
		b_{k+1}&=b_1+\frac{1}{2}|\ell-r| C_\Phi \delta+2|\ell-r| L_7\delta a_k+\frac{1}{2}|\ell-r| L_7\delta \tilde{a}_k+2|\ell-r| \delta b_k+\frac{1}{2}|\ell-r|\ell \delta b_k \\
		&\quad+|\ell-r|^2\delta b_{k-1}\,.
	\end{align}

	Now we prove the convergence of the sequence $(a_k)_{k\ge0}$ by verifying that it is  increasing and bounded above. From  the initial condition $0=a_0<a_1$ and the recursive relation, it follows that $a_0<a_1<a_2$ and 
	\begin{align}
		a_{k+2}-a_{k+1}=L_6(\frac{1}{2}|\ell-r|\ell+2|\ell-r|)\delta (a_{k+1}-a_k)+L_6(\ell-r)^2\delta (a_k-a_{k-1})\,, \text{ for }k\ge1.
	\end{align}
	This equation implies that $a_{k+2}-a_{k+1}> 0$ whenever $a_{k-1}< a_k< a_{k+1}$. Therefore by mathematical induction, the sequence $(a_k)_{k\ge 0}$ is  increasing. To show that this sequence is bounded above, recall the condition $L_6((\ell-r)^2+\frac{1}{2}|\ell-r|\ell+2|\ell-r|)\delta<1$.
	Define
	\begin{align}
		a:=\frac{a_1+\frac{1}{2}|\ell-r| L_6C_\Phi\delta}{1-L_6((\ell-r)^2+\frac{1}{2}|\ell-r|\ell+2|\ell-r|)\delta}
	\end{align}
	then $a>0$ and $a_0,a_1\le a.$
	By \eqref{eqn:recur}, if  $a_{k-1},a_k\le a$  then 
	\begin{align}
		a_{k+1}&=a_1+\frac{1}{2}|\ell-r| L_6C_\Phi\delta+L_6(\frac{1}{2}|\ell-r|\ell+2|\ell-r|)\delta a_k+L_6(\ell-r)^2\delta a_{k-1}\\
		&\le a_1+\frac{1}{2}|\ell-r| L_6C_\Phi\delta+L_6(\frac{1}{2}|\ell-r|\ell+2|\ell-r| +(\ell-r)^2)\delta a\\
		&\le a\,.
	\end{align}
	By induction, $a_k \le a$ for all $k \ge 0$, thus the sequence $(a_k)_{k\ge0}$ is bounded above by $a$.
	Since $(a_k)_{k\ge0}$ is  increasing and bounded above, it converges to a positive constant. Moreover, by the recursive relation in \eqref{eqn:recur}, it can be shown that  
	$	\lim_{k\to\infty}a_k=a.
	$ 
	Similarly, the sequences $(b_k)_{k\ge0}$ and $(\tilde{a}_k)_{k\ge0}$  also converge to positive constants. We have
	\begin{align} 
		\lim_{k\to \infty}\tilde{a}_k=C_{\Phi}+(\ell+2|\ell-r|)\frac{a_1+\frac{1}{2}|\ell-r| L_6C_\Phi\delta}{1-L_6((\ell-r)^2+\frac{1}{2}|\ell-r|\ell+2|\ell-r|)\delta} 
	\end{align}
	and
	\begin{align}\label{eqn:b_k}
		\lim_{k\to \infty}b_k=\frac{b_1+\frac{1}{2}|\ell-r| C_\Phi \delta+2|\ell-r| L_7\delta a +\frac{1}{2}|\ell-r| L_7\delta(C_\phi+(\ell+2|\ell-r|)a}{1-((\ell-r)^2+\frac{1}{2}|\ell-r|\ell+2|\ell-r|)\delta}\,.
	\end{align} 
	Taking $k\to\infty$ to the first inequality in  \eqref{eqn:ineq},
	we obtain the desired results.
\end{proof}

The above proof verifies the convergence of the Picard iteration 
$(Y^k)_{k\in\mathbb{N}}$ 
to the unique solution $Y$ of the BSDE \eqref{eqn: BSDE}
through the Banach fixed point theorem. 
It can also be shown that the sequence \((Z^k)_{k \in \mathbb{N}}\) converges to the unique solution \(Z\) of the BSDE. However, we omit the convergence of the sequence $(Z^k)_{k \in \mathbb{N}}$ in the following corollary, as it is not required for our purposes.
\begin{corollary} \label{cor:banach}
	Let $(Y^k,Z^k)_{k\in\mathbb{N}}$ be the
	Picard iteration for the BSDE \eqref{eqn: BSDE}. More precisely, define $(Y^0,Z^0)=(0,0)$
	and for $k\in \mathbb{N}$ let $(Y^k,Z^k)$ be a solution to 
	\begin{align}
		Y^k(s)=\int_s^Tg(u,X_u)
		-\ell Y^k(u)-\frac{\ell-r}{\delta}\int_{u-\delta}^uY^{k-1}(v)-Y^{k-1}(u)\,dv\,du-\int_s^TZ^k(u)\,dB(u)\,.\label{eqn:9}
	\end{align}
	Then $(Y^k)_{k\in\mathbb{N}}$ converges  to the unique solution $Y$ of the BSDE \eqref{eqn: BSDE} $\mathbb{Q}\otimes ds$-almost surely. 
\end{corollary}

\subsection{Proof of Theorem    \ref{thm:approx}}\label{L speci}

We now prove \eqref{thm:delay_1} and \eqref{thm:delay_2} in Theorem \ref{thm:approx}.
Part \eqref{thm:delay_1} 
is directly obtained by Theorem \ref{thn:main}.
The proof of \eqref{thm:delay_2} is as follows.

\begin{proof}
	In this proof, $L$ denotes a generic constant that depends only on $C_1,C_3, C_\Phi, r,\ell,\rho$ and may differ line by line and  $L_6,L_7,L_8$ are positive constants satisfying  
	\begin{equation} 
		\begin{aligned}\label{eqn:posi}
			&L_8<\ell-1\,,\\
			&L_6((\ell-r)^2+\frac{1}{2}|\ell-r|\ell+2|\ell-r|)\delta<1\,,\\
			&\mathbb{E}_s[|\!|X|\!|_T^\rho]\le (L_6|\!|X|\!|^\rho_s+L_7)e^{L_8(T-s)}\,.
		\end{aligned}     
	\end{equation}
	We first prove the uniqueness of solutions. Let $(Y^1,Z^1)$ and $(Y^2,Z^2)$ be two solutions to \eqref{eqn: infinie delay bsde} such that $|Y^i(s)|\le L(1+|\!|X|\!|^\rho_s)$ for $i=1,2.$ Define $\hat{Y}=Y^1-Y^2$ and $\hat{Z}=Z^1(s)-Z^2(s)$. Then we have
	\begin{align}
		\hat{Y}(s)=\hat{Y}(T)+\int_s^T-r \hat{Y}(u)-\frac{\ell-r}{\delta}\int_{u-\delta}^u\hat{Y}(v)\,dv\,du-\int_s^T\hat{Z}(u)\,dB(u)\,,\; 0\le s\le T<\infty\,.
	\end{align}
	Let $(\hat{Y}^k,\hat{Z^k})$ denote the Picard iteration of this BSDE. More precisely,
	define $(\hat{Y}^0,\hat{Z}^0)=(0,0)$ and for each $k\in \mathbb{N}$, let $(\hat{Y}^k,\hat{Z}^k)$ be a solution to 
	\begin{align}
		\hat{Y}^k(s)=\hat{Y}(T)+\int_s^T-\ell \hat{Y}^k(u)-\frac{\ell-r}{\delta}\int_{u-\delta}^u\hat{Y}^{k-1}(v)-\hat{Y}^{k-1}(u)\,dv\,du-\int_s^T\hat{Z}^k(u)\,dB(u)\,.
	\end{align}

	We construct four sequences $(a_k)_{k\ge0},(b_k)_{k\ge0},(\tilde{a}_k)_{k\ge0},(\tilde{b}_k)_{k\ge0}$   that satisfy   
	\begin{align}\label{eqn:sequ}
		&|\hat{Y}^k(s)|\le e^{-(\ell-L_8)(T-s)}(\sum_{j=0}^{k-1}\frac{(T-s)^{j}}{j!})(a_k|\!|X|\!|^{\rho}_s+b_k)\,,\\
		&|\mathbb{E}_s[\hat{Y}^k(s_1)-\hat{Y}^k(s_2)]|\le (s_2-s_1)e^{-(\ell-L_8)(T-s_2)}\Big(\sum_{j=0}^{k-1}\frac{(T-s_2)^{j}}{j!}\Big)(\tilde{a}_k\mathbb{E}_s[|\!|X|\!|^{\rho}_{s_2}]+\tilde{b}_k) 
	\end{align} 
	for all $0\le s\le s_1\le s_2$. 
	Define $a_0=b_0=\tilde{a}_0=\tilde{b}_0=0$ then \eqref{eqn:sequ} is satisfied with $\hat{Y}^0=0.$     
	It can be easily checked that there are constants $a_1$ and $b_1$, depending on $C_1,C_3, C_\Phi,r,\ell,\rho$, such that
	\begin{align}
		|\hat{Y}^1(s)|=|\mathbb{E}_s[e^{-\ell(T-s)}\hat{Y}(T)]|\le e^{-(\ell-L_8)(T-s)}(a_1|\!|X|\!|^{\rho}_s+b_1)\,.
	\end{align}
	Given $a_0,b_0,\tilde{a}_0,\tilde{b}_0, a_1,b_1$, 
	we define inductively    
	\begin{align}
		&a_{k+1}=a_1+2|\ell-r| L_6\delta a_k+\frac{1}{2}|\ell-r| L_6\delta \tilde{a}_k\,,\\ 
		&b_{k+1}=b_1+2|\ell-r| L_7\delta a_k+2|\ell-r| \delta b_k+\frac{1}{2}|\ell-r| L_7\delta \tilde{a}_k+\frac{1}{2}|\ell-r|\delta \tilde{b}_k\,,
	\end{align}
	and
	\begin{align}
		&\tilde{a}_{k+1}=\ell a_{k+1}+2|\ell-r| a_k\,, \\ 
		&\tilde{b}_{k+1}=\ell b_{k+1}+2|\ell-r| b_k\,.
	\end{align}
	According to It\^o's formula,
	\begin{align}
		\hat{Y}^{k+1}(s)=\mathbb{E}_s\Big[e^{-\ell(T-s)}\hat{Y}(T)-\int_s^T\frac{(\ell-r) e^{-\ell(u-s)}}{\delta}\int_{u-\delta}^u\hat{Y}^k(v)-\hat{Y}^k(u)\,dv\,du\Big]\,.
	\end{align}
	Decompose the right hand side of the above equation into 
	\begin{align}
		&\mathbb{E}_s[e^{-\ell(T-s)}\hat{Y}(T)]+\mathbb{E}_s\Big[\int_s^{s+\delta}-\frac{(\ell-r) e^{-\ell(u-s)}}{\delta}\int_{u-\delta}^u\hat{Y}^k(v)-\hat{Y}^k(u)\,dv\,du\Big]\\
		&+\mathbb{E}_s\Big[\int_{s+\delta}^{T}-\frac{(\ell-r) e^{-\ell(u-s)}}{\delta}\int_{u-\delta}^u\hat{Y}^k(v)-\hat{Y}^k(u)\,dv\,du\Big]\,.
	\end{align}
	Observe that
	\begin{align}
		&\mathbb{E}_s[e^{-\ell(T-s)}\hat{Y}(T)]\le  e^{-(\ell-L_8)(T-s)}(a_1|\!|X|\!|^{\rho}_s+b_1)\,,\\
		&\mathbb{E}_s\Big[\int_s^{s+\delta}\frac{|\ell-r| e^{-\ell(u-s)}}{\delta}\int_{u-\delta}^u\hat{Y}^k(v)-\hat{Y}^k(u)\,dv\,du\Big]\\
		&\le\mathbb{E}_s\Big[\int_s^{s+\delta}2|\ell-r| e^{-\ell(u-s)}e^{-(\ell-L_8)(T-u)}(\sum_{j=0}^{k-1}\frac{(T-u)^{j}}{j!})(a_k|\!|X|\!|^{\rho}_u+b_k)\,du\Big]\\
		&\le \int_s^{s+\delta}e^{-(\ell-L_8)(T-s)}\Big(\sum_{j=0}^{k-1}\frac{(T-u)^{j}}{j!}\Big)(2|\ell-r| a_kL_6|\!|X|\!|^{\rho}_s+2|\ell-r| L_7a_k+2|\ell-r| b_k)\,du\\
		&\le e^{-(\ell-L_8)(T-s)}(\sum_{j=0}^{k-1}\frac{(T-s)^{j}}{j!})(2|\ell-r| L_6\delta a_k|\!|X|\!|^{\rho}_s+2|\ell-r| L_7\delta  a_k+2|\ell-r|\delta b_k)\,,\\
		&\mathbb{E}_s\Big[\int_{s+\delta}^T\frac{|\ell-r| e^{-\ell(u-s)}}{\delta}\int_{u-\delta}^u(\hat{Y}^k(v)-\hat{Y}^k(u)\,dv\,du\Big]\\
		&\le 
		\int_s^T|\ell-r| e^{-\ell(u-s)} e^{-(\ell-L_8)(T-u)}\Big(\sum_{j=0}^{k-1}\frac{(T-u)^{j}}{j!}\Big)\frac{1}{2}\delta (\tilde{a}_k\mathbb{E}_s[|\!|X|\!|^{\rho}_{u}]+\tilde{b}_k)\,du\\
		&\le e^{-(\ell-L_8)(T-s)}\Big(\sum_{j=1}^{k}\frac{(T-s)^{j}}{j!}\Big)(\frac{1}{2}|\ell-r| L_6\delta \tilde{a}_k|\!|X|\!|^{\rho}_s+\frac{1}{2}|\ell-r| L_7\delta \tilde{a}_k+\frac{1}{2}|\ell-r|\delta \tilde{b}_k)\,,
	\end{align}
	where we have used that $e^{(\ell-L_8)(T-s)}\sum_{j=0}^k\frac{(T-s)^j}{j!}$ is an increasing function in $s$.
	Thus we have
	\begin{align}
		|\hat{Y}^{k+1}(s)|\le e^{-(\ell-L_8)(T-s)}\Big(\sum_{j=0}^{k}\frac{(T-s)^{j}}{j!}\Big)(a_{k+1}|\!|X|\!|^{\rho}_s+b_{k+1})\,.
	\end{align}
	For $s\le s_1\le s_2\le T$,
	\begin{align}
		&\;\quad\big|\mathbb{E}_s[\hat{Y}^{k+1}(s_1)-\hat{Y}^{k+1}(s_2)]\big|\\
		&=\Big|\mathbb{E}\Big[\int_{s_1}^{s_2} -\ell\hat{Y}^{k+1}(u)-\frac{\ell-r}{\delta}\int_{u-\delta}^u\hat{Y}^{k}(v)-\hat{Y}^{k}(u)\,dv\,du\Big]\Big|\\
		&\le (s_2-s_1)e^{-(\ell-L_8)(T-s_2)}\Big(\sum_{j=0}^{k}\frac{(T-s_2)^{j}}{j!}\Big)((\ell a_{k+1}+2|\ell-r| a_k)\mathbb{E}_s[|\!|X|\!|^{\rho}_{s_2}]+\ell b_{k+1}+2|\ell-r| b_k)\\
		&= (s_2-s_1)e^{-(\ell-L_8)(T-s_2)}\Big(\sum_{j=0}^{k}\frac{(T-s_2)^{j}}{j!}\Big)(\tilde{a}_{k+1}\mathbb{E}_s[|\!|X|\!|^{\rho}_{s_2}]+\tilde{b}_{k+1})\,.
	\end{align}

	We now verify $Y^1=Y^2$
	by taking $k\to\infty$ to \eqref{eqn:sequ}. 
	Following the arguments in \eqref{eqn:recur} and using  the inequalities in \eqref{eqn:posi},
	the limits
	\begin{align}
		&a:=\lim_{k\to \infty}a_k=\frac{a_1}{1-L_6((\ell-r)^2+\frac{1}{2}|\ell-r|\ell+2|\ell-r|)\delta}\,,\\
		&b:=\lim_{k\to \infty}b_k=\frac{b_1+aL_7((\ell-r)^2+\frac
			{1}{2}|\ell-r|\ell+2|\ell-r|)\delta}{1-((\ell-r)^2+\frac{1}{2}|\ell-r|\ell+2|\ell-r|)\delta} 
	\end{align} exist and  are positive. 
	By Corollary \ref{cor:banach}, the sequence $(\hat{Y}^k(s))_{k\ge0}$ converge to $\hat{Y}(s)$, thus we obtain $$
	|\hat{Y}(s)|\le Le^{-(\ell-L_8-1)(T-s)}(|\!|X|\!|^{\rho}_s+1)\,,\;s\ge0$$
	By letting $T\to\infty$, we have $Y^1=Y^2$. This induces  $Z^1=Z^2$.
	
	Now we prove the existence of solutions. For each $n\in \mathbb{N}$, consider the BSDE
	\begin{align}\label{bsde_n}
		Y^n(s)=\int_s^ng(u,X_u)-r Y^n(u)-\frac{\ell-r}{\delta}\int_{u-\delta}^uY^n(v)\,dv\,du-\int_s^nZ(u)\,dB(u)\,,\; 0\le s\le n \,.
	\end{align}
	For $n\le m$, let $\tilde{Y}(s):=Y^m(s)-Y^n(s)$ and $\tilde{Z}(s):=Z^m(s)-Z^n(s)$. Then for $s\le n$
	\begin{align}
		\tilde{Y}(s)=Y^m(n)+\int_s^n-r \tilde{Y}(u)-\frac{\ell-r}{\delta}\int_{u-\delta}^u\tilde{Y}(v)\,dv\,du-\int_s^n \tilde{Z}(u)\,dB(u)\,.
	\end{align}
	Observe that $|Y^m(n)|\le L(|\!|X|\!|^{\rho}_n+1)$  by Theorem \ref{prop: delay}.
	By a similar argument as above, it follows that 
	\begin{align}
		|\tilde{Y}(s)|=|Y^m(s)-Y^n(s)|\le Le^{-(\ell-L_8-1)(n-s)}(|\!|X|\!|^{\rho}_s+1)\,.
	\end{align}
	Therefore, for $0\le T\le n\le m$, we have
	\begin{align}
		\mathbb{E}[|\!|Y^n-Y^m|\!|^2_T]\le Le^{-2(\ell-L_8-1)(n-T)}(|\!|X|\!|^{2\rho}_T+1)\,.
	\end{align}
	This indicates that $(Y^n)_{n\in \mathbb{N}}$ is a Cauchy sequence in $\mathbb{S}^2(0,T;\mathbb{R})$ for each $T\ge0$ and we denote  as $Y^{\delta}$ the limit of $Y^n$. Moreover, by Theorem \ref{prop: delay}, there exists a constant $L>0$ such that $|Y^\delta|\le L(1+|\!|X|\!|^{\rho})$.
	The sequence   $(Z^n)_{n\in\mathbb{N}}$ is a Cauchy sequence in $\mathbb{H}^2(0,T;\mathbb{R}^m)$ for each $T\ge 0$
	because
	\begin{align}\label{eqn:z}
		\mathbb{E}\Big[\int_0^T|\tilde{Z}(u)|^2\,du\Big]&\le L\Big(\mathbb{E}\Big[|\tilde{Y}(T)|^2+\int_0^T|\tilde{Y}(u)|^2+\frac{1}{\delta}\int_{u-\delta}^u|\tilde{Y}(v)|^2\,dv\,du\Big]\Big)\\
		&\le L(1+T)(\mathbb{E}[|\!|X|\!|_T^{2\rho}]+1)e^{-2(\ell-L_8-1)(n-T)}\,.
	\end{align}
	which is obtained from It\^o's formula and the inequality $2ab\le a^2+b^2$ for $a,b>0$.
	Denote as $Z^{\delta}$   the limit of $Z^n$. 
	Because   $(Y^n(s),Z^n(s))_{0\le s\le n}$ satisfies 
	\begin{align}
		Y^n(s)= Y^n(T)+\int_s^Tg(u,X_u)-r Y^n(u)&-\frac{\ell-r}{\delta}\int_{u-\delta}^uY^n(v)\,dv\,du\\
		&-\int_s^TZ^n(u)\,dB(u)\,,\;0\le s\le T,
	\end{align}
	the Lebesgue dominated convergence theorem implies that 
	the pair $(Y^\delta,Z^\delta)$ is a solution to \eqref{eqn: infinie delay bsde}.
\end{proof}

We now prove \eqref{eqn:estimate} in Theorem \ref{thm:approx}.
\begin{proof} 
	We recall that $L$ denotes a generic constant that depends only on $C_1,C_3, C_\Phi, r,\ell,\rho$ and may differ line by line, and  $L_6,L_7,L_8$ are constants satisfying  
	\eqref{eqn:posi}.  By Theorems \ref{thm: BSDE no delay} and \eqref{thm:delay_2} in Theorem \ref{thm:approx}, there exists a unique solution $(Y,Z)$  to \eqref{eqn:BSDE}  and $(Y^\delta,Z^\delta)$ to \eqref{eqn: infinie delay bsde} in $\mathbb{S}^2(0,\infty;\mathbb{R})\times \mathbb{H}^2(0,\infty;\mathbb{R}^m)$. For each $n\in \mathbb{N}$, consider the finite-horizon BSDEs 
	\begin{equation}
		\begin{aligned}\label{eqn:Y^n}
			Y^n(s)&=\int_s^n\Phi(u,X_u,0)-\ell Y^n(u)\,du-\int_s^nZ^n(u)\,dB(u)\,,\\
			Y^{n,\delta}(s)&=\int_s^n\frac{1}{\delta}\int_{u-\delta}^u\Phi(v,X_v,0)\,dv-r Y^{n,\delta}(u)-\frac{\ell-r}{\delta}\int_{u-\delta}^u Y^{n,\delta}(v)\,dv\,du\\
			&\quad-\int_s^n Z^{n,\delta}(u)\,dB(u)\,.
		\end{aligned}
	\end{equation}
	Proposition \ref{prop: delay} guarantees the existence and uniqueness of  solutions to this BSDE. One can easily prove \begin{align}
		&\lim_{\delta\to 0}\mathbb{E}[|\!|Y^n-Y|\!|^2_T]\le L(1+|\!|X|\!|^{2\rho}_T)e^{-2(\ell-L_8)(n-T)}\,,\\
		&\lim_{\delta\to0}\mathbb{E}[|\!|Y^{n,\delta}-Y^{\delta}|\!|^2_T]\le L(1+|\!|X|\!|^{2\rho}_T)e^{-2(\ell-L_8-1)(n-T)}\,.
	\end{align}  
	Observe that
	\begin{equation}\label{eqn:Y}
		\begin{aligned}
			|\!|Y^{\delta}-Y|\!|_T&\leq |\!|Y^{\delta}-Y^{n,\delta}|\!|_T+ |\!|Y^{n,\delta}-Y^n|\!|_T+|\!|Y^n-Y|\!|_T\\
			& \leq L(1+|\!|X|\!|^\rho_T)e^{-(\ell-L_8)(n-T)}+|\!|Y^{n,\delta}-Y^n|\!|_T+L(1+|\!|X|\!|^{\rho}_T)e^{-(\ell-L_8-1)(n-T)}\,.
		\end{aligned}  	
	\end{equation}
	Then
	\begin{align}
		\mathbb{E}[|\!|Y^\delta-Y|\!|_T]&\le \big(\mathbb{E}[|\!|Y^\delta-Y|\!|_T^2]\big)^{\frac{1}{2}}\\
		&\le L(1+(\mathbb{E}[|\!|X|\!|^{2\rho}_T])^{\frac{1}{2}})e^{-(\ell-L_8)(n-T)}+(\mathbb{E}[|\!|Y^{n,\delta}-Y^n|\!|^2_T])^{\frac{1}{2}}\\
		&\quad+L(1+(\mathbb{E}[|\!|X|\!|^{2\rho}_T])^{\frac{1}{2}})e^{-(\ell-L_8-1)(n-T)}\,. \label{eqn:delta} 
	\end{align}

	We now show that  \begin{align}
		\lim_{\delta\to 0}\mathbb{E}[|\!|Y^{n,\delta}-Y^n|\!|^2_T]=0\; \textnormal{for } n\in \mathbb{N}\,.\label{eqn:n delta}
	\end{align} 
	Once this is proven, by taking $\delta\to 0$ and $n\to\infty$ to \eqref{eqn:delta}, we obtain the first inequality in \eqref{eqn:estimate}.  
	Let $\hat{Y}(s)=Y^{n,\delta}(s)-Y^n(s)$ and $\hat{Z}(s)=Z^{n,\delta}(s)-Z^n(s)$. Then we have
	\begin{align}\label{eqn:hat_Y}
		\hat{Y}(s)&=\int_s^n\frac{1}{\delta}\int_{u-\delta}^u \Phi(v,X_v,0)-\Phi(u,X_u,0)\,dv-\frac{\ell-r}{\delta}\int_{u-\delta}^uY^n(v)-Y^n(u)\,dv -r\hat{Y}(u)\\
		&\quad
		-\frac{\ell-r}{\delta}\int_{u-\delta}^u\hat{Y}(v)\,dv\,du-\int_s^n\hat{Z}(u)\,dB(u)\,,\; 0\le s\le n\,.
	\end{align}
	Let $(\hat{Y}^k,\hat{Z^k})$ denote the Picard iteration of this BSDE. More precisely, define $(\hat{Y}^0,\hat{Z}^{0})=(0,0)$ and for each $k\in \mathbb{N}$, let $(\hat{Y}^k,\hat{Z^k})$ be a solution to 
	\begin{align}\label{eqn:Picard}
		\hat{Y}^k(s)&=\int_s^n\frac{1}{\delta}\int_{u-\delta}^u \Phi(v,X_v,0)-\Phi(u,X_u,0)\,dv-\frac{\ell-r}{\delta}\int_{u-\delta}^uY^n(v)-Y^n(u)\,dv\\
		&\quad -\ell\hat{Y}^k(u)-\frac{\ell-r}{\delta}\int_{u-\delta}^u\hat{Y}^{k-1}(v)-\hat{Y}^{k-1}(u)\,dv\,du-\int_s^n\hat{Z}^k(u)\,dB(u)\label{eqn:11}\,,
	\end{align}
	and we define $$\hat{G}^{\delta}(s):=\mathbb{E}_s\Big[\int_s^n\frac{e^{-\ell(u-s)}}{\delta}\int_{u-\delta}^u |\Phi(v,X_v,0)-\Phi(u,X_u,0)|\,dv\,du \Big]$$ for $0\le s\le n$.
	
	We construct four sequences  $(a_k)_{k\ge0}$, $(b_k)_{k\ge0}$, $(\tilde{a}_k)_{k\ge0}$, $(\tilde{b}_k)_{k\ge0}$ that satisfy 
	\begin{align}\label{eqn :24}
		&|\hat{Y}^k(s)|\le \hat{G}^{\delta}(s)+a_k|\!|X|\!|^{\rho}_s+b_k\,,\\
		&|\mathbb{E}_s[\hat{Y}^k(s_1)-\hat{Y}^k(s_2)]\le (s_2-s_1)(\tilde{a}_k\mathbb{E}_s[|\!|X|\!|^{\rho}_{s_2}]+\tilde{b}_k)\,
	\end{align}
	for all $0\le s\le s_1\le s_2$. Let $a_0=b_0=\tilde{a}_0=\tilde{b}_0=0$ then \eqref{eqn :24} is satisfied with $\hat{Y}^0=0$. Observe that 
	\begin{align}
		&|Y^n(s)|\le L(1+|\!|X|\!|^\rho_s)\,,\\
		&|\mathbb{E}_s[Y^n(s_1)-Y^n(s_2)]|\le \Big|\mathbb{E}_s\Big[\int_{s_1}^{s_2}\Phi(u,X_u,0)-\ell Y^n(u)\,du\Big]\Big|\le L(1+\mathbb{E}_{s}[|\!|X|\!|^\rho_{s_2}])(s_2-s_1)\label{eqn: 13}
	\end{align}
	for $0\le s\le s_1\le s_2\le n$.
	Then we obtain
	\begin{align}
		|\hat{Y}^1(s)|=&\Big|\mathbb{E}_s\Big[\int_s^n\frac{e^{-\ell(u-s)}}{\delta}\int_{u-\delta}^u \Phi(v,X_v,0)-\Phi(u,X_u,0)\,dv\,du \\
		&\quad\quad-\int_s^n e^{-\ell(u-s)}\frac{\ell-r}{\delta}\int_{u-\delta}^u Y^n(v)-Y^n(u)\,dv\,du\Big]\Big|\\
		&\le \hat{G}^{\delta}(s)+ \Big|\mathbb{E}_s\Big[\int_s^{s+\delta} e^{-\ell(u-s)}\frac{|\ell-r|}{\delta}\int_{u-\delta}^u Y^n(v)-Y^n(u)\,dv\,du\Big]\Big|\\ &
		\quad+\mathbb{E}\Big[\int_{s+\delta}^n e^{-\ell(u-s)}\frac{|\ell-r|}{\delta}\int_{u-\delta}^u|\mathbb{E}_s[ Y^n(v)-Y^n(u)]|\,dv\,du\Big]\,\\
		&\le \hat{G}^{\delta}(s)+\mathbb{E}_s\Big[\int_s^{s+\delta}Le^{-\ell(u-s)}(1+|\!|X|\!|_u)\,du\Big]+\mathbb{E}_s\Big[\int_s^n L\delta e^{-\ell(u-s)}(1+|\!|X|\!|_u)\,du\Big]\\
		&\le \hat{G}^{\delta}(s)+L(1+|\!|X|\!|^\rho_s)\delta\\
		&\le \hat{G}^{\delta}
		(s)+a_1|\!|X|\!|^\rho_s+b_1
	\end{align}
	for $a_1:=L\delta$ and $b_1:=L\delta$.
	Given $a_0,b_0,\tilde{a}_0,\tilde{b}_0,a_1,b_1$, we define inductively 
	\begin{align}
		&a_{k+1}=L\delta+2|\ell-r| L_6\delta a_k+\frac{1}{2}|\ell-r|L_6\delta \tilde{a}_k\,,\\
		&b_{k+1}=L\delta+2|\ell-r|L_7\delta a_k+2|\ell-r|\delta b_k+\frac{1}{2}|\ell-r|L_7\delta \tilde{a}_k+\frac{1}{2}|\ell-r|\delta \tilde{b}_k\,,\\
		&\tilde{a}_{k+1}=L+\ell a_{k+1}+2|\ell-r| a_k\,,\\
		&\tilde{b}_{k+1}=L+\ell b_{k+1}+2|\ell-r| b_k\,.\label{eqn:30}
	\end{align}
	Observe that
	\begin{align}
		|Y^{k}(s)|&\le \hat{G}^{\delta}(s)+a_{k}|\!|X|\!|^\rho_s+b_{k}(s)\\
		&=\mathbb{E}_s\Big[\int_s^n\frac{e^{-\ell(u-s)}}{\delta}\int_{u-\delta}^u |\Phi(v,X_v,0)-\Phi(u,X_u,0)|\,dv\,du \Big]+a_{k}|\!|X|\!|^\rho_s+b_{k}\\
		&\le \mathbb{E}_s\Big[\int_s^n e^{-\ell(u-s)}2C_\Phi(1+|\!|X|\!|^\rho_u)\,du \Big]+a_{k}|\!|X|\!|^\rho_s+b_{k}\\
		&\le (a_{k}+L)|\!|X|\!|^\rho_s+b_{k}+L\,.
	\end{align}
	Applying I\^o's formula, we have
	\begin{equation} 
		\begin{aligned}
			\hat{Y}^{k+1}(s)&=\mathbb{E}_s\Big[\int_s^n \frac{e^{-\ell(u-s)}}{\delta}\int_{u-\delta}^u\Phi(v,X_v,0)-\Phi(u,X_u,0)\,dv\\
			&-\frac{(\ell-r) e^{-\ell(u-s)}}{\delta}\int_{u-\delta}^u Y^n(v)-Y^n(u)\,dv -\frac{(\ell-r) e^{-\ell(u-s)}}{\delta}\int_{u-\delta}^u \hat{Y}^k(v)-\hat{Y}^k(u)\,dv\,du \Big]\\
			&\le \hat{G}^{\delta}(s)+a_1|\!|X|\!|^\rho_s+b_1+\mathbb{E}_s\Big[\int_s^{s+\delta}\frac{|\ell-r| e^{-\ell(u-s)}}{\delta}\int_{u-\delta}^u \hat{Y}^k(v)-\hat{Y}^k(u)\,dv\,du \Big]\\
			&\quad+\mathbb{E}_s\Big[\int_{s+\delta}^{n}\frac{|\ell-r| e^{-\ell(u-s)}}{\delta}\int_{u-\delta}^u \hat{Y}^k(v)-\hat{Y}^k(u)\,dv\,du \Big]\\
			&\le\hat{G}^{\delta}(s)+a_{k+1}|\!|X|\!|^\rho_s+b_{k+1}\,.
		\end{aligned}    
	\end{equation}
	For $s\le s_1\le s_2\le n$,
	\begin{align}
		|\mathbb{E}_s[\hat{Y}^{k+1}(s_1)-\hat{Y}^{k+1}(s_2)]|&=\Big|\mathbb{E}_s\Big[\int_{s_1}^{s_2} \frac{1}{\delta}\int_{u-\delta}^u \Phi(v,X_v,0)-\Phi(u,X_u,0)\,dv\\
		&\quad\quad\quad\quad-\frac{\ell}{\delta}\int_{u-\delta}^uY^n(v)-Y^n(u)\,dv\\
		&\quad\quad\quad\quad-\ell\hat{Y}^{k+1}(u)-\frac{\ell-r}{\delta}\int_{u-\delta}^u\hat{Y}^{k}(v)-\hat{Y}^{k}(u)\,dv\,du\Big]\Big|\\
		&\le (s_2-s_1)(\tilde{a}_{k+1}\mathbb{E}_s[|\!|X|\!|^\rho_s]+\tilde{b}_{k+1})\,.
	\end{align}
	
	We now prove the equality \eqref{eqn:n delta}.
	Following the argument in \eqref{eqn:recur},
	the sequence $(a_k)_{k\ge 0}$   converges to a positive constant and the limit 
	\begin{equation} 
		a:=\lim_{k\to\infty}a_k= \frac{L\delta+\frac{1}{2}|\ell-r| L_6C_\Phi\delta}{1-L_6((\ell-r)^2+\frac{1}{2}|\ell-r|\ell+2|\ell-r|)\delta}
	\end{equation}   
	is bounded by $L\delta$ for some $L>0.$
	Similarly, the sequence $(b_k)_{k\ge 0}$  also converges to a positive constant, which is bounded by $L\delta$. 
	We recall that  $L$ is a generic constant and may differ line by line.
	Corollary \ref    {cor:banach} yields that 
	$\hat{Y}^k(s)\to\hat{Y}(s)=Y^{n,\delta}(s)-Y^n(s)$
	as  $k\to\infty.$
	Thus, by \eqref{eqn :24} we have    \begin{align}
		|Y^{n,\delta}(s)-Y^n(s)|\le \hat{G}^{\delta}(s)+L(1+|\!|X|\!|_s^\rho)\delta\,.\label{eqn:14}
	\end{align}
	By Doob's inequality and Jensen's inequality, we have
	\begin{align}
		&\quad\; (\mathbb{E}[(\sup_{0\le s\le T} 
		\hat{G}^{\delta}(s))^2])^{1+r}\\
		&=\Big(\mathbb{E}\Big[\Big(\sup_{0\le s\le T}\mathbb{E}_s\Big[\int_s^n\frac{e^{-\ell(u-s)}}{\delta}\int_{u-\delta}^u |\Phi(v,X_v,0)-\Phi(u,X_u,0)|\,dv\,du \Big]\Big)^2\Big]\Big)^{1+r}\\
		&\le e^{2\ell(1+r)T}\mathbb{E}\Big[\Big( \sup_{0\le s\le T}\mathbb{E}_s\Big[\int_0^n \frac{e^{-\ell u}}{\delta}\int_{u-\delta}^u |\Phi(v,X_v,0)-\Phi(u,X_u,0)|\,dv\,du\Big]\Big)^{2+2r} \Big]\\
		&\le e^{2\ell(1+r)T}(\frac{2+2r}{1+2r})^{2+2r}\sup_{0\le s\le T}\mathbb{E}\Big[\Big(\mathbb{E}_s\Big[\int_0^n \frac{e^{-\ell u}}{\delta}\int_{u-\delta}^u |\Phi(v,X_v,0)-\Phi(u,X_u,0)|\,dv\,du\Big]\Big)^{2+2r}\Big]\\
		&\le e^{2\ell(1+r)T}(\frac{2+2r}{1+2r})^{2+2r}\mathbb{E}\Big[\Big(\int_0^n \frac{e^{-\ell u}}{\delta}\int_{u-\delta}^u |\Phi(v,X_v,0)-\Phi(u,X_u,0)|\,dv\,du\Big)^{2+2r}\Big]\,
	\end{align}
	for some constant $r>0$. 
	Applying the dominated convergence theorem and the Lebesgue differentiation theorem to the right-hand side of the above inequality, \begin{align}
		\lim_{\delta\to 0}\mathbb{E}[(\sup_{0\le s\le T}\hat{G}^{\delta}(s))^2]=0\,.\label{eqn:13}   
	\end{align}
	Therefore by \eqref{eqn:14} and \eqref{eqn:13}, we obtain the  first equality in \eqref{eqn:estimate}.

	For the second equality in
	\eqref{eqn:estimate}, applying It\^o's formula to $(Y^\delta-Y)^2$, we have
	\begin{align}\label{eqn:Zesi}
		&\quad \int_0^T|Z^{\delta}(u)-Z(u)|^2\,du\\
		&\le L\Big( |Y^\delta(T)-Y(T)|^2+\int_0^T\frac{2|Y^{\delta}(u)-Y(u)|}{\delta}\int_{u-\delta}^u|\Phi(u,X_u,0)-\Phi(v,X_v,0)|\,dv\,du\\
		&\quad+\int_0^T\frac{2|Y^{\delta}(u)-Y(u)|}{\delta}\Big(\int_{u-\delta}^u|Y(u)-Y(v)|\,dv+\int_{u-\delta}^u|Y^\delta(v)-Y(v)|\,dv\,du\Big)\Big)\,.
	\end{align}
	Recall from Theorem \ref{thm: BSDE no delay} and \eqref{thm:delay_2} in Theorem \ref{thm:approx} that  $|Y(s)|\le L(1+|\!|X|\!|^\rho_s)$  and $|Y^\delta(s)|\le L(1+|\!|X|\!|_s^\rho)$. 
	Taking the expectation of both sides of \eqref{eqn:Zesi}, we obtain
	\begin{align}\label{eqn:Ez}
		&\quad\mathbb{E}\Big[\int_0^T|Z^{\delta}(u)-Z(u)|^2\,du\Big]\\
		&\le L\Big(\mathbb{E}[ |Y^\delta(T)-Y(T)|^2]+\mathbb{E}\Big[\int_0^T\frac{2|Y^{\delta}(u)-Y(u)|}{\delta}\int_{u-\delta}^u|\Phi(u,X_u,0)-\Phi(v,X_v,0)|\,dv\,du\Big]    \\
		&\quad+\mathbb{E}\Big[\int_0^T\frac{2|Y^{\delta}(u)-Y(u)|}{\delta}\Big(\int_{u-\delta}^u|Y(u)-Y(v)|\,dv+\int_{u-\delta}^u|Y^\delta(v)-Y(v)|\,dv\,du\Big)\Big)\Big]\Big)\\
		& \le L(1+T)(\mathbb{E}[|\!|Y^\delta-Y|\!|^2_T])^{\frac{1}{2}}(\mathbb{E}[1+|\!|X|\!|_T^{2\rho}])^{\frac{1}{2}}\,.
	\end{align}
	From \eqref{eqn:delta} and \eqref{eqn:n delta}, we have
	\begin{align}
		\lim_{\delta\to 0}\mathbb{E}[|\!|Y^\delta-Y|\!|_T^2]=0\,.
	\end{align}
	Therefore by letting $\delta \to 0$ to \eqref{eqn:Ez}, we obtain the second equality in 
	\eqref{eqn:estimate}\,.

\end{proof}

Finally we prove \eqref{eqn:esti} in Theorem \ref{thm:approx}.

\begin{proof} 
	Recall that $L$ denotes a generic constant that depends only on $C_1,C_3, C_\Phi, r,\ell,\rho$ and may differ line by line, and  $L_6,L_7,L_8$ are constants satisfying \eqref{eqn:posi}. 
	We first consider the case with $\rho>1.$
	To prove 
	the first inequality in   
	\eqref{eqn:esti}, we recall that $Y^n$ and $Y^{n,\delta}$ are solutions to the finite-horizon BSDEs \eqref{eqn:Y^n}.
	Consider 
	the
	Picard iteration 
	$(\hat{Y}^k,\hat{Z^k})$ in \eqref{eqn:Picard} for   $(\hat{Y}(s),\hat{Z}(s)):=(Y^{n,\delta}(s)-Y^n(s),Z^{n,\delta}(s)-Z^n(s))$.
	Corollary \ref    {cor:banach} yields that 
	$\hat{Y}^k(s)\to\hat{Y}(s)=Y^{n,\delta}(s)-Y^n(s)$
	as  $k\to\infty.$

	We construct four sequences $(a_k)_{k\ge0}$, $(b_k)_{k\ge0}$, $(\tilde{a}_k)_{k\ge0}$, $(\tilde{b}_k)_{k\ge0}$ that satisfy 
	\begin{align}\label{eqn:26}
		&|\hat{Y}^k(s)|\le a_k|\!|X|\!|^{\rho}_s+b_k\,,\\
		&|\mathbb{E}_s[\hat{Y}^k(s_1)-\hat{Y}^k(s_2)]\le (s_2-s_1)(\tilde{a}_k\mathbb{E}_s[|\!|X|\!|^{\rho}_{s_2}]+\tilde{b}_k)\,
	\end{align}
	for all $0\le s\le s_1\le s_2$. 
	Let $a_0=b_0=\tilde{a}_0=\tilde{b}_0=0$ then \eqref{eqn:26} is satisfied with $\hat{Y}^0=0$. Now we define $a_1$ and $b_1$.
	A direct calculation and  \eqref{eqn:posi} yield 
	\begin{align}
		(\mathbb{E}_s[|\!|X-X_{s_1}|\!|_{s_2}^\rho])^{\frac{1}{\rho}}&=(\mathbb{E}_s[\sup_{s_1\le r\le s_2}|X(r)-X(s_2)|^\rho])^{\frac{1}{\rho}}\\
		&\le C_1M_\rho(s_2-s_1)^{\frac{1}{2}}+(r(s_2-s_1)+C_3M_\rho(s_2-s_1)^{\frac{1}{2}})(\mathbb{E}_s[|\!|X|\!|_{s_2}^\rho])^{\frac{1}{\rho}}\label{eqn:pro1}\,,\\
		|Y^n(s)|&\le \mathbb{E}_s\Big[\int_s^n e^{-\ell(u-s)}|\Phi(u,X_u,0)|\,du\Big]\\
		&\le\mathbb{E}\Big[\int_s^ne^{-\ell(u-s)}C_\Phi(|\!|X|\!|^\rho_u+1)\,du\Big]\\
		&\le \int_s^ne^{-(\ell-L_8)(u-s)}C_\Phi(L_6|\!|X|\!|_s+L_7)+\int_s^n e^{-\ell(u-s)}C_\Phi\,du\\
		&\le\frac{C_\Phi L_6}{\ell-L_8}|\!|X|\!|_s^\rho+\frac{C_\Phi L_7}{\ell-L_8}+\frac{C_\Phi}{\ell}\label{eqn:pro2}
	\end{align}
	and
	\begin{align}
		|\mathbb{E}_s[Y^n(s_1)-Y^n(s_2)]|&\le \mathbb{E}_s\Big[\int_{s_1}^{s_2}|\Phi(u,X_u,0)|+\ell|Y^n(u)|\,du\Big]\\
		&\le (s_2-s_1)\Big(\Big(C_\Phi+\frac{\ell C_\Phi L_6}{\ell-L_8}\Big)|\!|X|\!|_s^\rho +2C_\Phi+\frac{\ell C_\Phi L_7}{\ell-L_8}\Big)\label{eqn:pro3}\,.
	\end{align} Then, using  the above inequalities and It\^o's formula, we obtain 
	\begin{equation}\label{eqn:pro5}
		\begin{aligned}
			|\hat{Y}^1(s)|&\le \mathbb{E}_s\Big[\int_s^{s+\delta} 
			\frac{e^{-\ell(u-s)}}{\delta}\int_{u-\delta}^u|\Phi(v,X_v,0)-\Phi(u,X_u,0)|\,dv\,du\Big]\\
			&\quad+\mathbb{E}_s\Big[\int_{s+\delta}^n 
			\frac{e^{-\ell(u-s)}}{\delta}\int_{u-\delta}^u|\Phi(v,X_v,0)-\Phi(u,X_u,0)|\,dv\,du\Big]\\
			&\quad+\Big|\mathbb{E}_s\Big[\int_s^{s+\delta} 
			\frac{(\ell-r) e^{-\ell(u-s)}}{\delta}\int_{u-\delta}^u Y^n(v)-Y^n(u)\,dv  \Big]\Big|\\
			&\quad+\Big|\mathbb{E}_s\Big[\int_{s+\delta}^n 
			\frac{(\ell-r) e^{-\ell(u-s)}}{\delta}\int_{u-\delta}^u Y^n(v)-Y^n(u)\,dv  \Big]\Big|\,\\
			&\le a(\rho)\sqrt{\delta} |\!|X|\!|_s+b(\rho) \sqrt{\delta}
		\end{aligned}    
	\end{equation}
	where 
	\begin{align}
		a(\rho)&=
		2C_\Phi L_6\sqrt{\delta}+\frac{L_6C_5}{\ell-L_8}\Big( \frac{3r\sqrt{\delta}}{2}+2M_\rho C_3+\frac{4(1+M_\rho C_1)}{3}\Big)\\
		&\quad +\frac{|\ell-r|C_\Phi L_6\sqrt{\delta}}{\ell-L_8}\Big(2L_6+\frac{\ell-L_8+\ell L_6}{2(\ell-L_8)} \Big)\,,\\ 
		b(\rho)&= C_\Phi\sqrt{\delta} \Big(2\Big(L_7+1+\frac{|\ell-r|L_7(1+L_6)}{\ell-L_8}+\frac{|\ell-r|}{\ell}\Big) +|\ell-r|\Big(\frac{L_7}{\ell-L_8}+\frac{\ell L_6 L_7}{2(\ell-L_8)^2}+\frac{1}{\ell}\Big)\Big)\\
		&\quad+ C_5\Big( 2(1+M_\rho C_1) +\frac{r\sqrt{\delta}}{2}+\frac{2M_\rho C_3(\ell-L_8+L_7)}{3(\ell-L_8)}+\frac{L_7(9r\sqrt{\delta}+8(1+M\rho C_1))}{6(\ell-L_8)}\Big)\,.
	\end{align}
	Thus, $a_1=a(\rho)\sqrt{\delta}$ and $b_1=b(\rho)\sqrt{\delta}$ yields the desired ones.
	Given $a_0,b_0,\tilde{a}_0,\tilde{b}_0,a_1,b_1$, we define inductively
	\begin{equation}  \begin{aligned}
			&a_{k+1}=a(\rho)\sqrt{\delta}+2|\ell-r| L_6\delta a_k+\frac{1}{2}|\ell-r| L_6\delta \tilde{a}_k\,,\\
			&b_{k+1}=b(\rho)\sqrt{\delta}+2|\ell-r| L_7\delta a_k+2|\ell-r| \delta b_k+\frac{1}{2}|\ell-r| L_7\delta \tilde{a}_k+\frac{1}{2}|\ell-r|\delta \tilde{b}_k\,,\\
			&\tilde{a}_{k+1}=2C_\Phi+\frac{2|\ell-r|C_{\Phi}L_6}{\ell-L_8}+\ell a_{k+1}+2|\ell-r| a_k\,,\\ 
			&\tilde{b}_{k+1}=2C_\Phi+ \frac{2|\ell-r|C_\Phi L_7}{\ell-L_8}+\frac{2|\ell-r|C_\Phi}{\ell}    +\ell b_{k+1}+2|\ell-r| b_k\,.
		\end{aligned} 
	\end{equation}
	This construction yields the desired four sequences $(a_k)_{k\ge0}$, $(b_k)_{k\ge0}$, $(\tilde{a}_k)_{k\ge0}$ and $(\tilde{b}_k)_{k\ge0}$.

	Following the argument in \eqref{eqn:recur}, the sequences 
	$(a_k)_{k\ge0}$ and $(b_k)_{k\ge0}$  converge to positive constants.
	A direct calculation yields
	\begin{equation} 
		\begin{aligned}
			a_{k+1}=a(\rho)\sqrt{\delta}+|\ell-r|L_6
			\delta C_\Phi(1+\frac{|\ell-r|L_6}{\ell-L_8})+L_6(\frac{1}{2}|\ell-r|\ell+2|\ell-r|)\delta a_k+L_6(\ell-r)^2\delta a_{k-1} 
		\end{aligned}    
	\end{equation}
	and this recursive relation implies that $  \lim_{k\to\infty}a_k=L_2\sqrt{\delta}$ where
	\begin{align} 
		L_2:= \frac{a(\rho) +|\ell-r|L_6
			C_\Phi(1+\frac{|\ell-r|L_6}{\ell-L_8})\sqrt{\delta}}{1-L_6\delta((\ell-r)^2+\frac{|\ell-r|\ell}{2}+2|\ell-r|)}\,.
	\end{align}     
	Similarly, the sequence $(b_k)_{k\ge0}$ converges to $L_1\sqrt{\delta}$ where $L_1$ is a constant defined as
	\begin{equation} 
		\begin{aligned}
			L_1=\frac{b(\rho)+a(\rho)|\ell-r|L_7\delta(2+\frac{\ell}{2}+|\ell-r|)+|\ell-r|C_\Phi  (L_7+1+\frac{L_6L_7|\ell-r|}{\ell-L_8}+\frac{L_7|\ell-r|}{\ell-L_8}+\frac{|\ell-r|}{\ell})\sqrt{\delta}}{1-\delta\big((\ell-r)^2+\frac{|\ell-r|\ell}{2}+2|\ell-r|\big)}.
		\end{aligned}    
	\end{equation}
	Since $L_6((\ell-r)^2+\frac{1}{2}|\ell-r|\ell+2|\ell-r|)\delta<1$, the constants $L_1$ and $L_2$ are positive.    
	Letting $k \to \infty$ in \eqref{eqn:26}, we obtain
	$$|Y^{n,\delta}(s)-Y^n(s)|=|\hat{Y}(s)|\le (L_1 +L_2|\!|X|\!|^{\rho}_s)\sqrt{\delta}\,.$$
	Similar to \eqref{eqn:Y}, letting $n \to \infty$, we deduce    
	$|Y^{\delta}(s)-Y(s)|\le (L_1 +L_2|\!|X|\!|^{\rho}_s)\sqrt{\delta}$.  
	This yields 
	\begin{equation}
		\label{eqn:Y_del}
		|\!|Y^{\delta}-Y|\!|_T\le (L_1+L_2|\!|X|\!|_T^\rho)\sqrt{\delta}\,,
	\end{equation}
	which corresponds to the first inequality in   
	\eqref{eqn:esti}.
	From \eqref{eqn:Zesi} and \eqref{eqn:pro2},
	we obtain  
	\begin{equation}
		\label{eqn:estiZ} 
		\mathbb{E}\Big[\int_0^T|Z^\delta(u)-Z(u)|^2\,du\Big]\le (L_3(T)\mathbb{E}[|\!|X|\!|_T^{2\rho}]+L_4(T)\mathbb{E}[|\!|X|\!|_T^{\rho}]+L_5(T))\sqrt{\delta}
	\end{equation}  
	where
	\begin{equation}
		\begin{aligned}
			{L}_3(T)&=\sqrt{\delta}L_2^2+2L_2\Big( 2C_\Phi+|\ell-r|\Big(\frac{2C_\Phi L_6}{\ell-L_8}+\sqrt{\delta}L_2\Big)\Big)T\,,\\
			{L}_4(T)&=2\sqrt{\delta}L_1L_2(1+2|\ell-r|T)\\
			&\quad +2\Big(2C_\Phi(L_2+L_1)+\frac{2|\ell-r|C_\Phi(L_2L_7+L_1L_6)}{\ell-L_8}+\frac{L_2C_\Phi}{\ell}\Big)T\,,\\
			{L}_5(T)&=\sqrt{\delta}L_1^2(1+2|\ell-r|T)+4L_1C_\Phi\Big(1+|\ell-r|\Big(\frac{L_7}{\ell-L_8}+\frac{1}{\ell}\Big) \Big)T\,.
		\end{aligned}
	\end{equation}
	This yields more refined upper bounds 
	than the second inequality in \eqref{eqn:esti}.

	The case with $\rho=1$ can be proven similarly with 
	$a(\rho)$  and $b(\rho)$ replaced by $a(1)$ and $b(1)$, respectively,
	where  
	\begin{equation} 
		\begin{aligned}
			a(1)&=2C_\Phi L_6\sqrt{\delta} +\frac{L_6C_5}{\ell-L_8}\Big(\frac{r\sqrt{\delta}}{2}+\frac{2M_1C_3}{3}\Big)+\frac{|\ell-r|C_\Phi L_6\sqrt{\delta}}{\ell-L_8}\Big(2L_6+\frac{\ell-L_8+\ell L_6}{2(\ell-L_8)} \Big)\,,\\
			b(1)&= C_\Phi\sqrt{\delta} \Big(2\Big(L_7+1+\frac{|\ell-r|L_7(1+L_6)}{\ell-L_8}+\frac{|\ell-r|}{\ell}\Big)+|\ell-r|\Big(\frac{L_7}{\ell-L_8}+\frac{\ell L_6 L_7}{2(\ell-L_8)^2}+\frac{1}{\ell}\Big)\Big)\\
			&\quad\quad + C_5 \Big( \frac{rL_7\sqrt{\delta}}{2(\ell-L_8)}+\frac{2L_7M_1C_3}{3(\ell-L_8)}+\frac{2(1+M_1C_1)}{3\ell}\Big)\,.
		\end{aligned} 
	\end{equation}
	The same inequalities \eqref{eqn:Y_del} and
	\eqref{eqn:estiZ} hold with the constants $L_1,L_2,L_3(T),L_4(T),L_5(T),$ where $a(\rho)$ and $b(\rho)$ are replaced by $a(1)$ and $b(1)$, respectively.

\end{proof}







\bibliographystyle{apalike}

\bibliography{references}

\begin{thebibliography}{}

\bibitem[Ackerer et~al., 2024]{ackerer2024perpetual}
Ackerer, D., Hugonnier, J., and Jermann, U. (2024).
\newblock Perpetual futures pricing.
\newblock Technical report, National Bureau of Economic Research.

\bibitem[Alexander et~al., 2020]{alexander2020bitmex}
Alexander, C., Choi, J., Park, H., and Sohn, S. (2020).
\newblock {BitMEX} bitcoin derivatives: {P}rice discovery, informational
  efficiency, and hedging effectiveness.
\newblock {\em Journal of Futures Markets}, 40(1):23--43.

\bibitem[Angeris and Chitra, 2020]{angeris2020improved}
Angeris, G. and Chitra, T. (2020).
\newblock Improved price oracles: {C}onstant function market makers.
\newblock In {\em Proceedings of the 2nd ACM Conference on Advances in
  Financial Technologies}, pages 80--91.

\bibitem[Angeris et~al., 2023]{angeris2023primer}
Angeris, G., Chitra, T., Evans, A., and Lorig, M. (2023).
\newblock A primer on perpetuals.
\newblock {\em SIAM Journal on Financial Mathematics}, 14(1):SC17--SC30.

\bibitem[Bally et~al., 2016]{bally2016stochastic}
Bally, V., Caramellino, L., Cont, R., Utzet, F., and Vives, J. (2016).
\newblock {\em Stochastic integration by parts and functional {I}t{\^o}
  calculus}.
\newblock Springer.

\bibitem[Butler and Rogers, 1971]{butler1971generalization}
Butler, G. and Rogers, T. (1971).
\newblock A generalization of a lemma of {B}ihari and applications to pointwise
  estimates for integral equations.
\newblock {\em J. Math. Anal. Appl}, 33(1):77--81.

\bibitem[Christin et~al., 2022]{christin2022crypto}
Christin, N., Routledge, B., Soska, K., and Zetlin-Jones, A. (2022).
\newblock The crypto carry trade.
\newblock {\em Preprint at http://gerbil. life/papers/CarryTrade. v1}, 2.

\bibitem[Confortola et~al., 2019]{confortola2019backward}
Confortola, F., Cosso, A., and Fuhrman, M. (2019).
\newblock Backward {SDE}s and infinite horizon stochastic optimal control.
\newblock {\em ESAIM: Control, Optimisation and Calculus of Variations}, 25:31.

\bibitem[Cont and Fourni{\'e}, 2013]{cont2013functional}
Cont, R. and Fourni{\'e}, D.-A. (2013).
\newblock Functional {I}t{\^o} calculus and stochastic integral representation
  of martingales.
\newblock {\em The Annals of Probability}, pages 109--133.

\bibitem[Corbet et~al., 2021]{corbet2021volatility}
Corbet, S., Hou, Y.~G., Hu, Y., and Oxley, L. (2021).
\newblock Volatility spillovers during market supply shocks: {T}he case of
  negative oil prices.
\newblock {\em Resources Policy}, 74:102357.

\bibitem[Cordoni et~al., 2020]{cordoni2020stochastic}
Cordoni, F., Di~Persio, L., Maticiuc, L., and Z{\u{a}}linescu, A. (2020).
\newblock A stochastic approach to path-dependent nonlinear {K}olmogorov
  equations via {BSDE}s with time-delayed generators and applications to
  finance.
\newblock {\em Stochastic Processes and their Applications}, 130(3):1669--1712.

\bibitem[Dai et~al., 2025]{dai2025arbitrage}
Dai, M., Li, L., and Yang, C. (2025).
\newblock Arbitrage in perpetual contracts.
\newblock {\em Available at SSRN 5262988}.

\bibitem[Delbaen and Schachermayer, 1994]{delbaen1994general}
Delbaen, F. and Schachermayer, W. (1994).
\newblock A general version of the fundamental theorem of asset pricing.
\newblock {\em Mathematische annalen}, 300(1):463--520.

\bibitem[Deuschel and Stroock, 1989]{deuschel1989large}
Deuschel, J.-D. and Stroock, D.~W. (1989).
\newblock {\em Large Deviations}, volume 137.
\newblock Academic Press.

\bibitem[Dupire, 2019]{dupire2019functional}
Dupire, B. (2019).
\newblock Functional {I}t{\^o} calculus.
\newblock {\em Quantitative Finance}, 19(5):721--729.

\bibitem[Ekren et~al., 2014]{ekren2014viscosity}
Ekren, I., Keller, C., Touzi, N., and Zhang, J. (2014).
\newblock On viscosity solutions of path dependent {PDE}s.
\newblock {\em The Annals of Probability}, 42(1):204--236.

\bibitem[Evans, 2020]{evans2020liquidity}
Evans, A. (2020).
\newblock Liquidity provider returns in geometric mean markets.
\newblock {\em arXiv preprint arXiv:2006.08806}.

\bibitem[He et~al., 2022]{he2022fundamentals}
He, S., Manela, A., Ross, O., and von Wachter, V. (2022).
\newblock Fundamentals of perpetual futures.
\newblock {\em arXiv:2212.06888}.

\bibitem[Protter, 2005]{protter}
Protter, P.~E. (2005).
\newblock {\em Stochastic integration and differential equations}, volume~21 of
  {\em Stochastic Modelling and Applied Probability}.
\newblock Springer-Verlag, Berlin, second edition.
\newblock Corrected third printing.

\bibitem[Ruan and Streltsov, 2024]{ruan2024perpetual}
Ruan, Q. and Streltsov, A. (2024).
\newblock Perpetual futures contracts and cryptocurrency market quality.
\newblock {\em SSRN:4218907}.

\bibitem[Viens and Zhang, 2019]{viens2019martingale}
Viens, F. and Zhang, J. (2019).
\newblock A martingale approach for fractional {B}rownian motions and related
  path dependent {PDEs}.
\newblock {\em The Annals of Applied Probability}, 29(6):3489--3540.

\bibitem[Wang and Zhang, 2025]{wang2025spot}
Wang, S. and Zhang, T. (2025).
\newblock Spot-futures manipulations in cryptocurrency markets.
\newblock {\em SSRN: 5125326}.

\bibitem[Zhang, 2017]{zhang2017backward}
Zhang, J. (2017).
\newblock {\em Backward Stochastic Differential Equations: {F}rom Linear to
  Fully Nonlinear Theory}, volume~86.
\newblock Springer.

\end{thebibliography}

\end{document}